\documentclass{report}

\usepackage{a4}

\usepackage{xr}

\scrollmode
\usepackage{amsmath}
\usepackage{amsfonts}
\usepackage{amssymb}
\usepackage{latexsym}
\usepackage{stmaryrd}
\usepackage{array}
\usepackage{exscale}
\newcommand{\nc}{\newcommand}
\newcommand{\ol}{\overline}
\newcommand{\ul}{\underline}
\newcommand{\es}{\emptyset}
\newcommand{\sm}{\setminus}
\newcommand{\ve}{\varepsilon}
\newcommand{\vp}{\varphi}

\newcommand{\bc}{\bigcup}
\newcommand{\bca}{\bigcap}
\newcommand{\Lra}{\Leftrightarrow}
\newcommand{\Lora}{\Longrightarrow}

\newcommand{\Ra}{\Rightarrow}

\newcommand{\ra}{\rightarrow}

\newcommand{\sse}{\subseteq}

\newcommand{\fa}{\forall}
\newcommand{\ex}{\exists}
\newcommand{\mr}{\mathrm}
\newcommand{\mc}{\mathcal}
\newcommand{\mf}{\mathfrak}

\newcommand{\DMO}{\DeclareMathOperator}
\newcommand{\DST}{\displaystyle}

\newcommand{\ZZ}{\mathbb{Z}}
\newcommand{\NN}{\mathbb{N}}
\newcommand{\NNZ}{\NN_0}

\newcommand{\RR}{\mathbb{R}}

\newcommand{\PP}{\mathbb{P}}

\newcommand{\EE}{\mathbb{E}}
\newcommand{\FF}{\mathbb{F}}

\newcommand{\Ende}{\ \rule{0.4em}{1.7ex}}

\newcommand{\NAA}{\setlength{\itemsep}{0pt}} 
\newcommand{\aru}{\ar @{-}} 
\usepackage{listings}
\newcommand{\inl}[1]{\lstinline$#1$}
\newcommand{\und}{{\:\wedge\:}} 
\newcommand{\oder}{{\:\vee\:}} 
\newcommand{\mb}{{\:|\:}} 
\newcommand{\set}[1]{\{ #1 \}}
\newcommand{\setb}[1]{\big \{ \, #1 \, \big \}}
\DMO{\dom}{dom}
\DMO{\id}{id}
\DMO{\cod}{cod} 
\DMO{\rg}{rg} 
\DMO{\tcomp}{\trans{\circ}} 
\DMO{\simrv}{\,\sim\hspace{-0.05em}}
\DMO{\simlv}{\!\sim\,}
\nc{\simlvi}[1]{\!\sim_{#1}}
\DMO{\card}{card}
\DMO{\proj}{pr}
\DMO{\inj}{in}
\DMO{\symdif}{\vartriangle} 
\DMO{\addcup}{{\stackrel{\text{\raisebox{-2.2ex}[-0ex][-0ex]{\large$\cdot$}}}{\cup}}} 
\DMO{\addbcup}{{\stackrel{\text{\raisebox{-4.2ex}[-0ex][-0ex]{\Large$\cdot$}}}{\bigcup}}} 
\nc{\apprel}[3]{{#1}(#2)_{(#3)}} 
\DMO{\Rel}{\mf{REL}} 
\DMO{\Abb}{\mf{MAP}} 
\DMO{\Tra}{\mf{T}} 
\DMO{\Per}{\mf{S}} 
\DMO{\Pert}{\Per_t} 
\DMO{\Ptr}{\mf{PT}} 
\DMO{\fix}{fix} 
\DMO{\Peri}{\Per_i} 
\DMO{\Pers}{\Per_s} 
\DMO{\Rrel}{\Rel_r} 
\DMO{\Srel}{\Rel_s} 
\DMO{\Trel}{\Rel_t} 
\DMO{\konkat}{\sqcup} 
\DMO{\cmpl}{\complement^1} 
\nc{\cmpli}[1]{\complement^1_{#1}} 
\DMO{\cmplz}{\complement^0} 
\nc{\cmplzi}[1]{\complement^0_{#1}} 
\DMO{\cmplzo}{\complement^*} 
\nc{\cmplzoi}[1]{\complement^*_{#1}} 
\DMO{\fsigma}{{\mf{F}}_{\sigma}} 
\DMO{\gdeltao}{\mf{G}_{\sigma}}
\DMO{\fs}{{\mf{F}}_{s}} 
\DMO{\fss}{{\mf{F}}_{s}^*} 
\nc{\zf}{\mr{ZF}}
\nc{\zfmf}{\zf^0} 
\nc{\zfc}{\mr{ZFC}}
\nc{\zfcmf}{\zfc^0} 
\nc{\bst}{\mr{BST}} 
\newcommand{\tb}[2]{\set{#1, \dots, #2}} 
\DMO{\re}{Re}
\DMO{\im}{Im}
\DMO{\sgn}{sgn} 
\providecommand{\abs}[1]{\lvert #1 \rvert} 
\DMO{\ld}{ld} 
\DMO{\nachbarn}{\Gamma}
\DMO{\enachbarn}{N}
\DMO{\nachbarnr}{\Gamma_{\!\mr{r}}}
\DMO{\nachbarnz}{\widetilde{\Gamma}}
\DMO{\nachbarnzr}{\widetilde{\Gamma}_{\!\mr{r}}}
\DMO{\inzEK}{\mc{I}^{\mr{V}}}
\DMO{\inzEKe}{\mc{I}^{\mr{V}}_1}
\DMO{\inzEKz}{\mc{I}^{\mr{V}}_2}
\nc{\inzEKi}[1]{\mc{I}^{\mr{V}}_{#1}}
\DMO{\inzKE}{\mc{I}^{\mr{E}}}
\DMO{\inzKEe}{\mc{I}^{\mr{E}}_1}
\DMO{\inzKEz}{\mc{I}^{\mr{E}}_2}
\nc{\inzKEi}[1]{\mc{I}^{\mr{E}}_{#1}}
\DMO{\inz}{I}
\DMO{\tinz}{\trans{\inz}} 
\DMO{\adjE}{\mc{A}^{\mr{V}}}
\DMO{\adjEe}{\mc{A}^{\mr{V}}_1}
\DMO{\adjEz}{\mc{A}^{\mr{V}}_2}
\nc{\adjEi}[1]{\mc{A}^{\mr{V}}_{#1}}
\DMO{\adjor}{\mc{A}_{\mr{S}}} 
\DMO{\adjK}{\mc{A}^{\mr{E}}}
\DMO{\adj}{A}
\DMO{\degmin}{\mu\!\deg}
\DMO{\degmax}{\nu\!\deg}
\DMO{\degdur}{\widetilde{\deg}} 
\DMO{\ideg}{idg} 
\DMO{\odeg}{odg} 
\DMO{\degmaxl}{\nu\!\deg_{<}}
\DMO{\degl}{\deg_{<}}
\DMO{\rankmin}{\mu\!\rank}
\DMO{\rankmax}{\nu\!\rank}
\DMO{\rankdur}{\widetilde{\rank}}
\DMO{\rankmaxl}{\nu\!\rank_{<}}
\DMO{\rankl}{\rank_{<}}
\DMO{\vertexcon}{\kappa} 
\DMO{\edgecon}{\lambda} 
\DMO{\treewidth}{tw} 
\DMO{\girth}{g} 
\DMO{\circumference}{cf} 
\DMO{\length}{lgth} 
\DMO{\npm}{\Phi} 
\DMO{\concomp}{cc} 
\DMO{\nconcomp}{ncc} 
\DMO{\indprim}{ip} 
\DMO{\indimprim}{iip} 
\DMO{\bouquet}{B}
\DMO{\dipol}{D}
\DMO{\jkg}{J} 
\DMO{\vjkg}{VK} 
\DMO{\Tr}{Tr} 
\DMO{\Ind}{Ind} 
\DMO{\Zuo}{Mat} 
\DMO{\Pzuo}{PMat} 
\DMO{\St}{St} 
\DMO{\Ints}{Ints} 
\DMO{\Cov}{Cov} 
\DMO{\closse}{clo_{\sse}} 
\DMO{\clospe}{clo_{\supseteq}} 
\DMO{\edgemg}{ML} 
\DMO{\kneserg}{K} 
\DMO{\knesern}{\tau_0} 
\DMO{\nis}{nis} 
\DMO{\PBD}{PBD}
\nc{\BD}[1]{{#1}\text{-}\mr{BD}}
\DMO{\BIBD}{BIBD}
\DMO{\Steiner}{S}
\DMO{\SteinerTriple}{STS}
\DMO{\SteinerQuadruple}{SQS}
\DMO{\progeo}{PG} 
\DMO{\affgeo}{AG} 

\DMO{\astriv}{A_t} 
\DMO{\KochenSpecker}{KS}
\DMO{\KochenSpeckerErw}{KS'}
\DMO{\rankd}{rd}
\DMO{\mnconcomp}{mncc}
\DMO{\gpk}{\Box} 
\DMO{\gpw}{\times} 
\DMO{\gps}{\boxtimes} 
\DMO{\gjoin}{\boxdot} 
\DMO{\gjoinplus}{\boxplus} 
\DMO{\Ketten}{\mc{L}}
\DMO{\Antiketten}{\mc{A}}
\DMO{\comparable}{\Bumpeq}
\DMO{\incomparable}{\parallel}
\DMO{\pot}{\PP} 
\DMO{\pote}{\PP_f} 
\DMO{\potfv}{\overrightarrow{\PP}} 
\DMO{\potfvn}{\overrightarrow{\PP}^{\!*}} 
\DMO{\potfr}{\overleftarrow{\PP}} 
\DMO{\fak}{fac}
\newcommand{\ueber}[2]{\genfrac{}{}{0pt}{}{#1}{#2}}
\DMO{\partitiont}{p}
\DMO{\teilt}{\mid} 
\DMO{\nteilt}{\nmid} 
\nc{\Prim}{\mc{PR}} 
\DMO{\ord}{ord}
\DMO{\ggt}{ggt}
\DMO{\kgv}{kgv}
\DMO{\opmod}{mod}
\DMO{\opdiv}{div}
\DMO{\eulphi}{\vp}
\DMO{\Li}{Li}
\DMO{\Ei}{Ei}
\newcommand{\Va}{\mc{V\hspace{-0.1em}A}}
\newcommand{\Dom}{\mc{DO\hspace{-0.08em}M}}
\newcommand{\Lit}{\mc{LIT}}
\newcommand{\Cl}{\mc{CL}}
\newcommand{\Cls}{\mc{CLS}}

\newcommand{\Pass}{\mc{P\hspace{-0.32em}ASS}}
\newcommand{\epa}{\pab{}} 

\newcommand{\Sat}{\mc{SAT}}

\newcommand{\Usat}{\mc{USAT}}

\newcommand{\Musat}{\mc{M\hspace{0.8pt}U}} 
\newcommand{\Musati}[1]{\Musat_{\!#1}} 
\newcommand{\Smusat}{\mc{S}\Musat} 
\newcommand{\Smusati}[1]{\Smusat_{\!#1}}
\newcommand{\Mmusat}{\mc{M}\Musat}
\newcommand{\Mmusati}[1]{\Mmusat_{\!#1}}

\nc{\Clsoo}{\Cls^{1,1}} 
\newcommand{\Clash}{\mc{HIT}} 
\newcommand{\Clashi}[1]{\Clash_{\!\!#1}}

\newcommand{\Sclash}{\mc{R}\Clash} 
\newcommand{\Sclashk}[1]{\Sclash_{\! #1}}
\newcommand{\Mclash}{\mc{M}\Clash} 

\DMO{\munpuclash}{\mu{}NH}
\DMO{\hdef}{\delta_{\mr{h}}} 
\DMO{\rdef}{\delta_{\mr{r}}} 

\newcommand{\Lean}{\mc{LEAN}}

\newcommand{\Mlean}{\mc{M}\Lean}
\newcommand{\Mleani}[1]{\Mlean_{\!#1}}
\newcommand{\Msat}{\mc{M}\Sat}

\newcommand{\Mcls}{\mc{M}\Cls} 

\DMO{\nulli}{null} 
\DMO{\lit}{lit}
\DMO{\var}{var}
\DMO{\val}{val}
\DMO{\res}{\diamond} 
\DMO{\resop}{Res} 
\DMO{\mresop}{mRes} 
\DMO{\dpl}{DP} 
\newcommand{\dpi}[1]{\dpl_{\!#1}}
\DMO{\comp}{Comp} 
\DMO{\compex}{\comp_{ER}} 
\DMO{\compr}{\comp_R} 
\DMO{\comptr}{\comp_{tR}} 
\newcommand{\Us}{\mc{U}} 
\DMO{\comptru}{\comp_{tR(\Us)}} 
\DMO{\compru}{\comp_{R(\Us)}}

\DMO{\hardness}{hd}

\DMO{\pebf}{PF} 
\DMO{\rt}{rt} 
\DMO{\nds}{nds} 
\DMO{\lvs}{lvs} 
\DMO{\nlvs}{\#lvs} 
\DMO{\nnds}{\#nds} 
\DMO{\height}{ht}
\DMO{\depth}{d}
\DMO{\cls}{cls}
\DMO{\newcommandls}{\#cls}
\DMO{\ds}{ds}
\DMO{\dst}{ds_T}
\DMO{\dsg}{ds_G}
\DMO{\dpr}{dp}
\DMO{\dprt}{dp_T}
\DMO{\dprg}{dp_G}
\DMO{\ind}{in}
\DMO{\indg}{in_G}
\DMO{\outd}{out}
\DMO{\outdg}{out_G}
\DMO{\peb}{peb} 
\newcommand{\pab}[1]{\langle #1 \rangle}
\newcommand{\pao}[2]{\langle #1 \ra #2 \rangle}

\DMO{\taum}{\max \tau}
\DMO{\tauprob}{\tau^p} 
\DMO{\mtau}{\mf{T}} 
\DMO{\concatbt}{;} 
\DMO{\compobt}{\merge} 
\nc{\bth}[1]{\langle{#1}\rangle} 
\DMO{\pc}{pc}
\DMO{\aut}{Auk} 
\DMO{\laut}{LAuk} 
\DMO{\lautz}{LAuk_0} 
\DMO{\maut}{MAuk}
\newcommand{\A}{\mc{A}} 
\DMO{\nv}{N} 
\DMO{\na}{\nv_a} 
\DMO{\nA}{\nv_{\A}} 
\DMO{\nla}{\nv_{la}}
\DMO{\nbla}{\nv_{bla}}
\DMO{\nma}{\nv_{ma}}
\DMO{\npa}{\nv_{pa}}
\DMO{\baut}{BAuk} 
\DMO{\blaut}{BLAuk} 
\DMO{\blautz}{BLAuk_0} 
\DMO{\paut}{PAut} 
\DMO{\pautz}{PAut_0} 
\newcommand{\SatA}{\Sat\!_{\A}}
\newcommand{\LeanA}{\Lean\!_{\A}}
\DMO{\resouz}{\overset{\Us, 0}{\vdash}}
\DMO{\resouo}{\overset{\Us, 1}{\vdash}}
\DMO{\resouk}{\overset{\Us,\, k}{\vdash}}
\DMO{\resou}{\,\overset{\Us}{\vdash}\,}
\DMO{\resour}{\,\overset{\Us_0}{\vdash}\,}
\DMO{\resourz}{\,\overset{\Us_0, 0}{\vdash}\,}
\DMO{\uresouk}{\resouk\hspace{-0.6em}\mbox{\raisebox{0.8ex}{\tiny u}}}
\DMO{\bresouk}{\resouk\hspace{-0.6em}\mbox{\raisebox{0.8ex}{\tiny b}}}
\DMO{\iresouk}{\resouk\hspace{-0.6em}\mbox{\raisebox{0.8ex}{\tiny i}}}
\DMO{\resok}{\overset{k}{\vdash}} 

\DMO{\wid}{wid} 
\DMO{\widl}{\hspace*{-1.5pt}wid}
\DMO{\widb}{\sideset{^{\mr{b}}}{}\widl}
\DMO{\widi}{\sideset{^{\mr{i}}}{}\widl}
\DMO{\cwid}{\mc{W}} 
\DMO{\cwidl}{\hspace*{-1pt}\mc{W}} 
\DMO{\cwidb}{\sideset{^{\mr{b}}}{}\cwidl}
\DMO{\cwidi}{\sideset{^{\mr{i}}}{}\cwidl}
\DMO{\modp}{mod_p} 
\DMO{\modt}{mod_t} 
\DMO{\moda}{\mf{S}} 
\DMO{\modf}{fal} 
\DMO{\mods}{mod} 
\DMO{\mus}{MU}
\DMO{\mss}{MS}
\DMO{\cmus}{CMU}
\DMO{\cmss}{CMS}
\DMO{\eqs}{EQ} 
\DMO{\neqs}{NEQ} 
\DMO{\scf}{CM} 
\DMO{\acf}{DCM} 
\DMO{\cmg}{cmg} 
\DMO{\cmdg}{cmdg} 
\DMO{\cg}{cg} 
\DMO{\gcg}{cgg} 
\DMO{\gcdg}{cgdg} 
\DMO{\rsg}{rg} 
\DMO{\srsg}{srg} 
\DMO{\vhg}{vhg} 
\DMO{\cvg}{cvg} 
\DMO{\cvmg}{cvmg} 
\DMO{\vig}{vig} 
\DMO{\vcg}{vcg} 
\DMO{\nscf}{bcp} 
\DMO{\nacf}{bcp_d} 
\DMO{\bcp}{bcp} 
\DMO{\tbcp}{tbcp} 
\DMO{\nsat}{\#sat}
\DMO{\nusat}{\#usat}
\DMO{\maxsat}{maxsat}
\DMO{\pmin}{\rankmin}
\DMO{\pmax}{\rankmax}
\DMO{\pav}{\rankdur}
\DMO{\ldeg}{ld} 
\DMO{\minldeg}{\mu\!\ldeg} 
\DMO{\maxldeg}{\nu\ldeg} 
\DMO{\vdeg}{vd} 
\DMO{\minvdeg}{\mu\!\vdeg} 
\DMO{\maxvdeg}{\nu\!\vdeg} 
\DMO{\avvdeg}{\widetilde{\vdeg}} 
\DMO{\cldeg}{cldg} 
\DMO{\mvardu}{\mu\!\vdeg}
\DMO{\Inj}{Inj}
\newcommand{\OKlibrary}{\texttt{OKlibrary}}

\DMO{\Ex}{Ex} 
\DMO{\surp}{\sigma} 
\DMO{\nonmer}{nM} 
\DMO{\primec}{prc} 
\DMO{\timem}{time}
\DMO{\spacem}{space}
\DMO{\dtime}{DTime}
\DMO{\dspace}{DSpace}
\DMO{\ndtime}{NTime}
\DMO{\ndspace}{NSpace}
\DMO{\condtime}{coNTime}
\DMO{\condspace}{coNSpace}
\nc{\Con}{\mr{Con}}
\nc{\Log}{\mr{Log}}
\nc{\Lin}{\mr{Lin}}
\nc{\Pol}{\mr{Pol}}
\nc{\ExL}{\mr{ExL}}
\nc{\ExP}{\mr{ExP}}
\nc{\CTime}{\mr{CTime}}
\nc{\CSpace}{\mr{CSpace}}
\nc{\LTime}{\mr{LTime}}
\nc{\LSpace}{\mr{L}}
\nc{\NLSpace}{\mr{NL}}
\nc{\LinTime}{\mr{LinTime}}
\nc{\LinSpace}{\mr{LinSpace}}
\nc{\PTime}{\mr{P}}
\nc{\PSpace}{\mr{PSpace}}
\nc{\Np}{\mr{NP}}
\nc{\Conp}{\text{coNP}}
\nc{\NPSpace}{\mr{NPSpace}}
\nc{\CoNPSpace}{\mr{coNPSpace}}
\nc{\ELTime}{\mr{ELTime}}
\nc{\ELSpace}{\mr{ELSpace}}
\nc{\EPTime}{\mr{EPTime}}
\nc{\EPSpace}{\mr{EPSpace}}
\nc{\NEPTime}{\mr{NEPTime}}
\nc{\polydelta}[1]{\Delta_{#1}^{\mr P}}
\nc{\polypi}[1]{\Pi_{#1}^{\mr P}}
\nc{\polysigma}[1]{\Sigma_{#1}^{\mr P}}
\nc{\Ph}{\mr{PH}}
\DMO{\exP}{\ex^{\mr P}}
\DMO{\faP}{\fa^{\mr P}}
\nc{\Dp}{D^P}
\nc{\PllC}[2]{{\text{$\mr{PT}$/$\mr{WK}$}(#1, #2)}} 
\nc{\Nc}{\mr{NC}}
\nc{\Nci}[1]{\Nc^{#1}}
\nc{\Ac}{\mr{AC}}
\nc{\Aci}[1]{\Ac^{#1}}
\nc{\pmodpoly}{P / \mathrm{poly}}
\nc{\Wh}[1]{\mr{W}[#1]} 
\nc{\Rl}{\mr{RL}}
\nc{\coRl}{\mr{coRL}}
\nc{\Rp}{\mr{RP}}
\nc{\coRp}{\mr{coRP}}
\nc{\Zpp}{\mr{ZPP}}
\nc{\Bpp}{\mr{BPP}}
\nc{\Pp}{\mr{PP}}
\nc{\ramz}[3]{\mr{ram}_{#1}^{#2}(#3)} 
\DMO{\ramzg}{ram} 
\nc{\waez}[2]{\mr{vdw}_{#1}(#2)} 
\DMO{\waezg}{vdw} 
\nc{\gtz}[2]{\mr{grt}_{#1}(#2)} 
\DMO{\gtzg}{grt} 
\DMO{\FvdW}{F_{W}} 
\DMO{\FRam}{F_{R}} 
\DMO{\arithp}{ap} 
\DMO{\arithpp}{ap_{pr}} 
\DMO{\crarithp}{cr_{ap}} 
\DMO{\crarithpp}{cr_{ap}^{pr}} 
\usepackage{theorem} 
\usepackage[hypertex]{hyperref} 
\newtheorem{defi}{Definition}[section]
\newtheorem{lem}[defi]{Lemma}
\newtheorem{thm}[defi]{Theorem}
\newtheorem{corol}[defi]{Corollary}

\newtheorem{conj}[defi]{Conjecture}

\theorembodyfont{\rmfamily}

\theorembodyfont{}
\nc{\bm}{\boldmath}
\nc{\bmm}[1]{\mbox{\bm$\DST #1$}}
\nc{\mi}[1]{\bmm{\mathrm{(#1):}} \quad}

\usepackage[all]{xy}
\usepackage{enumerate}

\newenvironment{proof}{\noindent\textbf{Proof} \hspace*{0.2em}}{\Ende}

\DMO{\rd}{wn} 
\DMO{\tmc}{\hat{t}} 
\DMO{\tcm}{\check{t}} 
\DMO{\expan}{expn} 
\DMO{\autind}{\chi_a} 
\nc{\Bva}{\Va_{\set{0,1}}}
\nc{\Bcls}{\Cls(\Bva)}
\DMO{\alo}{ALO}
\DMO{\amo}{AMO}
\DMO{\transl}{\tau} 
\DMO{\ftrans}{\Theta}
\DMO{\ntrans}{\Theta_n}
\DMO{\sat}{sat}
\DMO{\hitdeg}{hd}
\DMO{\multhit}{mh}
\DMO{\stabpar}{sir}
\DMO{\mvd}{mvd}
\DMO{\mmvd}{mmvd}

\begin{document}

\title{Constraint satisfaction problems in clausal form }

\author{Oliver Kullmann\thanks{Partially supported by EPSRC Grant GR/S58393/01}\\[1ex]
  Computer Science Department\\
  Swansea University\\
  Swansea, SA2 8PP, United Kingdom\\
  O.Kullmann@Swansea.ac.uk \\
  {\small
    \url{http://cs.swan.ac.uk/~csoliver}}}

\maketitle

\begin{abstract}
  This is the report-version of a mini-series of two articles \cite{Kullmann2007ClausalFormZI,Kullmann2007ClausalFormZII} on the foundations of conjunctive normal forms with non-boolean variables. These two parts are here bundled in one report, each part yielding a chapter.

  \textbf{Part I}
  We consider the problem of generalising boolean formulas in conjunctive normal form by allowing \emph{non-boolean variables}, with the goal of maintaining \emph{combinatorial} properties. Requiring that a literal involves only a single variable, the most general form of literals are the well-known ``signed literals'', corresponding to unary constraints in CSP. However we argue that only the restricted form of ``negative monosigned literals'' and the resulting generalised clause-sets, corresponding to ``sets of no-goods'' in the AI literature, maintain the essential properties of boolean conjunctive normal forms. In this first part of a mini-series of two articles, we build up a solid foundation for (generalised) clause-sets, including the notion of autarky systems, the interplay between autarkies and resolution, and basic notions of (DP-)reductions. As a basic combinatorial parameter of generalised clause-sets we introduce the (generalised) notion of \emph{deficiency}, which in the boolean case is the difference between the number of clauses and the number of variables. \emph{Autarky theory} plays a fundamental role here, and we concentrate especially on \emph{matching autarkies} (based on matching theory). A natural task is to determine the structure of \emph{(matching) lean clause-sets}, which do not admit non-trivial (matching) autarkies. A central result is the computation of the lean kernel (the largest lean subset) of a (generalised) clause-set in polynomial time for bounded maximal deficiency.

  \textbf{Part II}
  Concluding this mini-series of 2 articles on the foundations of generalised clause-sets, we study the combinatorial properties of non-boolean conjunctive normal forms (clause-sets), allowing arbitrary (but finite) sets of values for variables, while literals express that some variable shall not get some (given) value. First we study the properties of the \emph{direct translation} (or ``encoding'') of generalised clause-sets into boolean clause-sets. Many combinatorial properties are preserved, and as a result we can lift fixed-parameter tractability of satisfiability in the maximal deficiency from the boolean case to the general case. Then we turn to \emph{irredundant clause-sets}, which generalise minimally unsatisfiable clause-sets, and we prove basic properties. The simplest irredundant clause-sets are hitting clause-sets, and we provide characterisations and generalisations. Unsatisfiable irredundant clause-sets are the \emph{minimally unsatisfiable clause-sets}, and we provide basic tools. These tools allow us to characterise the minimally unsatisfiable clause-sets of minimal deficiency. Finally we provide a new translation of generalised boolean clause-sets into boolean clause-sets, the \emph{nested translation}, which preserves the conflict structure. As an application, we can generalise results for boolean clause-sets regarding the hermitian rank/defect, especially the characterisation of unsatisfiable hitting clause-sets where between every two clauses we have exactly one conflict. We conclude with a list of open problems, and a discussion of the ``generic translation scheme''.

  \textbf{Keywords:} generalised clause-sets, signed formulas, non-boolean variables, satisfiability problem, constraint satisfaction problem, autarkies, deficiency, polynomial time, matching autarkies, lean clause-sets, boolean translations, direct encoding, irredundant clause-sets, minimally unsatisfiable clause-sets, hitting clause-sets, disjoint DNF, hermitian defect, nested translation
\end{abstract}

\tableofcontents

\chapter{Autarkies and deficiency}
\label{cha:I}

\section{Introduction}
\label{sec:intro1}

Satisfiability problems with variables having more than two values occur naturally in many places, for example in colouring problems. Translations into boolean satisfiability problems are interesting and useful (see \cite{FrischPeugniez2001NBLocalSearch,Prestwich2003Encodings,AnsoteguiManya2004FiniteDomainBooleanDomain,Pre09HBSAT,Kullmann2010GreenTao} for various techniques), and may even improve performance of SAT solving, however they hide to a certain degree the structure of the original problem, which causes these translations typically to be not very well suited for theoretical studies on the structure of the original problem. In this series of two articles\footnote{based on the report \cite{Kullmann2007ClausalFormECCC2}} we study non-boolean satisfiability problems closest to boolean conjunctive normal form, namely satisfiability of what is called \emph{generalised clause-sets} (or sets of ``no-goods''). Combining suitable generalisations of boolean techniques with suitable translations into boolean clause-sets (preserving certain structures) we obtain non-trivial generalisations of fundamental theorems on autarkies and minimally unsatisfiable formulas.

Three aspects of clauses (as combinations of literals) make processing of boolean clause-sets especially efficient:
\begin{enumerate}[(i)]
\item When the underlying variable of a literal gets a value, then the literal becomes either true or false (this enables efficient handling of literals).
\item Only by assigning a value to all the variables in a clause can we falsify the clause, and for each variable the value here is uniquely determined (this makes a tight connection between partial assignments and clauses, and enables ``conflict learning'' by clauses).
\item By giving just one variable a right value we are always able to satisfy a clause (this enables simple satisfaction-based heuristics).
\end{enumerate}
Taking these properties as axiomatic, a ``generalised clause'' should be a disjunction of generalised literals, and a ``generalised literal'' should have exactly one possibility to become false, while otherwise it should evaluate to true. We arrive naturally at the following concept for generalised literals (the earliest systematic use seems to be in \cite{Ba95}): A variable $v$ has a domain $D_v$ of values, and a literal is a pair $(v,\ve)$ of the variable and a value $\ve \in D_v$ such that the literal becomes true under an assignment $\vp$ iff $\vp$ sets $v$ to a value different than $\ve$ (i.e., $\vp(v) \in D_v \sm \set{\ve}$); to express this interpretation, often when displaying formulas we will write ``$v \not= \ve$'' for the literal $(v, \ve)$. In case of $D_v = \set{0,1}$, variable $v$ becomes an ordinary boolean variable with the literal $(v,0)$ representing the positive literal. We remark here that a fourth property of boolean clauses, namely that if all literals except of one are falsified, that then the value for the variable in the remaining literal is uniquely determined, which is the basis for the ubiquitous unit-clause propagation, is necessarily lost here.

In this first part we investigate the basic combinatorial properties of generalised clause-sets and the basics of the \emph{theory of autarkies}, while in the second part structural properties of \emph{minimally unsatisfiable generalised clause-sets} (and, more generally, \emph{irredundant generalised clause-sets}) will be investigated. In the following Subsections \ref{sec:introdef} - \ref{sec:introsignedformulas} we give an introduction into basic motivations and ideas, while in Subsection \ref{sec:introover1} we give an overview on the structure of this article and its main results. We remark that the most comprehensive and up-to-date overview on minimal unsatisfiability and autarkies in the boolean setting one finds in Chapter 11 in the Handbook of Satisfiability (see \cite{Kullmann2007HandbuchMU}).

\subsection{Generalising the notion of deficiency}
\label{sec:introdef}

Using $c(F)$ for the number of clauses in a clause-set, and $n(F)$ for the number of variables, in \cite{FrGe98} the \emph{deficiency} $\bmm{\delta(F)} := c(F) - n(F)$ has been introduced and made fruitful, for the study of minimally unsatisfiable boolean clause-sets as well as for the introduction of a new polynomial time decidable class of ``matched'' satisfiable (boolean) clause-sets:
\begin{itemize}
\item Let $\Musat$ denote the class of minimally unsatisfiable clause-sets (unsatisfiable clause-sets, where each strict sub-clause-set is satisfiable). For boolean $F \in \Musat$ the property $\fa\, F' \subset F : \delta(F') < \delta(F)$ has been shown; using $\delta^*(F) := \max_{F' \sse F} \delta(F')$ for the \emph{maximal deficiency} we get $\delta^*(F) = \delta(F)$ as well as ``Tarsi's lemma'' $\delta(F) \ge 1$ (since for the empty clause-set $\bmm{\top} \subset F$ we have $\delta(\top) = 0$).
\item The class $\Msat$ of ``matching satisfiable'' clause-sets $F$ is defined by the condition $\delta^*(F) = 0$. All matching satisfiable clause-sets are in fact satisfiable, since by Hall's theorem the bipartite graph $B(F)$ contains a matching covering all variables, where the vertices of $B(F)$ are the clauses of $F$ on the one side and the variables of $F$ on the other side, while an edge joins a variable and a clause if that variable appears in the clause (positively or negatively). Or, using Tarsi's lemma, one argues that if $F \in \Msat$ would be unsatisfiable, then $F$ would contain some minimally unsatisfiable $F' \sse F$, for which $\delta(F') \ge 1$ would hold, contradicting $\delta^*(F) = 0$.
\end{itemize}

The study of the levels $\Musat(k)$ of minimally unsatisfiable boolean clause-sets $F$ with $\delta(F) \le k$ has attracted some attention. In \cite{AhLi86} (where also Tarsi's lemma has been proven) the class $\Smusat$ of ``strongly minimally unsatisfiable clause-sets'' has been introduced, which are minimally unsatisfiable clause-sets such that adding any literal to any clause renders them satisfiable, and a nice characterisation of $\Smusat(1) = \set{ F \in \Smusat : \delta(F) = 1}$ has been given (yielding polynomial time decision of $\Smusat(1)$). Then in \cite{DDK98} a characterisation of $\Musat(1)$ has been obtained, followed by a (partial) characterisation of $\Musat(2)$ in \cite{KleineBuening2000SubclassesMU} (achieving poly-time decision), while in \cite{XD99} some subclasses of $\Musat(3)$ and $\Musat(4)$ have been shown to be poly-time decidable. For arbitrary (constant) $k \in \NN$ it has been shown in \cite{Kl98} that for $F \in \Musat(k)$ there is a tree resolution refutation using at most $2^{k-1} \cdot n(F)^2$ steps, and thus the classes $\Musat(k)$ are in NP (for fixed $k$; by just guessing the refutation). In \cite{Kl98} it has been conjectured that in fact all classes $\Musat(k)$ are in P (for fixed $k$).

This conjecture has been proven true in \cite{Ku99dKo}, using tools from matroid theory. Actually the classes $\Sat(k)$, consisting of all satisfiable clause-sets $F$ with $\delta^*(F) \le k$, have been shown poly-time decidable, from which immediately poly-time decision of the classes $\Musat(k)$ and $\Smusat(k)$ follows (all for arbitrary but fixed $k$). Regarding the method used, more precisely the classes $\Usat(k)$ of \emph{unsatisfiable} clause-sets $F$ with $\delta^*(F) \le k$ have been shown poly-time decidable by improving the ``splitting theorem'' from \cite{DDK98}. Tree resolution refutations for $F$ using at most $2^{k-1} \cdot n(F)$ steps and of a simple recursive structure have been obtained, so that these refutations can be found in polynomial time by means of enumeration of the circuits of the transversal matroid $T(F)$ associated to the bipartite graph $B(F)$, where the independent subsets of $T(F)$ are the matching satisfiable sub-clause-sets of $F$. Independently also in \cite{FS00a} poly-time decision of the classes $\Musat(k)$ has been derived by extending techniques from bipartite matching theory to \emph{directed} bipartite graphs. Improving the proofs from \cite{FS00a}, the present author joined the team in \cite{FKS00}. Actually refining the techniques from \cite{Ku99dKo}, in \cite{Szei2002FixedParam} fixed-parameter tractability of $\Sat(k)$ is shown (all this for the boolean case).

After setting syntax and semantics for generalised clause-sets, the first main task tackled in the present paper is to transfer these results regarding the deficiency to generalised clause-sets. After suitably generalising the notion of \emph{deficiency} and \emph{matching satisfiability} (which is not completely straight-forward; in Subsection \ref{sec:comparisonearlier} an earlier version is discussed, which doesn't seem to have the right properties), in Corollary \ref{cor:poly1} the ``satisfiability-based'' approach from \cite{FKS00} yields polynomial time satisfiability decision for generalised clause-sets with bounded maximal deficiency. Generalising fixed-parameter tractability turns out not to be straight-forward (again), and only by combining the generalised approach with a suitable translation into the boolean case we arrive in Theorem \ref{thm:MaximalerDefektFPT} (Part II) at fixed parameter tractability also for generalised clause-sets. The general framework for our considerations is autarky theory as started in \cite{Ku98e}, with emphasise on \emph{matching autarkies} as introduced in \cite{Ku00f}. We remark that autarky theory provides a natural connection to hypergraph colouring, as investigated in connection with ``Polya's Problem'' in \cite{Kullmann2007Balanciert} (for the case of 2-colourings); see the subsequent subsection for the general notions.

A key point for structural investigations in (generalised) clause-sets is to understand the effects of applying partial assignments; see for example \cite{KleineBueningZhao2003StructureSomeClassesMU,KleineBueningZhao2007ComplexitySomeSubclassesMU}, where splitting of minimally unsatisfiable boolean clause-sets is studied in some depth. In Part II we will consider the basic questions regarding \emph{irredundant} and \emph{minimally unsatisfiable generalised clause-sets}, which leads in a natural way to the study of \emph{hitting clause-sets} and generalisations. The well-known classifications of the simplest case of minimally unsatisfiable clause-sets, namely boolean clause-sets of deficiency $1$, finds a natural generalisation in Theorem \ref{thm:CharakMUSATd1} (Part II, where again the proof is not straight-forward, caused by the breakdown of the ``saturation method''). While minimally unsatisfiable clause-sets have a deficiency at least $1$, in Corollary \ref{cor:GPvKlm} (Part II) we show (generalising \cite{Ku2003e}) that \emph{regular hitting clause-sets}, generalised clause-sets which have a constant number of clashes between any two different clauses, have deficiency at most $1$. The intersection of both classes, the unsatisfiable regular hitting clause-sets, is characterised in Corollary \ref{cor:CharakUnsatreghit} (Part II), and found for the special case of boolean clause-sets a recent application in \cite{SloanSzoerenyiTuran2005Primimplikanten_1}, while we have generalised these results in \cite{HendersonKullmann2007Multicliquen}.

\subsection{Examples: Colourings and homomorphisms}
\label{intro:hypergraphcolouring}

Generalised clause-sets allow a natural representation of many combinatorial problems, and in this section we discuss basic examples.

Given a hypergraph $G$ and a set $C$ of ``colours'', a $C$-colouring of $G$ is a map $f: V(G) \ra C$ such that no hyperedge $H \in E(G)$ is ``monochromatic'' (that is, there must be vertices $v, w \in H$ with $f(v) \not= f(w)$). Translating this colouring problem into a generalised satisfiability problem $F_C(G)$ is straightforward, using the vertices of $G$ as variables with (uniform) domain $C$:\footnote{This translation generalises the well-known translation of graph 2-colouring problems into boolean CNF. Appending the translation from generalised clause-sets into boolean clause-sets via the direct translation (see Section \ref{sec:translating} in Part II), this combined translation also generalises the well-known direct translation of (arbitrary) graph colouring problems into boolean CNF.}
\begin{itemize}
\item For hyperedge $H \in E(G)$ and colour $\ve \in C$ form the clause $\set{v \not= \ve : v \in H}$.
\item $F_C(G)$ is the set of all these clauses; thus $n(F_C(G)) = \abs{V(G)}$ and $c(F_C(G)) = \abs{C} \cdot \abs{E(G)}$.
\end{itemize}
Obviously the $C$-colourings of $G$ correspond 1-1 to the (total) satisfying assignments for $F_C(G)$. Interesting examples of hypergraph colouring problems are given by the diagonal van der Waerden problems and the diagonal Ramsey problems. Computing van der Waerden numbers has been considered in \cite{DransfieldLiuMarekTruszcynski2004VanderWaerden,KourilFranco2005Resolutiontunnel,HerwigHeuleLambalgenMaaren2005VanderWaerden,KourilPaulW26,Ahmed2009vdW,Kullmann2009RamseySATExtern2,Ahmed2010vdW,AhmedKullmannSnevily2011VdW3k}, and SAT solvers are performing quite well (likely the best method available), actually helping to compute new van der Waerden numbers\footnote{Ramsey problems seem harder for SAT solvers in the sense that these problems carry some structure unknown to the SAT solver (but known to the researcher), so that the problems related to unknown Ramsey numbers are rather big; see \cite{Radziszowski2006RamseySurvey} for the known numbers. Nevertheless improving the currently known bounds $43 \le \ramz 2{}5 \le 49$ might be in reach for specialised SAT solvers (where $\ramz 2{}5$ is the minimal number of vertices in a complete graph such that every labelling of the edges with two ``colours'' must have a ``monochromatic'' $5$-clique). Though one needs to say here that there is strong evidence for $\ramz 2{}5 = 43$ (see \cite{McKayRadziszowski1997RamseyNumbers}), and so an apparently very hard unsatisfiable problem is to be tackled.}, so here is the problem:

Consider natural numbers $k, n \in \NN$, and let the hypergraph $\arithp(k,n)$ have vertex set $\tb 1n$, while the hyperedges of $\arithp(k,n)$ are the subsets $H \sse \tb 1n$ of size $k$ which form an arithmetic progression, that is, for every $H$ there exist $a, d \in \tb 1n$ with $H = \set{a + i \cdot d : i \in \tb{0}{k-1}}$. Now for $m \in \NN$ the van der Waerden number $\waez{m}{k}$ is the minimal $n$ such that $\arithp(k,n)$ is not $m$-colourable. The corresponding generalised clause-sets are $\FvdW(m,k,n) := F_{\tb 1m}(\arithp(k,n))$, and if $\FvdW(m,k,n)$ is satisfiable, then $\waez{m}{k} > n$, while if $\FvdW(m,k,n)$ is unsatisfiable, then $\waez mk \le n$; for $m = 2$ we obtain boolean clause-sets (I would like to point out how natural the translation is --- no auxiliary variables are involved\footnote{The translation is the core of the two translations discussed in \cite{DransfieldLiuMarekTruszcynski2004VanderWaerden} --- the additional constraints used in \cite{DransfieldLiuMarekTruszcynski2004VanderWaerden} just express the structural property of a generalised clause-set that every variable gets exactly one value of its domain.}). The only known precise van der Waerden numbers (besides the trivial values for $m = 1$ or $k \le 2$) are $\waez 23 = 9$, $\waez 24 = 35$, $\waez 25 = 178$, $\waez 33 = 27$ and $\waez 43 = 76$, and all these numbers can be easily calculated using most current SAT solvers; furthermore recently in \cite{KourilPaulW26} by (extensive) SAT-computations $\waez 26 = 1132$ has been confirmed (as conjectured by \cite{KourilFranco2005Resolutiontunnel}). Directly expressing the problem instance as a generalised clause-set (and skipping the hypergraph colouring problem), in this way also the non-diagonal versions of van der Waerden and Ramsey problems (making different requirements on the different colours) can be immediately translated into generalised clause-sets; see \cite{LandmanRobertsonCulver2005vanderWaerden} for some direct computations, and \cite{Ahmed2009vdW,Ahmed2010vdW,AhmedKullmannSnevily2011VdW3k} for more advanced computations using SAT solving (both regarding van der Waerden problems). In \cite{Kullmann2009RamseySATExtern2} we give an overview on the known van der Waerden numbers, and ``Green-Tao numbers'' are introduced, based on the celebrated strengthening of van der Waerden's theorem by the Green-Tao theorem (\cite{GreenTao2005Primes}). For further applications of the mapping $G \mapsto F_C(G)$ from hypergraphs to clause-sets see \cite{Kullmann2005b2,Kullmann2007Balanciert}.

For the more general \emph{list-hypergraph colouring problem}, for each vertex $v$ a list $L(v)$ of allowed colours is given; this can be translated into a generalised clause-set $F_C(G, L)$ by just restricting the domain of $v$ to $L(v)$. At this point it is worth noticing that also the still more general \emph{list-hypergraph-homomorphism problem} has a direct (structure-preserving) translation into a satisfiability problem for generalised clause-sets. Given two hypergraphs $G_1, G_2$ and for each vertex $v \in V(G_1)$ a non-empty set $L(v) \sse V(G_2)$ of allowed image vertices, the problem is to find a map $f: V(G_1) \ra V(G_2)$ with $f(v) \in L(v)$ for all $v \in V(G_1)$ such that for each hyperedge $H \in E(G_1)$ we have $f(H) \in E(G_2)$. Note that if we take for $G_2$ the hypergraph $G_C$ with vertex set $C$ and hyperedges all subsets of $C$ with at least two elements, then the homomorphisms from $G_1$ to $G_2$ are exactly the $C$-colourings for $G_1$. For the translation of the list-hypergraph-homomorphism problem we use the set $V(G_1)$ of vertices as the set of variables, while the domain of $v$ is $D_v = L(v)$, and for each hyperedge $H \in E(G_1)$ and for each map $f: H \ra V(G_2)$ such that for each $v \in H$ we have $f(v) \in L(v)$ and such that $f(H) \notin E(G_2)$ holds, we have a clause $C_{H,f} := \set{v \not= f(v) : v \in H}$. Now satisfying assignments of the generalised clause-set $F(G_1,G_2,L)$ consisting of all clauses $C_{H,f}$ are exactly the hypergraph homomorphisms from $G_1$ to $G_2$ respecting the restrictions given by $L$. Note that the translation of hypergraph colouring problems is a special case via $F_C(G, L) = F(G, G_C, L)$.

The colourings considered above are also called ``weak hypergraph colourings'', to distinguish them from \emph{strong hypergraph colourings} of hypergraph $G$ by a set $C$ of colours, which is a map $f: V(G) \ra C$ such that for all hyperedges $H \in E(G)$ and all $v, w \in H$, $v \not= w$, we have $f(v) \not= f(w)$. Collecting all such binary clauses $\set{f(v) \not= \ve, f(w) \not= \ve}$ for $\ve \in C$ we obtain a generalised clause-sets whose satisfying (total) assignments correspond 1-1 to the strong colourings of $G$. To conclude our list of translations for colouring problems we mention that \emph{mixed hypergraph colouring} as studied in \cite{Voloshin2002GemischteHypergraphenfaerbung} takes a pair $G_1, G_2$ of hypergraphs with $V(G_1) = V(G_2)$, and the ``mixed colourings'' using colour-set $C$ are (weak) colourings of $G_1$ which are \emph{not} strong colourings of $G_2$; the most natural translation of this problem seems to consist of a pair of a generalised CNF (the translation of the (weak) colouring problem) and a generalised DNF (the negation of the translation of the strong colouring problem), using ``monosigned literals'' (see the next subsection), that is, allowing for inequalities ``$v \not= \ve$'' as well as equalities ``$v = \ve$''.

In the same vein as for hypergraph homomorphisms we can also translate \emph{homomorphism problems for relational structures}: Let $\mc{A} = (A, (R_i)_{i \in I})$ and $\mc{B} = (B, (R'_i)_{i \in I})$ be two compatible finite relational structures, that is, $A, B$ as well as $I$ are finite sets, the $R_i$ are relations (of arbitrary arity) on $A$ and the $R_i'$ are relations on $B$, while $R_i$ has the same arity as $R_i'$. We want to express the set of homomorphisms $f : A \ra B$, defined by the property that for $i \in I$ and all $\vec{x} \in R_i$ we have $f(\vec{x}) \in R_i'$, where $f$ is applied componentwise to $\vec{x}$. For this we choose $A$ as the set of variables, which all have the same domain $B$, and for each $i \in I$, each $\vec{x} \in R_i$ and each $\vec{y} \in B^m \sm R_i'$, where $m$ is the arity of $R_i$, we have the clause $C_{i,\vec{x},\vec{y}} := \set{\vec{x}_i \not= \vec{y}_i : i \in \tb 1m}$. We obtain the generalised clause-set $F(\mc{A}, \mc{B})$ by collecting all these clauses. The size of $F(\mc{A}, \mc{B})$ is polynomial in the sizes of $A, B$ together with the number of tuples in $R_i$ and the number of tuples \emph{not} in $R_i'$. The requirement that the homomorphism $f$ is injective can be encoded by the binary clauses $\set{a \not= b, a' \not= b}$ for $a, a' \in A$ with $a \not= a'$ and for $b \in B$.\footnote{On the other hand, to formulate surjectivity of $f$ requires an exponential number of clauses; one sees that for the good combinatorial properties, which can be exploited for problems expressed in the language of (generalised) clause-sets (which is quite restricted from the constraint programming point of view), we have to pay the price that some natural problems do not have succinct representations (without using additional variables).} We note that the translations $F(G_1, G_2, L)$ as well as $F(\mc{A}, \mc{B})$ are ``direct'' (homomorphisms are directly encoded as assignments) and ``negative'' (we use ``forbidden'' value combinations).

If we wish to have $F(\mc{A}, \mc{B})$ polynomial in the number of tuples in $R_i'$, then we can use an ``indirect'' and ``positive'' translation as follows: Variables are pairs $(i, \vec{x})$ for $i \in I$ and $\vec{x} \in R_i$, where the domain of variable $(i, \vec{x})$ is $R_i'$; so instead of mapping elements of $A$ to elements of $B$, where constraints forbid that allowed tuples are mapped to disallowed tuples, here now we directly map tuples from relations in $\mc{A}$ to tuples in the corresponding relation in $\mc{B}$, and the constraints will ensure that this mapping actually is induced by some mapping from $A$ to $B$. The constraints are the unit clauses $\set{(i, \vec{x}) \not= \vec{y}}$ for variables $(i, \vec{x})$ and values $\vec{y}$ such that indices $k, k'$ exist with $\vec{x}_k = \vec{x}_{k'}$ but $\vec{y}_k \not= \vec{y}_{k'}$, and the binary clauses $\set{(i, \vec{x}) \not= \vec{y}, (i', \vec{x}') \not= \vec{y}'}$ for variables $(i, \vec{x}), (i', \vec{x}')$ and values $\vec{y}, \vec{y}'$, such that an index $k$ exists with $\vec{x}_k = \vec{x}'_k$ but with $\vec{y}_k \not= \vec{y}'_k$.

\subsection{Signed formulas and resolution}
\label{sec:introsignedformulas}

Are there still more general versions of ``generalised conjunctive normal forms'' suitable in our context? The most general form of variable-based literals allows literals of the form ``$v \in S$'' for some $S \sse D_v$, and then (our) generalised literals $(v, \ve)$ are represented using $S = D_v \sm \set{\ve}$. See \cite{BeckertHaehnleManya2000SignedCNF,AnsoteguiBejarCabiscolLiManya2002ManyValued,AnsoteguiManya2004FiniteDomainBooleanDomain,AnsoteguiBonetLevyManya2007CSPSAT} for entry points into the literature, where $S$ is called a ``sign'', while literals of the form ``$v \in S$'' are called ``signed literals'', and clause-sets made of signed literals ``signed CNF formulas''. More precisely literals ``$v \in S$'' are ``positive'', while a ``negative'' literal is of the form ``$v \notin S$'', and so generalised literals $(v,\ve)$ correspond best to ``$v \notin \set{\ve}$'' instead of (as above) ``$v \in D_v \sm \set{\ve}$''; obviously when considering all possible signs then the positive and the negative formulations are equivalent, but this changes when considering only special types of signs, as we will see below. See also \cite{FrischPeugniez2001NBLocalSearch}, which uses the same class of formulas, but calling them ``nb-formulas''. Our generalised clause-sets, based on ``$v \not= \ve$'', are ``negative monosigned CNF formula'' (the signs are singleton sets), while ``monosigned CNF formula'' allow also positive signed literals (with singleton signs; alternatively, negative signs with just one elements missing can be used). We remark that in general negative literals naturally belong to CNF, while positive literals naturally belong to DNF.

So the closest (further) generalisation of our ``clause-sets'' are ``monosigned CNF formulas'', which allow besides ``$v \not= \ve$ also ``$v = \ve$''. Considering this extension is also motived by the fact that these formulas correspond exactly to their boolean counterpart via the natural translation (which uses a boolean variable to express that a variable is equal (or not) to a specific value). However, monosigned formulas seem to lack the good combinatorial properties which ``negative monosigned formulas'' have. This can be seen for example by the fact that the boolean translation of monosigned CNF formulas need the ``AMO'' clauses (expressing that every (original) variable gets at most one value), making the translated formula unwieldy (from a combinatorial point of view), while the AMO clauses are not needed for negative monosigned CNF --- here, without destroying the satisfaction relation we can just select some value in case an original variable gets several values, which is not possible in the presence of literals demanding that a variable gets some fixed value, since here those several values make more clauses true than any single of them. The point is that negative monosigned formulas need ALO, while AMO is optional, and positive monosigned formulas need AMO, while ALO is optional, and finally monosigned formulas need ALO as well as AMO.

An important point has been raised in \cite{HwangMitchell2005MehrereZweige}, where it has been shown that splitting on the boolean translation of generalised clause-sets can have an \emph{exponential speed-up} over the (wide) splitting only available when splitting on the original (``negative'') literals, where one considers $\abs{D_v}$ many branches for splitting on a variable $v$, each branch fixing a value of $v$.\footnote{The corresponding form of resolution has been studied in some depth in \cite{Ku00g} (generalised there through the use of oracles), and is considered in this paper in Section \ref{sec:Resolution}.} This seems to be an inherent weakness of using generalised clause-sets for SAT solving --- they are not expressive enough to allow certain crucial inferences to be formulated succinctly. But actually our model of generalised clause-sets can allow the form of binary splitting corresponding to splits on the boolean translation as follows: Our literals can express only ``$v \not= \ve$'', but since we allow arbitrary variable domains, we can have a binary splitting with a domain collapse $D_v \mapsto \set{\ve}$ in one branch (splitting on the positive literal ``$v = \ve$'') and a domain restriction $D_v \mapsto D_v \sm \set{\ve}$ in the other branch (i.e., splitting on the negative literal ``$v \not= \ve$''): In the first branch all literals with variable $v$ would become true or false, while in the second branch possibly the literal stays, and only the domain of $v$ is restricted (globally).\footnote{That is, since in the second branch we do not assign a value to variable $v$, we do not get rid off $v$ in the second branch. As a consequence, we need a global domain management.} More generally, if $(D_i)_{i \in I}$ is a partition of $D_v$ then we can split into $\abs{I}$ branches where in branch $i$ variable $v$ gets the new domain $D_i$; if for a literal $(v,\ve)$ we have $\ve \notin D_i$, then the literal (and thus the clause) becomes true, while if $D_i = \set{\ve}$, then the literal becomes false, and otherwise just the domain of $v$ is restricted (globally). The splitting trees for (generalised) clause-sets with domain-splittings ``$(\set{\ve}, D_v \sm \set{\ve})$'' correspond exactly to the splitting trees for the natural boolean translations. The price we have to pay however for this more powerful branching is, that if we stick with (generalised) clause-sets, then we can not have (full) clause learning --- if we want to use clause learning, in this way reflecting the search process within the ``clause-database'', then at least for recording the learned clauses we need monosigned clauses to record these binary splittings (and signed literals for more general domain splittings).\footnote{This discussion shows in my opinion a major reason, why generalising boolean reasoning proved to be difficult in the past, and (boolean) SAT solvers have an edge over constraint solvers: Either we restrict ourselves to wide branching, which has inherent inefficiencies, or we use more powerful branching, and then we have to use a more complicated domain management than in the boolean case (where none is needed), and finding out whether a literal actually became true or false becomes considerably more complicated (while it is trivial in the boolean case). Furthermore, if we want to use learning, which seems of importance for many ``real-world'' problems, then we have to use more complicated literal structures, and domain and literal (occurrence) management gets further complicated.} This distinction between the ``logic'' representing the input and the ``logic'' representing branching will be further developed in the \OKlibrary\footnote{See \url{http://www.ok-sat-library.org}.}, a generative library for generalised SAT solving (\cite{Kullmann2009OKlibrary} gives an introduction into the general ideas of this ``research platform''): Problem instances can be represented by arbitrary structures (based on the axiomatic approach in \cite{Ku00g}), while the branching logic (inference and learning) uses various types of ``literals'' appropriate for the algorithms used.

\subsection{Overview and main results}
\label{sec:introover1}

In \textbf{Section \ref{sec:prelim}} we present some preliminaries for our study of generalised clause-sets: Partial assignments for non-boolean variables, and fundamental notions and notations for graphs. Then in \textbf{Section \ref{sec:prelimgeneralcls}} generalised clause-sets and the main associated operations are introduced. Autarkies and autarky systems for generalised clause-sets are reviewed in \textbf{Section \ref{sec:prelimAutarkies}} (a useful result here is Lemma \ref{lem:vonschlanknachred}, showing how to actually find a non-trivial autarky when just given an oracle deciding whether a non-trivial autarky exists or not). Resolution for generalised clause-sets is the subject of \textbf{Section \ref{sec:Resolution}} (in Theorem \ref{thm:ResAut} it is proven that a clause can be used in some resolution refutation iff it can not be satisfied by some autarky; computation of the lean kernel via ``intelligent backtracking solvers'' follows). And the most basic polynomial time reductions for generalised clause-sets are presented in \textbf{Section \ref{sec:Reduction}}.

After these preliminaries and foundations, we turn to the main subject of this article, the study of matching autarkies and their uses. First in \textbf{Section \ref{sec:satviamatch}} the notion of matching satisfiable clause-sets (introduced in \cite{FrGe98} under the name of ``matched clause-sets'') is generalised in a natural way to generalised clause-sets, based on the generalised notion of deficiency. Matching satisfiable clause-sets are satisfiable by special satisfying assignments, called ``matching-satisfying assignments'', and these assignments are studied in \textbf{Section \ref{sec:matchsatass}}. Theorem \ref{thm:Reparatur} is the first main result, guaranteeing the existence of satisfying assignments ``close enough'' to matching-satisfying assignments. In Corollary \ref{cor:poly1} we derive poly-time satisfiability decision for generalised clause-sets with bounded maximal deficiency, generalising and strengthening the approach from \cite{FKS00}. Proving fixed parameter tractability with respect to the maximal deficiency for generalised clause-sets will be established later, in Theorem \ref{thm:MaximalerDefektFPT} (Part II), and that result will not use Theorem \ref{thm:Reparatur}. The main application of Theorem \ref{thm:Reparatur} is Theorem \ref{thm:leanKernelpoly}, where the harder problem of finding a non-trivial autarky (not just a satisfying assignment) is shown to be solvable in polynomial time in the maximal deficiency.

\textbf{Section \ref{sec:autviamatch}} now studies matching autarkies for generalised clause-sets. In Subsection \ref{sec:matchautsub} matching autarkies for generalised clause-sets are introduced, and the main properties are proven. A typical result here is the generalisation of ``Tarsi's Lemma'' in Corollary \ref{cor:musatdelta} (every generalised minimally unsatisfiable clause-set has deficiency at least one), exploiting the notion of matching leanness. In \textbf{Section \ref{sec:leanKernelpoly}}, as already mentioned, we strengthen polynomial time satisfiability decision w.r.t.\ bounded maximal deficiency $\delta^*(F)$ to the ability of computing the lean kernel (in polynomial time). Finally in \textbf{Section \ref{sec:expansion}} the notions of ``expansion'' and ``surplus'' are transferred from matching theory, yielding a simplified proof of FPT for SAT decision w.r.t.\ the parameter $\delta^*(F)$ in the boolean case (however we do not get a proof for the general case). Many interesting open problems arose in the journey so far, and we discuss some of them in \textbf{Section \ref{sec:open1}}.

\section{Preliminaries}
\label{sec:prelim}

We use $\NN = \ZZ_{\ge 1}$ and $\NNZ = \ZZ_{\ge 0}$. For a set $X$ by $S_X$ the group of all bijections from $X$ to $X$ is denoted, while for $n \in \NNZ$ we set $S_n := S_{\tb 1n}$.

\subsection{Variables and partial assignments}
\label{sec:prelimpass}

Fundamental for our considerations is the \textbf{monoid $(\Pass, \circ, \epa)$ of partial assignments}. Here we just recall the basic definitions, while a full account (for the boolean case) can be found in \cite{Kullmann2007HandbuchMU}.

The universe of variables is denoted by the infinite set \bmm{\Va}, while the universe of domain elements is the infinite set \bmm{\Dom}; a \textbf{(value-)domain} is a finite non-empty subset of $\Dom$, and for each variable $v \in \Va$ we denote the associated (value-)domain by \bmm{D_v}; thus variables have fixed (value-)domains, and change of domain (for example removal of values) must be performed by renaming. For a domain $D$ by \bmm{\Va_D} the set of all variables with domain $D$ is denoted; to avoid running out of variables and to ease renaming, we make the assumption that for all domains $D$ the set \bmm{\Va_D} has the same cardinality as $\Va$ itself. A variable $v \in \Va$ is called \textbf{boolean} if $D_v = \set{0,1}$ (and thus $\Bva$ is the set of all boolean variables).

A \textbf{partial assignment} is a map $\vp$ with finite domain $\bmm{\var(\vp)} := \dom(\vp) \sse \Va$, such that for all $v \in \var(\vp)$ we have $\vp(v) \in D_v$. The domain size of a partial assignment $\vp$ is denoted by $\bmm{n(\vp)} := \abs{\var(\vp)} \in \NNZ$. A special partial assignment is the empty partial assignment $\epa$. The set of all partial assignments is denoted by \bmm{\Pass}, while for some set $\Va' \sse \Va$ of variables we denote by $\bmm{\Pass(\Va')} := \set{\vp \in \Pass : \var(\vp) \sse \Va'}$ the set of partial assignments with variables from $\Va'$ (thus $\Pass(\Bva)$ is the set of partial assignments for boolean variables). We use the notation $\pab{v_1 \ra \ve_1, \dots, v_m \ra \ve_m}$ to denote the partial assignment $\vp$ with $n(\vp) = m$ and $\vp(v_i) = \ve_i$ (for distinct variables $v_i$, or for literals with distinct underlying variables).

For two partial assignments $\vp, \psi \in \Pass$ their \textbf{composition} \bmm{\vp \circ \psi} is defined as the partial assignment $\vp \circ \psi$ with domain $\var(\vp \circ \psi) = \var(\vp) \cup \var(\psi)$ such that first $\psi$ is evaluated and then $\vp$, i.e., $(\vp \circ \psi)(v) = \psi(v)$ if $v \in \var(\psi)$ while otherwise $(\vp \circ \psi)(v) = \vp(v)$. An important basic observation is that $(\Pass, \circ, \epa)$ is a monoid. An alternative representation of this structure is obtained as follows: Make each $D_v$ a (``right zero'') semigroup $(D_v, \cdot)$ by defining $\ve_1 \cdot \ve_2 := \ve_2$ for $\ve_1, \ve_2 \in D_v$. Adjoin an identity element ``$*$'' to each $D_v$, obtaining monoids $D_v^*$. Now $\Pass$ is isomorphic to the direct sum $\sum_{v \in \Va} D_v^*$ of the monoids (the sub-monoid of the direct product $\prod_{v \in \Va} D_v^*$ given by those elements where only finitely many components are different from $*$), where $\vp \in \Pass$ corresponds to the map $\vp^* \in \prod_{v \in \Va} D_v^*$ with $\vp(v) = \vp^*(v)$ for $v \in \var(\vp)$ and $\vp^*(v) = *$ for $v \in \Va \sm \var(\vp)$. This representation of partial assignments as total maps with distinguished ``undefined'' value $*$ actually has certain advantages over representation using partial maps, since working with total maps is often easier than working with partial maps, and we get a somewhat richer algebraic structure; however in this article we stick to the above representation of partial assignments.

\subsection{Graphs and matchings}
\label{sec:Graphsmatching}

A (finite) \emph{graph} $G$ here is a pair $G = (V, E)$ with finite vertex set $V(G) = V$ and edge set $E(G) = E \sse \binom V2$, where for a set $M$ and $k \in \NNZ$ by $\binom Mk$ we denote the set of all subsets $T \sse M$ with $\abs{T} = k$. So graphs here have no parallel edges and no loops. A graph $G'$ is a \emph{subgraph} of a graph $G$ if $V(G') \sse V(G)$ and $E(G') \sse E(G)$; $G'$ is called a \emph{partial subgraph} of $G$ if $G'$ is a subgraph of $G$ and $V(G') = V(G)$. A graph $G$ is called a \emph{forest} if $G$ contains no cycle. A graph $G$ is \emph{complete} if all distinct vertices $v, w \in V(G)$ are adjacent. $G$ is \emph{bipartite}, if the chromatic number of $G$ is at most $2$, while $G$ is \emph{complete bipartite} if $G$ is bipartite and addition of any edge to $G$ either destroys the graph property (i.e., creates a loop or a parallel edge) or destroys the bipartiteness property. More generally, $G$ is called \emph{complete $k$-partite} for $k \in \NNZ$ if the chromatic number of $G$ is at most $k$, and addition of any edge to $G$ either destroys the graph property or increases the chromatic number. $G$ is complete $k$-partite iff $G$ is the union of at most $k$ independent sets, such that each pair of vertices from different independent sets is adjacent (equivalently, iff the complement of $G$ is the disjoint union of at most $k$ cliques).

A function $f: S \ra \RR$, where $S$ is some set system stable under union and intersection, is called \emph{submodular} resp.\ \emph{supermodular} if for all $A, B \in S$ we have $f(A \cup B) + f(A \cap B) \le f(A) + f(B)$ resp.\ $f(A \cup B) + f(A \cap B) \ge f(A) + f(B)$, while $f$ is called \textbf{modular} if $f$ is submodular and supermodular. A prototypical example for a modular function is $A \sse X \mapsto f(A) := \abs{A}$, where $X$ is some finite set. For a graph $G$ and a vertex set $A \sse V(G)$ the \emph{(closed) neighbourset} $\Gamma_G(A)$ is defined as the set of vertices adjacent to at least one element of $A$. The function $A \sse V(G) \mapsto \abs{\Gamma(A)}$ is a prototypical example for a submodular function, while the \emph{deficiency} $\delta(A) := \abs{A} - \abs{\Gamma(A)} \in \ZZ$ is a supermodular function (as the difference of a modular function and a submodular function).

A \emph{matching} $M$ in a graph $G$ is a set $M \sse E(G)$ of edges such that two distinct elements of $M$ are non-adjacent. A matching in $G$ which is of maximal size is called a \emph{maximum matching}, and the size of a maximum matching $M$ of $G$ is denoted by $\nu(G) := \abs{M}$. If $G$ is a bipartite graph with bipartition $(A, B)$ (also called ``colour classes''), then $\nu(G) = \abs{A} - \delta^*(A) = \abs{B} - \delta^*(B)$, where for any $S \sse V(G)$ we set $\delta^*(S) := \max_{S' \sse S} \delta(S')$.\footnote{See for example Theorem 22.2 in \cite{Schrijver2003CombOptA}, where the notion of ``transversals'' or ``systems of distinct representatives'' of a set system is used (not to be mixed up with ``transversals'' in hypergraphs), and where the set system is $(\Gamma_G(\set{a}))_{a \in A}$ resp.~$(\Gamma_G(\set{b}))_{b \in B}$.} A maximum matching in a bipartite graph can be computed in time $O(\sqrt{\nu(G)} \cdot \abs{E(G)}) \le O(\sqrt{\abs{V(G)}} \cdot \abs{E(G)})$ (see Theorem 16.4 in \cite{Schrijver2003CombOptA}).\footnote{We won't dwell here on the details of graph representations; however when stating complexity results for (multi-)clause-sets we will be more precise.}

A \emph{maximal matching} $M$ in a graph $G$ is one which can not be extended (that is, there is no matching $M'$ in $G$ with $M \subset M'$), while the vertices \emph{covered} by a matching $M$ are the vertices incident to one of the edges in $M$. An \emph{$M$-augmenting path} for a matching $M$ in $G$ is a path $P$ of odd length with endpoints not covered by $M$ and whose edges are alternately out of $M$ and in $M$ (so necessarily start edge and end edge (which might coincide) are out of $M$). A new matching $M^+$ is obtained by adding the edges from $P$ to $M$, which were not in $M$, while removing the other edges of $P$ from $M$; we then have $\abs{M^+} = \abs{M} + 1$. A matching $M$ in a graph $G$ has an augmenting path if and only if $M$ is not maximum (see for example Theorem 16.1 in \cite{Schrijver2003CombOptA}). If $G$ is bipartite, then deciding whether $M$ has an augmenting path, and finding one if existent, can by done by breadth-first search in the directed bipartite graph naturally associated with the notion of augmenting paths (see Section 16.3 in \cite{Schrijver2003CombOptA}), and so this process takes time $O(\abs{E(G)}$. Using this process to construct a maximum matching, starting with the empty matching, takes time $O(\nu(G) \cdot \abs{E(G)})$ which is worse than the bound given above, however if a matching $M$ of ``reasonable size'' is already given, then the time $O((\nu(G) - \abs{M}) \cdot \abs{E(G)})$ it takes to construct a maximum matching, starting with $M$, might be better.

A \emph{vertex cover} of a graph $G$ is a set $T \sse V(G)$ of vertices such that every edge of $G$ is incident with (at least) one of the vertices in $T$; a vertex cover of minimal size is called a \emph{minimum vertex cover}, the size of a minimum vertex cover of $G$ is denoted by $\tau(G)$. For bipartite graphs $G$ we have $\tau(G) = \nu(G)$, and given a maximum matching of $G$, in time $O(\abs{E(G)})$ a minimum vertex cover can be computed (see Theorem 16.6 in \cite{Schrijver2003CombOptA}). A witness for $\nu(G) \ge k$ can always be given by a matching $M$ in $G$ with $\abs{M} \ge k$, while for bipartite graphs $G$ a witness for $\nu(G) \le k$ can always be given by a vertex cover $T$ with $\abs{T} \le k$. For our applications, witnesses using the notion of deficiency are more useful (they will yield autarkies in Lemmas \ref{lem:surpone} and \ref{lem:etabsurptwo}): Given a bipartite graph $G$ with bipartition $(A,B)$ and vertex cover $T$ of $G$, the set $A' := A \sm T$ has deficiency $\delta(A') \ge \abs{A} - \abs{T}$ (while given $A' \sse A$, the vertex cover $\Gamma(A') \cup (A \sm A')$ has size $\abs{A} - \delta(A')$).

\section{Generalised (multi-)clause-sets}
\label{sec:prelimgeneralcls}

In this section we review the notion of generalised multi-clause-sets. In Subsection \ref{sec:notmulticlausesets} we introduce ``generalised multi-clause-sets'' and ``generalised clause-sets'', while in Subsection \ref{sec:semantics} (partial) assignments and their operation on (multi-)clause-sets are discussed. Besides this substitution of truth values into variables we only need another special case of general substitution, namely the renaming of variables, which is treated in Subsection \ref{sec:renamingvariables}. Two different clauses by definition have a ``conflict'' if they do not have a common falsifying assignment, and some basic notions regarding the ``conflict graph'' of a (multi-)clause-set are introduced in Subsection \ref{sec:conflictstructure}. This introduction into ``syntax and semantics of generalised clause-sets'' is completed in Subsection \ref{sec:opsetsvar} with the discussion of various operations on (multi-)clause-sets $F$ regarding their variable structure (disregarding the different ``polarities'' of variables, i.e., when considering literals, then here only the underlying variables play a role).

For more background information, see \cite{Ku00g,Ku01a} for a general, axiomatic framework for ``generalised satisfiability problems'', while in Subsection 2.3 of \cite{Ku00g} generalised clause-sets are discussed, and in Section 2 of \cite{Ku2003e} boolean multi-clause-sets are considered (see also \cite{Ku2003c} for more information). In this article, when we speak of ``clause-sets'' then we always mean ``generalised clause-sets'', while clause-sets in the ``traditional'' sense are always qualified as ``boolean clause-sets''; however in lemmas, corollaries and theorems we speak of ``generalised clause-sets'' to ease independent access.

\subsection{Syntax: The notion of ``multi-clause-sets''}
\label{sec:notmulticlausesets}

A \textbf{literal} is a pair $(v,\ve)$ of a variable $v \in \Va$ and a value $\ve \in D_v$; we write $\bmm{\var(v,\ve)} := v$ and $\bmm{\val(v,\ve)} := \ve$. The set of all literals is denoted by \bmm{\Lit}, and for any $\Va' \sse \Va$ we write $\bmm{\Lit(\Va')} := \set{x \in \Lit : \var(x) \in \Va'}$ for the set of literals with variables from $\Va'$ (thus $\Lit(\Bva)$ is the set of boolean literals). For a partial assignment $\vp \in \Pass$ and a literal $(v, \ve)$ with $v \in \var(\vp)$ we set $\vp((v,\ve)) = 1$ if $\vp(v) \not= \ve$, while we set $\vp((v,\ve)) = 0$ if $\vp(v) = \ve$; thus a literal $(v,\ve)$ has the meaning ``$v$ shall \emph{not} get value $\ve$''. Accordingly a literal $(v, \ve)$ is often denoted by ``$v \not= \ve$''.

A \textbf{clause} $C$ is a finite set of literals not containing ``clashing literals'', that is for literals $x, y \in C$ with $x \not= y$ we have $\var(x) \not= \var(y)$. The set of all clauses is denoted by \bmm{\Cl}. For a clause $C$ we set $\var(C) := \set{\var(x) : x \in C}$, and for a set $\Va' \sse \Va$ we write $\bmm{\Cl(\Va')} := \set{C \in \Cl : \var(C) \sse \Va'}$ for the set of clauses with variables from $\Va'$ (thus $\Cl(\Bva)$ is the set of boolean clauses). The empty clause is denoted by $\bmm{\bot} \in \Cl$.

Given a clause $C$, we obtain the corresponding partial assignment $\bmm{\vp_C} \in \Pass$ as the partial assignment $\vp$ with $\var(\vp) = \var(C)$ and $\vp(v) = \ve$ for $(v, \ve) \in C$; on the other hand, given a partial assignment $\vp$, we obtain the corresponding clause $\bmm{C_{\vp}} \in \Cl$ as the clause $C$ with $\var(C) = \var(\vp)$ such that for $\vp(v) = \ve$ we have $(v, \ve) \in C$.  Using the representation of maps as ordered pairs of arguments and values, actually $\vp_C = C$ and $C_{\vp} = \vp$ (and thus $\Cl = \Pass$); explicitely said, a clause corresponds to the partial assignment which sets exactly the literals in the clause to false.\footnote{The motivation is, that with a partial assignment $\vp$ we restrict the search space, and in case the partial assignment $\vp$ is inconsistent with the clause-set $F$, then the clause $C_{\vp}$ can be ``learned'' (i.e., follows from $F$).}

A (finite) \textbf{multi-clause-set} is a map $F: \Cl \ra \NNZ$, assigning to each clause its number of occurrences, such that only for finitely many $C \in \Cl$ we have $F(C) \not= 0$, while a (finite) \textbf{clause-set} is a finite subset of $\Cl$. Clause-sets $F$ can be implicitly converted to multi-clause-sets by setting $F(C) := 1$ for $C \in F$ and $F(C) := 0$ otherwise. For a multi-clause-set $F$ the \textbf{underlying clause-set} \bmm{\tmc(F)} is defined as $\tmc(F) = \set{C \in \Cl : F(C) \not= 0}$. This conversion is only performed if necessary to apply an operation only defined for clause-sets. We have \bmm{C \in F} for a multi-clause-set $F$ iff $F(C) > 0$, and we set $\var(F) := \bc \set{\var(C) : C \in F}$. For a variable $v \in \Va$ we define $\bmm{\val_v(F)} := \set{\ve \in D_v \mb \ex\, C \in F : (v,\ve) \in C}$ as the set of occurring values. We have $\var(F) = \set{v \in \Va : \val_v(F) \not= \es}$. Finally the empty clause-set as well as the empty multi-clause-set is denoted by \bmm{\top}.

A clause-set $F$ has a \textbf{uniform domain} if $\fa\, v, w \in \var(F) : D_v = D_w$ holds. Boolean clause-sets have uniform domain, and every finite (generalised) clause-set has a ``domain-uniformisation'' by using the union of all relevant domains and adding unit-clauses to forbid unwanted domain elements; see Subsection \ref{sec:renamingvariables} for details.

We use the following complexity measures for multi-clause-sets $F$:
\begin{enumerate}\NAA
\item $\bmm{\#_{(v,\ve)}(F)} := \sum_{C \in F, (v,\ve) \in C} F(C) \in \NNZ$ measures the number of occurrences of a literal;
\item $\bmm{\#_v(F)} := \sum_{\ve \in D_v} \#_{(v,\ve)}(F) = \sum_{C \in F, v \in \var(C)} F(C) \in \NNZ$ measures the number of occurrences of a variable;
\item $\bmm{s_{(v,\ve)}(F)} := \sum_{\ve' \in D_v \sm \set{\ve}} \#_{(v,\ve)}(F) = \#_v(F) - \#_{(v,\ve)}(F) \in \NNZ$ measures the number of occurrences of literals with variable $v$ and value different from $\ve$ (this is the number of satisfied clauses when assigning value $\ve$ to $v$; see Subsection \ref{sec:semantics});
\item $\bmm{n(F)} := \abs{\var(F)} \in \NNZ$ measures the number of variables;
\item $\bmm{c(F)} := \sum_{C \in F} F(C) \in \NNZ$ measures the number of clauses;
\item $\bmm{\ell(F)} := \sum_{C \in F} F(C) \cdot \abs{C} = \sum_{v \in \var(F)} \#_v(F) \in \NNZ$ measures the number of literal occurrences.
\end{enumerate}
And for multi-clause-sets $F_1, F_2$ we use the following operations and relations:
\begin{enumerate}\NAA
\item the multi-clause-set \bmm{F_1 + F_2} is defined by
  \begin{displaymath}
    (F_1 + F_2)(C) := F_1(C) + F_2(C)
  \end{displaymath}
  for clauses $C$;
\item the multi-clause-set \bmm{F_1 \cup F_2} resp.\ \bmm{F_1 \cap F_2} is given by
  \begin{eqnarray*}
    (F_1 \cup F_2)(C) & := & \max(F_1(C), F_2(C))\\
    (F_1 \cap F_2)(C) & := & \min(F_1(C), F_2(C))
  \end{eqnarray*}
  for clauses $C$; if $F_1, F_2$ are clause-sets, then these operations coincide with the ordinary set operations;
\item if $F_2$ is a clause-set, then the multi-clause-set \bmm{F_1 \sm F_2} is defined by setting $(F_1 \sm F_2)(C) := 0$ for $C \in F_2$, while otherwise $(F_1 \sm F_2)(C) := F_1(C)$; if also $F_1$ is a clause-set, then $F_1 \sm F_2$ is the ordinary set operation;
\item the relation \bmm{F_1 \le F_2} holds if for all clauses $C$ we have $F_1(C) \le F_2(C)$; we use \bmm{F' \lneq F} for $F' \le F \und F' \not= F$; if $F_1, F_2$ are clause-sets, then $F_1 \le F_2 \Lra F_1 \sse F_2$;
\item $F_1$ is called a \textbf{sub-multi-clause-set} of $F_2$ if $F_1 \le F_2$ holds, while $F_1$ is called an \textbf{induced sub-multi-clause-set} of $F_2$ if $F_1 \le F_2$ and $\fa\, C \in F_1 : F_1(C) = F_2(C)$ holds; every sub-clause-set of a clause-set is induced;
\item if $F_2$ is a sub-multi-clause-set of $F_1$, then the multi-clause-set \bmm{F_1 - F_2} is defined via $(F_1 - F_2)(C) := F_1(C) - F_2(C)$ for clauses $C$.
\end{enumerate}

The set of all multi-clause-sets is denoted by \bmm{\Mcls}, the set of all clause-sets by \bmm{\Cls}, while for a set $\Va' \sse \Va$ of variables we use $\bmm{\Mcls(\Va')} := \set{F \in \Mcls : \var(F) \sse \Va'}$ and $\bmm{\Cls(\Va')} := \set{F \in \Cls : \var(F) \sse \Va'}$ (thus $\Mcls(\Bva)$ is the set of boolean multi-clause-sets, and $\Cls(\Bva)$ is the set of boolean clause-sets). If $\mc{C}$ is a set of multi-clause-sets and $f: \mc{C} \ra \RR$, then by \bmm{\mc{C}_{f \le b}} for some $b \in \RR$ we denote the set of all $F \in \mc{C}$ with $f(F) \le b$; analogously we define $\mc{C}_{f = b}$, $\mc{C}_{f \ge b}$ and so on. A special function usable here is $\sat: \Cls \ra \set{0,1}$ with $\sat(F) = 1 \Lra F \in \Sat$ (that is, $\sat$ is the characteristic function of the set of satisfiable clause-sets defined below); we can combine several such indices, and for typographical reasons we may use then for example $\Mcls^{g\le b'}_{f\le b}$.

\subsection{Semantics: The operation of partial assignments}
\label{sec:semantics}

Now we define the operation $*: \Pass \times \Mcls \ra \Mcls$ of $\Pass$ on multi-clause-sets, and the (derived) operation $*: \Pass \times \Cls \ra \Cls$ on clause-sets, which in both cases have the meaning of substituting values for variables and carrying out the resulting simplifications (viewing a clause as a disjunction of its literals, and a (multi-)clause-set as a conjunction of its clauses), with the only difference that in the case of clause-\emph{sets} contractions in the result are carried out (distinct clauses can become equal after a substitution). The case of clause-sets is reduced to the case of multi-clause-sets, using the explicit transformation $\tcm: \Cls \ra \Mcls$ of clause-sets into multi-clause-sets. For $F \in \Mcls$ and $\vp \in \Pass$ we define $\bmm{\vp * F} \in \Mcls$ by
\begin{displaymath}
  (\vp * F)(C) = \sum_{\ueber{C' \in \Cl}{\vp * \set{C'} = \set{C}}} F(C'),
\end{displaymath}
for $C \in \Cl$, where for a clause $C$ we set $\vp * \set{C} := \top \in \Cls$ if there exists a literal $x \in C$ with $\vp(x) = 1$, while otherwise we set $\vp * \set{C} := \set{C \sm C_{\vp}} \in \Cls$, i.e., we remove the falsified literals from $C$. And for $F \in \Cls$ we define $\vp * F \in \Cls$ as
\begin{displaymath}
  \vp * F := \tmc(\vp * \tcm(F)).
\end{displaymath}
We have here (where $F$ is a clause-\emph{set}) $\vp * F = \bc_{C \in F} \vp * \set{C}$. The effect on the basic measures of applying a partial assignment $\pao{v}{\ve}$ to $F \in \Mcls$ with $v \in \var(F)$ is given by
\begin{eqnarray*}
  n(\pao{v}{\ve} * F) & \le & n(F) - 1\\
  c(\pao{v}{\ve} * F) & = & c(F) - s_{(v,\ve)}(F).
\end{eqnarray*}
A clause-set $F \in \Cls$ is \textbf{satisfiable} if there exists a partial assignment $\vp \in \Pass$ with $\vp * F = \top$, while otherwise $F$ is \textbf{unsatisfiable}; the set of all satisfiable clause-sets is denoted by \bmm{\Sat}, the set of all unsatisfiable clause-sets by \bmm{\Usat}. A multi-clause-set $F \in \Mcls$ is called \textbf{minimally unsatisfiable} if $F$ is unsatisfiable, but every $F' \lneq F$ is satisfiable; obviously if $F$ is minimally unsatisfiable, then $F$ actually is a clause-set. The set of all minimally unsatisfiable clause-sets is denoted by \bmm{\Musat}.

It is useful to have some notations for the set of satisfying assignments (``models'') as well as for the set of falsifying assignments. For a finite $V \sse \Va$ let \bmm{\Pass(V)} be the set of $\vp \in \Pass$ with $\var(\vp) = V$. Note that we have 
\begin{displaymath}
  \abs{\Pass(V)} = \prod_{v \in V} \abs{D_v}.
\end{displaymath}
Now for a clause-set $F \in \Mcls$ and for a finite set $V$ of variables with $\var(F) \sse V$ let \bmm{\mods_V(F)} be the set of $\vp \in \Pass(V)$ with $\vp * F = \top$, while \bmm{\modf_V(F)} is the set of $\vp \in \Pass(V)$ with $\bot \in \vp * F$. Thus $F$ is satisfiable iff $\mods_V(F) \not= \es$; and for any clause $C$ with $\var(C) \sse V$ we have
\begin{displaymath}
  \abs{\modf_V(\set{C})} = \abs{\Pass(V \sm \var(C))} = \prod_{v \in V \sm \var(C)} \abs{D_v},
\end{displaymath}
that is, the falsifying assignments for a clause $C$ are uniquely determined on variables from $C$ and arbitrary elsewhere. At this point it might be useful to point out that for (multi-)clause-sets typically \emph{falsifying assignments} are significantly easier to handle than satisfying assignments (``those elusive idols'', since we are dealing with CNF; for DNF the situation is reversed). Obviously $\mods_V(F) \cap \modf_V(F) = \es$ and $\mods_V(F) \cup \modf_V(F) = \Pass(V)$. By definition we have
\begin{displaymath}
  \modf_V(F) = \bc_{C \in F} \modf_V(\set{C}).
\end{displaymath}
For clause-sets $F_1, F_2$ we write \bmm{F_1 \models F_2} (``$F_1$ implies $F_2$'') if for all $\vp \in \Pass$ with $\vp * F_1 = \top$ we have $\vp * F_2 = \top$ as well, and for clauses $C$ we write $F \models C$ instead of $F \models \set{C}$. Trivially $F$ is unsatisfiable iff $F \models \bot$. Note that $F_1 \models F_2$ holds iff for $V := \var(F_1) \cup \var(F_2)$ we have $\modf_V(F_2) \sse \modf_V(F_1)$. We call $F_1, F_2$ \textbf{equivalent} if $F_1 \models F_2$ and $F_2 \models F_1$ holds.

The basic laws for the operation of partial assignments on multi-clause-sets are as follows, using $F, F_1, F_2 \in \Mcls$ and $\vp, \psi \in \Pass$:
\begin{gather*}
  \epa * F = F\\
  \vp * \top = \top\\
  (\vp \circ \psi) * F = \vp * (\psi * F)\\
  \vp * (F_1 + F_2) = \vp * F_1 + \vp * F_2.
\end{gather*}
The first three laws hold also for the operation of partial assignments on clause-sets, while the composition of multi-clause-sets by addition is to be replaced by union, that is for $F_1, F_2 \in \Cls$ we have
\begin{displaymath}
  \vp * (F_1 \cup F_2) = \vp * F_1 \cup \vp * F_2
\end{displaymath}
(note that in general this does not hold for multi-clause-sets, where we also have defined union). Furthermore for a multi-clause-set $F$ and a clause-set $F'$ we have $\vp * (F \sm F') \ge (\vp * F) \sm (\vp * F')$.

\subsection{Renaming variables}
\label{sec:renamingvariables}

Consider a multi-clause-set $F$ and variables $v, w \in \Va$ (which might be equal) together with $h: D_v \ra D_w$ such that in case of $v \not= w$ we have $w \notin \var(F)$. Then \textbf{replacing $v$ by $w$ using $h$ in $F$} results in the multi-clause-set $F'$ where every occurrence of a literal $(v, \ve)$ is replaced by the literal $(w, h(\ve))$. The map $h$ here is called the \textbf{value transfer}; if $D_v \sse D_w$ and $h$ is unspecified, then the canonical injection is used.

Similarly, replacing $v$ by $w$ using $h$ in a partial assignment $\vp$, where in case of $v \not= w$ we have $w \notin \var(\vp)$, results in a partial assignment $\vp'$ with $\dom(\vp') = (\dom(\vp) \sm \set{v}) \cup \set{w}$ such that $\vp'(u) = \vp(u)$ for $u \in \dom(\vp') \sm \set{w}$, while $\vp'(w) = h(\vp(v))$. Here the value transfer needs to be specified only for the special value $\vp(v)$. If $v = w$, then we just speak of \textbf{flipping $v$ to $\ve$ in $\vp$} for $\ve = h(\vp(v))$.

The replacement of $v$ by $w$ using $h$ in $F$ is \textbf{injective}, if for literals $(v, \ve), (v, \ve')$ occurring in $F$ with $\ve \not= \ve'$ we have $h(\ve) \not= h(\ve')$. If $\abs{D_w} \ge \#_v(F)$, then there is always some $h: D_v \ra D_w$ such that replacing $v$ by $w$ in $F$ using $h$ is injective. For every injective $h$, replacing $v$ by $w$ in $F$ using $h$ is injective. Note that injective replacements alter the ``meaning'' (the set of models modulo isomorphism) exactly in the case where a non-pure variable (a variable such that all values occur in $F$; see Subsection \ref{sec:Reduction}) is rendered a pure variable by using a domain $D_w$ with $\abs{D_w} > \#_v(F)$. Special injective replacements are \textbf{renamings}, where $h$ is a bijection from $D_v$ to $D_w$. If we have a renaming of $v$ by $w$ using $h$ in $F$, resulting in $F'$, then we also have the renaming of $w$ to $v$ using $h^{-1}$ in $F'$, resulting in $F$. So the satisfying assignments for $F'$ here are exactly the satisfying assignments for $F$ where $v$ is replaced by $w$ using $h$. If several variables are renamed simultaneously, then we require that the same result is obtained by renaming single variables one after another (in some order).

 For a multi-clause-set $F$ a \textbf{domain-uniformisation} of $F$ is $F$ for $\var(F) = \es$, while otherwise we consider the domain $D := \bc_{v \in \var(F)} D_v$, rename all variables $v$ in $F$ to $v'$ with domain $D$, and add all unit clauses $\set{v' \not= \ve}$ for $\ve \in D \sm D_v$.

\subsection{Conflict structure}
\label{sec:conflictstructure}

A study of the ``combinatorics of conflicts'' for boolean clause-sets has been initiated with \cite{Ku2003c,Ku2003e} and continued with \cite{GalesiKullmann2003bHermitian,Kullmann2004c}. We generalise here only a very few simple notions used later in this article.

The \textbf{conflict graph} \bmm{\cg(F)} of a clause-set $F \in \Mcls$ has as vertices the clauses of $F$, and edges joining two vertices $C, D$ with a \textbf{clashing literal} between $C$ and $D$, that is, there is a literal $x \in C$ for which there exists a literal $y \in D$ with $\var(x) = \var(y)$ and $x \not= y$. While the conflict graph does not take multiple clashing literals into account, the \textbf{conflict multigraph} \bmm{\cmg(F)} contains as many parallel edges between two clauses as there are conflicts. A clause-set $F$ is called a \textbf{hitting clause-set} if the conflict graph of $F$ is a complete graph, and the \textbf{hitting degree} $\bmm{\hitdeg(F)} \in \NN$ of a hitting clause-set with at least two clauses is the maximum of the number of (parallel) edges joining two different vertices in the conflict multigraph of $F$. More specifically we call $F$ an \textbf{\bmm{r}-regular hitting clause-set} for $r \in \NNZ$ if every two different clauses in $F$ have exactly $r$ conflicts (thus if $F$ is $r$-regular hitting for $r \ge 1$, then $F$ is hitting), while a \textbf{regular hitting clause-set} is an $r$-regular hitting clause-set for some $r \ge 0$, and we denote the set of regular hitting clause-sets by \bmm{\Sclash}.\footnote{Regular hitting clause-sets were called ``uniform hitting clause-sets'' in \cite{Ku2003e}, but ``regular'' seems a better choice, since regularity refers to constant degree, while uniformity typically refers to constant hyperedge-size, and so should better be used for clause-sets of constant clause-size.}

\subsection{Three operations of sets of variables on multi-clause-sets}
\label{sec:opsetsvar}

Finally we consider various operations with sets of variables. The operation \bmm{V * F} is defined for finite $V \sse \Va$ and $F \in \Mcls$ via
\begin{displaymath}
  (V * F)(C) := \sum_{\ueber{C' \in \Cl}{V * C' = C}} F(C'),
\end{displaymath}
where for a clause $C$ we set $V * C := \set{x \in C : \var(x) \notin V} \in \Cl$. That is, $V * F$ is obtained from $F$ by crossing out all literal occurrences $x$ with $\var(x) \in V$. Two basic properties are
\begin{eqnarray*}
  \var(V * F) & = & \var(F) \sm V\\
  c(V * F) & = & c(F).
\end{eqnarray*}
 The operation \bmm{V * F} for $F \in \Cls$ is defined by
\begin{displaymath}
  V * F := \tmc(V * \tcm(F)) \in \Cls.
\end{displaymath}
We have here $V * F = \set{V * C : C \in F}$. The basic laws for $F, F_1, F_2 \in \Mcls$ and finite $V, V' \sse \Va$ are
\begin{gather*}
  \epa * F = F\\
  V * \top = \top\\
  (V \cup V') * F = V * (V' * F)\\
  V * (F_1 + F_2) = V * F_1 + V * F_2.
\end{gather*}
Again, the first three laws also hold for the operation of sets of variables on clause-sets, while for $F_1, F_2 \in \Cls$ we have
\begin{displaymath}
  V * (F_1 \cup F_2) = V * F_1 \cup V * F_2
\end{displaymath}
(again this does not hold for multi-clause-sets in general). 

We conclude by discussing the two basic ways of selecting parts of a multi-clause-set. By \bmm{F_V} we denote the induced sub-multi-clause-set of $F$ with $C \in F_V \Lra \var(C) \cap V \not= \es$; in other words, $F_V = F \sm \set{C \in F : \var(C) \cap V = \es}$. Basic properties are:
\begin{enumerate}\NAA
\item $F_{\es} = \top$ and $F_{\var(F)} = F \sm \set{\bot}$.
\item If $V_1 \sse V_2$, then $F_{V_1}$ is an induced sub-multi-clause-set of $F_{V_2}$.
\item $F_{V_1 \cup V_2} = F_{V_1} \cup F_{V_2}$.
\item For $v \in \Va$ we have $c(F_{\set{v}}) = \#_v(F)$.
\end{enumerate}
Additionally restricting the variables is achieved by
\begin{displaymath}
  \bmm{F[V]} := (\var(F) \sm V) * F_V = ((\var(F) \sm V) * F) \sm \set{\bot} \in \Mcls,
\end{displaymath}
which, in analogy to the same process for hypergraphs (using usually also the same notation, see for example \cite{Duchet1995Hypergraphs}), can be considered as the ``restriction of $F$ to $V$''. Basis properties are
\begin{enumerate}
\item $F[\es] = \top$ and $F[\var(F)] = F \sm \set{\bot}$.
\item $c(F[V]) = c(F_V)$, $\var(F[V]) \sse \var(F_V)$.
\item $\var(F[V]) = V$ for $V \sse \var(F)$.
\end{enumerate}
To summarise: We obtain $V * F$ from $F$ by keeping all clauses but removing those literals $x$ from them with $\var(x) \in V$, while we obtain $F_V$ from $F$ by removing those clauses $C$ from $F$ with $\var(C) \cap V = \es$ (while keeping all clauses intact); finally $F[V]$ is obtained from $F$ by first constructing $F_V$, and then crossing out all literal occurrences for literals $x$ where there exists a clause $C \in F$ with $\var(C) \cap V = \es$ and $\var(x) \in \var(C)$.

$F[V]$ is the formula derived from $F$ when we want to consider total assignments relative to the variable set $V$, and is basic for the theory of autarkies reviewed in the subsequent section, while $V * F$ and $F_V$ are fundamental constructions. An example for these operations:
\begin{enumerate}
\item Consider boolean variables $a,b,c$ (the domains of variables do not matter).
\item Let
  \begin{enumerate}
  \item $C_1 := \set{(a,0),(b,1),(c,0)}$,
  \item $C_2 := \set{(a,0), (b,0), (c,1)}$,
  \item $C_3 := \set{(a,1), (b,0), (c,1)}$,
  \item $C_4 := \set{(b,1), (c,1)}$.
  \end{enumerate}
\item Let $F := \sum_{i=1}^4 \set{C_i}$; $F$ corresponds to the boolean CNF 
  \begin{displaymath}
    (a \oder \ol{b} \oder c) \und (a \oder b \oder \ol{c}) \und (\ol{a} \oder b \oder \ol{c}) \und (\ol{b} \oder \ol{c}).
  \end{displaymath}
\item Now we have 
  \begin{enumerate}
  \item $F_{\set{a}} = \sum_{i=1}^3 \set{C_i}$,
  \item $\set{a} * F = \set{\set{(b,1),(c,0)}} + 2 \cdot \set{\set{(b,0),(c,1)}} + \set{\set{(b,1),(c,1)}}$,
  \item $F[\set{a}] = 2 \cdot \set{\set{(a,0)}} + \set{\set{(a,1)}}$.
  \end{enumerate}
\end{enumerate}

\section{Autarkies for generalised multi-clause-sets}
\label{sec:prelimAutarkies}

Of central importance to our work is the notion of \emph{autarkies}, such partial assignments which satisfy parts of the formula and leave the rest untouched.  \emph{Autarky systems} allow to tailor the notion of autarkies to special purposes. We review here the general properties of autarkies and autarky systems for generalised multi-clause-sets. See Section 3 in \cite{Ku01a} for a general theory of autarkies and autarky systems, while in Section 4 of \cite{Ku01a} autarky systems for generalised clause-sets have been discussed (easily generalised to autarky systems for generalised multi-clause-sets). General properties of autarkies for boolean clause-sets are thoroughly investigated in \cite{Ku98e}, Section 3, while autarky systems for boolean clause-sets have been introduced in \cite{Ku00f} (see Sections 4 and 8 there for the general theory). For boolean clause-sets, the handbook chapter \cite{Kullmann2007HandbuchMU} gives an overview on autarky theory, while the textbook \cite{Marek2009Satisfiability} introduces basic properties of autarkies.

\subsection{The notion of an autarky}
\label{sec:notionautarky}

A partial assignment $\vp \in \Pass$ is an \textbf{autarky} for $F \in \Mcls$ if one (and thus all) of the following four equivalent conditions is fulfilled:
\begin{enumerate}\NAA
\item for all clauses $C \in F$ we have $\var(\vp) \cap \var(C) \not= \es \Ra \vp * \set{C} = \top$;
\item $\fa\, F' \le F : \vp * F' \le F'$;
\item $\vp$ is a satisfying assignment for $F_{\var(\vp)}$;
\item $\vp$ is a satisfying assignment for $F[\var(\vp)]$.
\end{enumerate}
Obviously, $\vp$ is an autarky for $F$ iff $\vp$ is an autarky for $F \sm \set{\bot}$ iff $\vp$ is an autarky for the underlying clause-set. The set of all autarkies for $F$ is denoted by \bmm{\aut(F)}; it is $\aut(F)$ a sub-monoid of $\Pass$, containing all satisfying assignments for $F$, and $\aut(F) = \aut(\tmc(F))$. If $F' \le F$, then $\aut(F) \sse \aut(F')$, and for finite $V \sse \Va$ we have 
\begin{displaymath}
  \set{\vp \in \aut(V * F) : \var(\vp) \cap V = \es} = \set{\vp \in \aut(F) : \var(\vp) \cap V = \es},
\end{displaymath}
that is, the autarkies for $F$ which do not use variables from $V$ are exactly the autarkies for $V * F$ (which do not use $V$). Furthermore we have 
\begin{displaymath}
  \aut(F_1 + F_2) = \aut(F_1) \cap \aut(F_2).
\end{displaymath}
If $\vp \in \aut(F)$ and $\psi \in \aut(\vp * F)$, then $\psi \circ \vp \in \aut(F)$. An autarky $\vp \in \aut(F)$ is called \textbf{non-trivial} if $\var(\vp) \cap \var(F) \not= \es$ holds. $F$ is called \textbf{lean}, if $F$ has no non-trivial autarky; the set of all lean multi-clause-sets is denoted by \bmm{\Lean}. Every $F$ having only variables with trivial domain (i.e., one-element domains) is lean. A sum of lean multi-clause-sets again is lean. If $F$ is lean, so is $V * F$ for $V \sse \Va$.

An \textbf{autarky reduction} is a reduction $F \mapsto \vp * F$ for some non-trivial autarky $\vp$ for $F$ (note that $\vp * F$ is satisfiability equivalent to $F$). Autarky reduction is terminating and confluent (generalising Lemma 4.1 in \cite{Ku00f}, a special case of Lemma 3.7 in \cite{Ku01a}), and thus the result of iterated autarky reductions until no further reductions are possible is uniquely determined; we denote it by $\bmm{\nv_{\aut}(F)} \le F$. The operator $\bmm{\na} := \nv_{\aut}$ is a ``kernel operator'', that is, $\na(F) \le F$, $\na(\na(F)) = \na(F)$, and $F_1 \le F_2 \Ra \na(F_1) \le \na(F_2)$; furthermore $\na(F)$ is satisfiability equivalent to $F$, and $\na(F) = \top$ iff $F \in \Sat$. We have $\na(F) \in \Lean$, and $\na(F)$ is called the \textbf{lean kernel} of $F$; $F$ is lean iff $\na(F) = F$. There exists an autarky $\vp \in \aut(F)$ with $\na(F) = \vp * F$ (while for all $\vp \in \aut(F)$ we have $\na(F) \le \vp * F$). $\na(F)$ is the largest lean sub-multi-clause-set of $F$.

An \textbf{autark sub-multi-clause-set} of $F$ is an induced sub-multi-clause-set $F'$ of $F$, such that there exists an autarky $\vp \in \aut(F)$ so that for $C \in F$ we have $C \in F'$ iff $\vp * \set{C} = \top$ (note that in this case we have $F' = F_{\var(\vp)}$). The set of autark sub-multi-clause-sets of $F$ is closed under union, and contains the smallest element $\top$ and the largest element $F \sm \na(F)$. A fundamental observation is that $F' \le F$ is an autark sub-multi-clause-set of $F$ iff there is $V \sse \var(F)$ with $F_V = F'$ and $F[V] \in \Sat$.

The relation between the lean kernel of $F$ and the largest autark sub-multi-clause-set of $F$ can be summarised as follows: For $F \in \Mcls$ there exist induced sub-multi-clause-sets $F_1, F_2 \le F$ with $F_1 + F_2 = F$, such that $F_1$ is lean, while $\var(F_1) * F_2$ is satisfiable; in this decomposition $F_1, F_2$ are uniquely determined, namely $F_1 = \na(F)$ is the largest lean sub-multi-clause-set (the lean kernel), while $F_2$ is the largest autark sub-multi-clause-set; furthermore we have $F_2 = F_{\var(F) \sm \var(F_1)}$, and thus $\var(F_1) * F_2 = F[\var(F) \sm \var(F_1)]$.

For an example consider variables $a,b,c,d$ with $D_a = D_b = \set{0,1,2}$ and $D_c = D_d = \set{0,1}$, and consider the clause-set 
\begin{eqnarray*}
  F & := & F_1 \cup F_2\\
  F_1 & := & \set{\set{a \not= 0, b \not= 0},\, \set{a \not= 0, b \not= 1},\, \set{a \not= 0, b \not= 2},\, \set{a \not= 1},\, \set{ a \not= 2}}\\
  F_2 & := & \set{\set{a \not= 0, c \not= 0, d \not= 1},\, \set{b \not= 0, c \not= 1, d \not= 0}}.
\end{eqnarray*}
To see whether there is an autarky for $F$ invoking exactly one variable we check satisfiability of $F[\set{v}]$ for $v \in \set{a,b,c,d}$; we see that all these clause-sets are unsatisfiable (e.g., $F[\set{c}] = \set{\set{c \not= 0},\, \set{c \not= 1}}$), and so the smallest non-trivial autarky for $F$ (if there is any) must involve at least two variables. Now $F[\set{c,d}] = \set{\set{c \not= 0,d \not= 1},\set{c \not= 1,d \not= 0}} \in \Sat$, and thus the two partial assignments $\pab{c,d \ra 0}, \pab{c,d \ra 1}$ are autarkies for $F$; applying any of them yields $F_1$, which is lean ($F_1$ actually is minimally unsatisfiable), and thus $F_1$ is the lean kernel of $F$, while $F_2$ is the largest autark sub-clause-set of $F$.

\subsection{Autarky systems}
\label{sec:Autarkysystems}

After having reviewed the general facts for autarkies for generalised multi-clause-sets, we now consider ``autarky systems''. The motivation is that instead of (computationally hard) general autarkies we want to consider more feasible, restricted notions of autarkies. Under mild assumptions on these restricted autarkies all the above facts carry over (in generalised form). A definition of ``autarky systems'' for abstract ``systems with partial instantiation'' has been given in \cite{Ku01a}; the monoid $(\Pass, \circ, \epa)$ together with the partial order $(\Mcls, \le, \top)$ with least element and together with the operation $*$ of $\Pass$ on $\Mcls$ fulfils all the axioms required in Section 3 of \cite{Ku01a}, and thus all the general results there on autarky systems can be carried over. Again we refer to the handbook chapter \cite{Kullmann2007HandbuchMU} for more information.

An \textbf{autarky system} for generalised multi-clause-sets is a map $\A$, which assigns to every $F \in \Mcls$ a sub-monoid $\A(F)$ of $\aut(F)$, such that for $F_1 \le F_2$ we have $\A(F_2) \sse \A(F_1)$. The elements of $\A(F)$ are called \textbf{\bmm{\A}-autarkies} for $F$. Further possible restrictions on $\A$ are expressed by the following notions:
\begin{enumerate}\NAA
\item $\A$ is \textbf{iterative}, if for $\vp \in \A(F)$ and $\psi \in \A(\vp * F)$ we always have $\psi \circ \vp \in \A(F)$.
\item $\A$ is called \textbf{standardised}, if for a partial assignment $\vp \in \Pass$ we have $\vp \in \A(F)$ iff $\vp \mb \var(F) \in \A(F)$ (where $\vp \mb \var(F)$ is the restriction of the map $\vp$ to the domain $\var(\vp) \cap \var(F)$). (Remark: Thus for a standardised autarky system $\A$ all partial assignments $\vp$ with $\var(\vp) \cap \var(F) = \es$ are (trivial) $\A$-autarkies for $F$. In \cite{Ku01a} only the direction ``$\vp \in \A(F) \Ra \vp \mb \var(F) \in \A(F)$ was required, but now it seems more systematic to me to require also the other direction.)
\item $\A$ is \textbf{\bmm{\bot}-invariant}, if always $\A(F) = \A(F + \set{\bot})$ holds (in \cite{Ku01a,Ku00f} this was called ``normal'').
\item $\A$ is \textbf{invariant under variable elimination}, if for finite $V \sse \Va$ we always have $\set{\vp \in \A(V * F) : \var(\vp) \cap V = \es} = \set{\vp \in \A(F) : \var(\vp) \cap V = \es}$.
\item $\A$ is \textbf{invariant under renaming}, if for every $F'$ obtained from $F$ by renaming $v$ to $w$ using $h$ (recall Subsection \ref{sec:renamingvariables}) and for every autarky $\vp \in \A(F)$ we have $\vp' \in \A(F')$ for the partial assignment $\vp'$ obtained from $\vp$ by renaming $v$ to $w$ using $h$.
\item $\A$ is \textbf{stable for unused values}, if for $\vp \in \A(F)$, $v \in \dom(\vp)$ and for $\ve \in D_v$ such that none of the two literals $(v, \vp(v)), (v, \ve)$ occurs in $F$, also $\vp' \in \A(F)$ holds, where $\vp'$ is obtained from $\vp$ by flipping $v$ to $\ve$.
\end{enumerate}
An autarky system $\A$ is called \textbf{normal}, if it fulfils these six criteria, that is, if it is iterative, standardised, $\bot$-invariant, invariant under variable elimination, invariant under renaming and stable for unused values. Considering the boolean case (where stability for unused values is covered by the standardisation condition, while invariance under renaming was not considered), in \cite{Ku01a,Ku00f} ``normal autarky systems'' were called ``strong autarky systems'', but meanwhile the above properties turned out not to be so strong but quite ``normal'', while ``ab-normality'' is a defect which can be repaired; see for example Lemma 8.4 in \cite{Ku00f}, which can be generalised to generalised clause-sets. Examples for normal autarky systems are the smallest standardised autarky system $F \in \Mcls \mapsto \set{\vp \in \Pass: \var(\vp) \cap \var(F) = \es}$ and the largest autarky system $F \in \Mcls \mapsto \aut(F)$. In this article our main interest is in normal autarky systems, and we don't investigate the detailed relations between the above notions and the other properties of autarky systems, but we will state general results either for all autarky systems or for all normal autarky systems.

Consider an autarky system $\A$.
\begin{itemize}
\item An \textbf{\bmm{\A}-reduction} is a reduction $F \mapsto \vp * F$ for some non-trivial $\vp \in \A(F)$.
\item Since multi-clause-sets have finite variable sets, $\A$-reduction is terminating, and thus by Lemma 3.7 in \cite{Ku01a} $\A$-reduction is confluent, and the result of applying $\A$-reductions as long as possible is uniquely determined, yielding a normal form $\bmm{\nA(F)} \le F$. As before, the operator $\nA$ is a kernel operator, that is
  \begin{enumerate}
  \item $\nA(F) \le F$, 
  \item $\nA(\nA(F)) = \nA(F)$,
  \item $F_1 \le F_2 \Ra \nA(F_1) \le \nA(F_2)$.
  \end{enumerate}
\item Multi-clause-sets $F$ with $\nA(F) = \top$ are called \textbf{\bmm{\A}-satisfiable}, while in case of $\nA(F) = F$ we call $F$ \textbf{\bmm{\A}-lean}; the set of all $\A$-satisfiable multi-clause-sets is denoted by \bmm{\SatA}, the set of all $\A$-lean multi-clause-sets by \bmm{\LeanA}.
\item $F$ is $\A$-lean iff $\A(F)$ contains no non-trivial autarky. The learn kernel $\nA(F)$ is the largest $\A$-lean sub-multi-clause-set of $F$. A sum of $\A$-lean multi-clause-sets again is $\A$-lean.
\item Finally $F$ is called \textbf{minimally \bmm{\A}-unsatisfiable}, if $F$ is not $\A$-satisfiable, but every $F' \le F$ with $F' \not= F$ is $A$-satisfiable, while $F$ is called \textbf{barely $\A$-lean} if $F$ is $\A$-lean, but every $F' \le F$ with $c(F') = c(F) - 1$ is not $A$-lean; for more on these two notions see \cite{Kullmann2007Balanciert}, while in this article we will consider only some basic properties of minimal $\A$-unsatisfiability.
\item If $F$ is  minimally $\A$-unsatisfiable, then $F$ is $\A$-lean, and for $F \not= \set{\bot}$ it is $F$ also barely $\A$-lean. If $\A$ is the full autarky system, then minimal $\A$-unsatisfiability is just (standard) minimal unsatisfiability.
\end{itemize}
For the remainder of this subsection now assume that the autarky system $\A$ is normal. Then $F$ is $\A$-satisfiable iff there exists $\vp \in \A(F)$ with $\vp * F = \top$. More generally, there always exists $\vp \in \A(F)$ with $\vp * F = \nA(F)$. If $F$ is $\A$-lean, then so is $V * F$ and $F[V]$ for finite $V \sse \Va$. The $\A$-autark sub-multi-clause-sets of $F$, i.e., those multi-clause-sets $F'$ where there is $\vp \in \A(F)$ with $F' = F_{\var(\vp)}$, are exactly those $F_V$ for some $V \sse \var(F)$ where $F[V]$ is $\A$-satisfiable. The set of $\A$-autark sub-multi-clause-sets of $F$ is closed under union, and contains the smallest element $\top$ and the largest element $F \sm \nA(F)$. As before, the relation between the $\A$-lean kernel of $F$ and the largest $\A$-autark sub-multi-clause-set of $F$ can be summarised as follows:

\begin{lem}\label{lem:autdecomp}
  For $F \in \Mcls$ there exist induced sub-multi-clause-sets $F_1, F_2 \le F$ with $F_1 + F_2 = F$, such that $F_1$ is $\A$-lean, while $\var(F_1) * F_2$ is $\A$-satisfiable. In this decomposition, $F_1, F_2$ are uniquely determined, namely $F_1 = \nA(F)$ is the largest $\A$-lean sub-multi-clause-set (the $\A$-lean kernel), while $F_2$ is the largest $\A$-autark sub-multi-clause-set. Furthermore we have $F_2 = F_{\var(F) \sm \var(F_1)}$, and thus $F' := \var(F_1) * F_2 = F[\var(F) \sm \var(F_1)]$; the multi-clause-set $F'$ is $\A$-satisfiable, and every $\A$-satisfying assignment $\vp$ for $F'$ (with $\var(\vp) \sse \var(F') = \var(F_2) \sm \var(F_1)$) is an $\A$-autarky for $F$ with $\vp * F = F_1$ and $F_{\var(\vp)} = F_2$.
\end{lem}

We finish our review on autarkies and autarky systems by generalising Lemma 8.6 in \cite{Ku00f}. The proof can be literally transferred to our generalised context, and thus is not reproduced here.
\begin{lem}\label{lem:vonschlanknachred}
  Let $\A$ be a normal autarky system. Given decision of membership in $\LeanA$ as an oracle, the normal form $F \mapsto \nA(F)$ for $F \in \Mcls$ can be computed in polynomial time as follows:
  \begin{enumerate}
  \item\label{item:ersterSch} If $F \in \LeanA$ then output $F$.
  \item Let $\var(F) = \set{v_1, \dots, v_{n(F)}}$.
  \item\label{item:letzterSch} Since $\epa * F = F \notin \LeanA$ and $\var(F) * F = c(F) \cdot \tcm(\set{\bot}) \in \LeanA$ holds, there is an index $1 \le i \le n(F)$ with
    \begin{displaymath}
      \set{v_1, \dots, v_{i-1}} * F \notin \LeanA \ \text{ and } \ \set{v_1, \dots, v_i} * F \in \LeanA.
    \end{displaymath}
    Replace $F$ by the induced sub-multi-clause-set of $F$ given by the clauses of $F$ not containing variable $v_i$, and go to Step \ref{item:ersterSch}.
  \end{enumerate}
  The output of this procedure is $\na(F)$. If we use $V$ for the set of variables $v_i$ selected in Step \ref{item:letzterSch}, then there exists an autarky $\vp$ for $F$ with $\var(\vp) = V$, and $F_V$ is the largest autark subset of $F$ (note $F$ is the original input). Thus an autarky $\vp$ for $F$ with $\vp * F = \na(F)$ is obtained by an $\A$-satisfying assignment $\vp$ for $F[V]$ (with $\var(\vp) \sse V$).
\end{lem}
The idea behind the algorithm of Lemma \ref{lem:vonschlanknachred} is that we want to find a variable $v$ such that there exists an autarky $\vp$ for $F$ with $v \in \var(\vp)$; if there is no such variable, then $F$ is lean while otherwise we can eliminate all clauses from $F$ containing variable $v$. Now the variable $v_i$ selected in Step \ref{item:letzterSch} must be such a variable: Consider a non-trivial autarky $\vp_i$ for $F_i := \set{v_1, \dots, v_{i-1}} * F$ with $\var(\vp_i) \sse \var(F_i)$. Since $\set{v_i} * F_i$ is lean, $v_i \in \var(\vp_i)$ must be the case, while $\vp_i$ is an autarky also for $F$. For more on such reductions (for boolean clause-sets and the full autarky system) see \cite{KullmannMarekTruszczynski2007Autarkien,Marek2009Satisfiability}, where it is also shown, given an oracle for $\A$-leanness decision, how to find a non-trivial autarky (if existent), which through iteration yields an autarky $\vp \in \A(F)$ with $\vp * F = \nA(F)$. This works via addition of unit-clauses, similar to the usual self-reduction, but now used together with crossing out of variables as above, obtaining a minimal non-trivial autarky (if existent).

\section{Resolution}
\label{sec:Resolution}

In this section we discuss the fundaments of resolution for generalised clause-sets. In Subsection \ref{sec:resolutionrule} the obvious generalisation of the boolean resolution rule is discussed, deriving a (generalised) resolvent-clause from (generalised) parent-clauses. The duality between resolution and autarkies is presented in Subsection \ref{sec:dualresaut}, showing that the clauses satisfiable by some autarky are exactly those not usable in any resolution refutation. Finally in Subsection \ref{sec:DPres} the special case of ``Davis-Putnam resolution'' is considered, due to its importance for reductions.

\subsection{The resolution rule}
\label{sec:resolutionrule}

 Let us consider for the moment the most general form of (variable-based) CNF-clause-sets, that is, negative signed clause-sets (recall the discussion in Subsection \ref{sec:introsignedformulas}), where the literals are of the form ``$v \notin S$'' for some $S \sse D_v$ (so literals are unary constraints, but presented by their sets of falsifying values).\footnote{If only one sign is present (``monosigned''), then for CNF the negative sign is appropriate, and for DNF the positive sign, since CNFs \emph{ex}clude total assignments, while DNFs \emph{in}clude them.} Representing literals ``$v \notin S$'' by pairs $(v,S)$, the most general resolution rule allows for given negative signed clauses $C_1, \dots, C_m$ ($m \ge 0$) with $x_i = (v,S_i) \in C_i$ for $i \in \tb 1m$ the formation of the resolvent $R := (\bc_{i \in \tb 1m} C_i \sm \set{x_i}) \cup \set{(v,\bc_{i \in \tb 1m} S_i)}$, a negative signed clause collecting all literals from the $C_i$ except of the $x_i$, and including additionally the literal on variable $v$ whose (negative) sign is the union of the (negative) signs of the $x_i$; note that $(v,\es)$ represents truth value $1$ (is always true), while $(v,D_v)$ represents truth value $0$ (is always false). This most powerful resolution rule captures immediately the power of resolution on the direct boolean translation (which needs to use both ALO and AMO clauses here), and in fact the latter can also polynomially simulate the former. As discussed in Subsection \ref{sec:introsignedformulas}, also when starting with (just) generalised clause-sets, using the general resolution rule for signed clause-sets can yield polynomial-size refutations where ordinary resolution on generalised clause-sets (as discussed in this section), which stays inside the class of generalised clause-sets, has only exponential-size refutations. However in this article we consider resolution mainly as a ``combinatorial tool'', and then (restricted) resolution for generalised clause-sets is appropriate.

The resolution rule for generalised clause-sets is well-known. A thorough study is approached in \cite{Ku00g}, where actually resolution is considered for general ``fipa-systems'' (systems with finite instantiation by partial assignments), by reducing resolution for such axiomatic systems to resolution for generalised clause-sets through the use of ``no-goods'', i.e., out of the general system we get the clauses $C$ belonging to the resolution refutation as clauses $C_{\vp}$ associated with such partial assignments $\vp$ which led to the contradictions. We now review the most basic notions, with special emphasise on the fundamental duality with the notion of autarkies. A technical comment: For autarky systems the number of occurrences of a clause in a multi-clause-set might make a difference (as it is the case for matching autarkies introduced in the Section \ref{sec:autviamatch}), however for all known resolution systems we do not need this distinction, and thus only (generalised) clause-sets are considered for resolution (that is, if multi-clause-sets $F \in \Mcls$ are to be treated, then they are automatically ``downcast'' to the underlying clause-set $\tmc(F)$).

Consider a variable $v \in \Va$. ``Parent clauses'' $C_1, \dots, C_{\abs{D_v}}$ are called \textbf{resolvable with resolution variable $v$}, if $\val_v(\set{C_1, \dots, C_{\abs{D_v}}}) = D_v$ and the \textbf{resolvent} $R := \bc_{i=1}^{\abs{D_v}} \set{v} * C_i$ actually is a clause (contains no clashing literals), that is, whenever there are literals $x \in C_i$, $y \in C_j$ for some $i, j \in \tb 1{\abs{D_v}}$ with $x \not= y$ and $\var(x) = \var(y)$, then $\var(x) = v$ must be the case. Resolution is a complete and sound refutation system; see for example Corollary 5.9 in \cite{Ku00g}, where, translating branching trees into resolution trees, the existence of a resolution tree with at most $(\max_{v \in \var(F)} \abs{D_v})^{n(F)}$ many leaves for unsatisfiable generalised clause-sets $F$ is shown. Also stated in \cite{Ku00g} is the (well-known) ``strong completeness'' of resolution, that is, for a multi-clause-set $F \in \Mcls$ and a clause $C \in \Cl$ we have $F \models C$ iff there exists a resolution tree with axioms from $F$ deriving a clause $C' \sse C$.

\subsection{The duality between resolution and autarkies}
\label{sec:dualresaut}

In Theorem 3.16 in \cite{Ku98e} the ``resolution-autarky duality theorem'' was shown for boolean clause-sets, namely that the lean kernel of a clause-set $F$ consists exactly of all clauses $C \in F$ which can be used in some resolution refutation of $F$.\footnote{Where the resolution refutation may not contain ``dead ends'', which can be most easily enforced by considering only resolution \emph{trees}.} A predecessor is \cite{Ge98}, where it was shown that when a tableau-calculus based (boolean) SAT solver gets ``stuck'', i.e., can not find a refutation based on the current top-clause, then an autarky satisfying this top-clause can be extracted, while clauses satisfied by some autarky can not participate in any tableau-refutation.\footnote{We remark that \cite{Ok98} extends this approach by parallelisation, exploiting the observation (made there first) that the combination of autarkies yields again an autarky.} In \cite{Ku01a}, Theorem 4.1, the generalisation of the duality-theorem to generalised clause-sets is stated, but without a proof, which we now outline. Consider the set $U(F)$ of clauses $C \in F$ for which there exists a tree resolution refutation of $F$ using $C$ as an axiom. The direction, that a clause $C \in F \sm \na(F)$ can not be used in tree resolution refutations of $F$ (i.e., $U(F) \sse \na(F)$), is easily proved by induction: an autarky of $F$ satisfying $C$ is an autarky for all clauses in the tree and satisfies all clauses derived from $C$. For the reverse direction the main technical lemma is, that for each variable $v \in \var(U(F))$ and each $\ve \in D_v$ the unit-clause $\set{(v,\ve)}$ can be derived from $U(F)$ by resolution (this is a little proof-theoretic exercise; see Lemma 3.14 in \cite{Ku98e} for the boolean case). Now it follows that $F \sm U(F)$ is an autark sub-clause-set of $F$, since if the clause-set $\var(U(F)) * (F \sm U(F))$ would be unsatisfiable, then there would be a tree resolution refutation $T$ of $\var(U(F)) * (F \sm U(F))$, where the axioms of $T$ could be derived from the clauses in $F \sm U(F)$ and the clauses in $U(F)$ by the above technical lemma, and thus we could construct a tree resolution refutation involving some clause of $F \sm U(F)$, contradicting the definition of $U(F)$ (compare with Lemma 3.15 in \cite{Ku98e} for the boolean case). That $F \sm U(F)$ is an autark sub-clause-set of $F$ means $\na(F) \sse U(F)$, and altogether we have shown:
\begin{thm}\label{thm:ResAut}
  For any generalised clause-set $F \in \Cls$ the lean kernel $\na(F)$ equals the set $U(F)$ of clauses of $F$ usable in some (tree) resolution refutation of $F$. So $F$ is lean if and only if $F = U(F)$, that is, if every clause of $F$ can be used in some (tree) resolution refutation of $F$.
\end{thm}
As shown in Section 6 of \cite{Ku01a}, Theorem \ref{thm:ResAut} yields an algorithm for computing $\na(F)$ by using ``intelligent backtracking solvers'', which on unsatisfiable instances can return the set of variables used in some resolution refutation of the input. Crossing out these variables from the input, removing the empty clause obtained, and repeating this process, we finally obtain a satisfiable clause-set $F^*$, and now any satisfying assignment $\vp$ for $F^*$ with $\var(\vp) \sse \var(F^*)$ is an autarky for $F$ with $\vp * F = F \sm \na(F)$. See \cite{KullmannLynceSilva2005Autarkies} for more details on this computation of the lean kernel (in \cite{KullmannLynceSilva2005Autarkies} only boolean clause-sets are considered, but based on the results of the present article, all (mathematical) results can be generalised in the natural way).

\begin{corol}\label{cor:polytimeLKern}
  Consider a class $\mc{C} \sse \Cls$ of generalised clause-sets, such that for all unsatisfiable $F \in \mc{C}$ and all $V \sse \Va$ we have $F[V] \in \mc{C}$. Assume furthermore that there is an algorithm running for inputs $F \in \mc{C}$ in polynomial time, which either computes a satisfying assignment for $F$ or computes the set of variables used in some tree resolution refutation of $F$. Then for inputs $F \in \mc{C}$ the lean kernel $\na(F)$ is computable in polynomial time.
\end{corol}

A general source of classes $\mc{C}$ as required in Corollary \ref{cor:polytimeLKern} are the classes $\mc{G}_k(\mc{U}, \mc{S})$ for levels $k \in \NNZ$ and suitable oracles $\mc{U}$ for unsatisfiability and $\mc{C}$ for satisfiability, as introduced in \cite{Ku00g}. In this way for example we get poly-time computation of the lean kernel for generalised Horn clause-sets (allowing non-boolean variables as well as considering higher levels $k$ of ``Horn-structures''). For further information on algorithmic problems related to autarkies see \cite{Kullmann2007HandbuchMU}.

\subsection{DP-resolution}
\label{sec:DPres}

We conclude this section by defining the \textbf{Davis-Putnam operator $\dpl$} for generalised clause-sets. Consider a clause-set $F \in \Cls$ and a variable $v \in \var(F)$. Let $F_v$ be the set of all resolvents of parent clauses in $F$ with resolution variable $v$. Now we set $\bmm{\dpi{v}(F)} := \set{C \in F : v \notin \var(C)} \cup F_v$. From the completeness results for (generalised) resolution in \cite{Ku00g} it follows immediately that $\dpi{v}(F)$ is satisfiability equivalent to $F$, and that $F$ is unsatisfiable if and only if by repeated applications of the Davis-Putnam operator we finally obtain the clause-set $\set{\bot}$ (while for satisfiable $F$ finally we will obtain the clause-set $\top$). We can also generalise Lemma 7.6 in \cite{KuLu98} about the commutativity of the Davis-Putnam operator, that is, if $G_1$ is the result of applying first $\dpi{v_1}$, then $\dpi{v_2}$, ..., and finally applying $\dpi{v_m}$, while $G_2$ is the result of applying first $\dpi{v_{\pi(1)}}$, then $\dpi{v_{\pi(2)}}$, ..., and finally applying $\dpi{v_{\pi(m)}}$, for some permutation $\pi \in S_m$, then after elimination of subsumed clauses in $G_1$ and $G_2$ (see the following subsection) $G_1$ becomes equal to $G_2$. It follows that for any set of variables $V$ the operator $\dpi{V}$, computed by running through the variables of $V$ in some order, is uniquely determined up to subsumption reduction in the result. We always have $\dpi{V}(F) = \dpi{V}(F_V) \cup (F \sm F_V)$. If for some $V \sse \var(F)$ we have $\dpi{V}(F_V) = \top$, then $F$ and $F \sm F_V$ are satisfiability equivalent, generalising the elimination of autark sub-clause-sets: If $\vp \in \aut(F)$, then $\dpi{\var(\vp)}(F_{\var(\vp)}) = \top$, while the reverse direction need not hold, as the example $F = \set{\set{v,a},\set{\ol{v},\ol{a}}} \cup F'$, $v \notin \var(F')$, with $V = \set{v}$ shows (for boolean variables). We see that the Davis-Putnam operator, whose application for generalised clause-sets is semantically the same as existential quantification (that is, the results are logically equivalent), yields more powerful reductions than autarky reduction, but this at the cost of potential exponential space usage.

\section{Reductions}
\label{sec:Reduction}

In this section we review the basic polynomial-time reduction concepts. For a thorough discussion in the boolean case, see \cite{KuLu98}. We consider only clause-sets $F \in \Cls$, but all results are easily generalised to multi-clause-sets. In Subsection \ref{sec:Basicreductions} the most basic reductions are considered, and in Subsection \ref{sec:DPreductions} reductions by means of DP-resolution are presented.

\subsection{Basic reductions}
\label{sec:Basicreductions}

The most basic reduction (by which we mean a satisfiability-equivalent transformation, simplifying the clause-set in some sense) is \textbf{subsumption elimination}, the elimination of subsumed clauses, i.e., the transition $F \ra F \sm \set{C}$ for $C \in F$ in case there exists $C' \in F$ with $C' \subset C$. Iterated elimination of subsumed clauses is confluent, and thus the result of applying subsumption elimination as long as possible is uniquely determined (namely it is the set of all minimal clauses of $F$); if $F$ has no subsumed clauses, then we call $F$ \textbf{subsumption-free}.

The next reduction can be called the \textbf{trivial-domain reduction}: If there exists $v \in \var(F)$ with $\abs{D_v} = 1$ (we call such a variable \textbf{trivial}), then for $D_v = \set{\ve}$ reduce $F \mapsto \pao{v}{\ve} * F$ (that is, all literal occurrences with underlying variable $v$ are removed).

Elimination of ``pure literals'' is now better called \textbf{elimination of pure variables}: If there is $v \in \var(F)$ with $\abs{\val_v(F)} < \abs{D_v}$, that is, one of the values of $v$ is not used in $F$, then for some $\ve \in D_v \sm \val_v(F)$ reduce $F \mapsto \pao{v}{\ve} * F$. This is the basic form of a \textbf{pure autarky} as mentioned in Subsection 4.4 of \cite{Ku01a}.

\textbf{Unit-clause elimination} for generalised clause-sets is less powerful than in the boolean case: If $F$ contains a unit-clause $\set{(v,\ve)} \in F$, then in case of  $D_v = \set{\ve}$ by trivial-domain reduction we conclude that $F$ is unsatisfiable, but otherwise we can only conclude that value $\ve$ is to be excluded from the domain of $v$, and in general we can not eliminate the variable $v$. It seems most convenient here to include trivial-domain reduction into unit-clause elimination (so that we properly generalise the boolean case); in our context, where we fixed the domain of each variable (and thus renaming is needed to achieve a change of domain), \textbf{unit-clause propagation} for (generalised) clause-sets $F$ is then the following procedure:
\begin{enumerate}
\item\label{item:UCP1} Apply trivial-domain reduction to $F$ to eliminate all trivial variables.
\item If $\bot \in F$, then reduce $F$ to $\set{\bot}$.
\item\label{item:UCP3} For $\set{(v,\ve)} \in F$ eliminate all clauses containing the literal $(v,\ve)$ from $F$, and replace variable $v$ in the remaining occurrences by a new variable $v'$ with $D_{v'} = D_v \sm \set{\ve} \not= \es$.
\item\label{item:UCP4} Repeat Steps \ref{item:UCP1} - \ref{item:UCP3} until all trivial variables and all unit-clauses have been eliminated.
\end{enumerate}
So after unit-clause propagation every variable has a domain with at least two elements and, except of the case $F = \set{\bot}$, every clause contains at least two literals. Unit-clause propagation for boolean clause-sets is confluent (the final result does not depend on the order and choice of single reduction steps), while unit-clause propagation for arbitrary clause-sets is confluent \emph{modulo renaming}. Generalising the well-known linear time algorithm for unit-clause propagation in the boolean case, this normal form (the result of unit-clause propagation) can be computed in linear time.

Considering clauses $C \in \Cl$ as  constraints of scope $\var(C)$ (see \cite{Dechter2003Constraints}, Subsection 2.1.1), and thus clause-sets $F \in \Cls$ as \emph{constraint networks} (or \emph{constraint satisfaction problems}), the basic reduction concepts for constraint satisfaction problems are translated as follows:
\begin{enumerate}
\item $F$ is $k$-regular (in natural generalisation of Definition 5.14 in \cite{Apt2003CSP}) iff for every variable set $V \sse \var(F)$ with $\abs{V} \le k$ there is at most one clause $C \in F$ with $\var(C) = V$. Due to the restricted power of (generalised) clauses, this notion seems to be not of great relevance here (on the other side, it might make sense to comprise binary clauses and possibly also ternary clauses with the same scope into constraints).
\item $F$ is relational arc-consistent (every assignment to all but one of the variables of a clause $C \in F$ can be extended to a satisfying assignment for $C$; see \cite{Dechter2003Constraints}, Definition 8.1) iff $F$ does not have trivial variables.
\item $F$ is (hyper-)arc-consistent (every assignment to one variable of a clause $C \in F$ can be extended to a satisfying assignment for $C$; see \cite{Dechter2003Constraints}, Definition 3.6) iff for all $C \in F \sm \set{\bot}$ after trivial-domain reduction at least two literals are left.
\end{enumerate}
So unit-clause propagation achieves both hyper-arc-consistency and relational arc-consistency. More information on the relations to the basic reduction concepts for constraint problems one finds in Subsection 3.3 of \cite{Kullmann2007Uebersicht}, while more extensive studies have been performed by David Mitchell.

\subsection{DP-reductions}
\label{sec:DPreductions}

Finally we consider the most harmless cases for \textbf{DP-reductions}. In general, application of $\dpi{v}$ to $F$ eliminates $\#_v(F) = \sum_{\ve \in D_v} \#_{(v,\ve)}(F)$ clauses and creates up to $\prod_{\ve \in D_v} \#_{(v,\ve)}(F)$ new clauses (with potential repetitions; less iff some of the parent clause combinations are not eligible for resolution due to additional clashes). Thus we have
\begin{equation}
  \label{eq:DP}
  c(\dpi{v}(F)) \le c(F) - \sum_{\ve \in D_v} \#_{(v,\ve)}(F) + \prod_{\ve \in D_v} \#_{(v,\ve)}(F).
\end{equation}
Note that equality holds in \eqref{eq:DP} if $v$ is pure for $F$ (which is equivalent to the product being zero). If in \eqref{eq:DP} we have a strict inequality or $v$ is a pure variable for $F$, then we call $v$ a \textbf{degenerated DP-variable w.r.t.\ $F$}, while otherwise $v$ is called a \textbf{non-degenerated DP-variable w.r.t.\ $F$}.  Note that a missing new clause due to additional clashes is not the only cause of a strict inequality, but it is also possible that a resolvent is already contained in the rest of $F$, or that two resolvents coincide (and thus in both cases contraction occurs). Two trivial cases:
\begin{enumerate}
\item If variable $v \in \var(F)$ has a trivial domain (i.e., $\abs{D_v} = 1$), then either we have a subsumption $C, C \cup \set{(v,\ve)} \in F$ (for some clause $C$ not containing $v$), in which case $v$ is a degenerated DP-variable w.r.t.\ $F$, or otherwise $v$ is a non-degenerated DP-variable with $c(\dpi{v}(F)) = c(F)$, and in both cases $\dpi{v}(F)$ is the result of applying trivial domain reduction to $F$.
\item If $v$ is pure w.r.t.\ $F$, then $F$ is a degenerated DP-variable with $c(\dpi{v}(F)) = c(F) - \#_v(F)$, and $\dpi{v}(F)$ is the result of applying the elimination of pure variable $v$ to $F$.
\end{enumerate}
Besides these cases, in this article we consider only one very restricted form of DP-resolution, characterised by the condition that at most one of the factors in the product from \eqref{eq:DP} might be greater than one:

We call a variable $v$ a \textbf{singular variable} w.r.t.\ $F$ if there exists $\ve \in D_v$ such that for all $\ve' \in D_v \sm \set{\ve}$ we have $\#_{(v,\ve')}(F) = 1$, while $\#_{(v,\ve)}(F) \ge 1$, i.e., $v$ is not pure for $F$. In such a case of a singular variable, application of $\dpi{v}$ eliminates $\abs{D_v} - 1 + \#_{(v,\ve)}(F)$ clauses and creates up to $\#_{(v,\ve)}(F)$ new clauses, so that the number of clauses goes down at least by one if $\abs{D_v} \not= 1$. For a singular variable $v$ we have:
\begin{itemize}
\item If $v$ is a non-degenerated DP-variable w.r.t.\ $F$, then we have $c(\dpi{v}(F)) = c(F) - \abs{D_v} + 1$.
\item If $v$ is a degenerated DP-variable w.r.t.\ $F$, then at least one of the clauses in $F$ containing $v$ can be eliminated satisfiability-equivalently, and we call such a clause elimination a \textbf{singular DP-degeneration reduction}.
\end{itemize}
Since a singular DP-degeneration reduction can not be applied to a minimally unsatisfiable clause-set, a singular variable w.r.t.\ a minimally unsatisfiable clause-set must be non-degenerated. Actually more can be said here:

\begin{lem}\label{lem:singDPMUSAT}
  Consider a generalised clause-set $F \in \Cls$ and a singular variable $v$ w.r.t.\ $F$. Then the following two conditions are equivalent:
  \begin{enumerate}
  \item $F$ is minimally unsatisfiable.
  \item $v$ is a non-degenerated DP-variable w.r.t.\ $F$ and $\dpi{v}(F)$ is minimally unsatisfiable.
  \end{enumerate}
\end{lem}
\begin{proof} We have already seen that if $F$ is minimally unsatisfiable, then $v$ is non-degenerated. If $\dpi{v}(F)$ were not minimally unsatisfiable, then there would be a clause $C \in \dpi{v}(F)$ such that $\dpi{v}(F) \sm \set{C}$ would still be unsatisfiable, and thus would have a resolution refutation --- now it is easy to see that in this case we would also obtain a resolution refutation of $F$ not using one of the clauses in $F$.

For the reverse direction assume that $v$ is non-degenerated and that $\dpi{v}(F)$ is minimally unsatisfiable. By a similar argumentation as for the other direction, if there would be a resolution refutation of $F$ not using one of the clauses from $F$, then one could construct a resolution refutation of $\dpi{v}(F)$ not using (at least) one of the clauses from $\dpi{v}(F)$.
\end{proof}

We call application of DP-reduction for non-degenerated singular variables \textbf{non-degenerated singular DP-reduction}. In the boolean case, this form of DP-reduction is used at many places in the literature (in \cite{Ku99dKo}, Appendix B, it is called ``$(1,\infty)$-reduction''); see Lemma \ref{lem:dpausnahme} for more on singular DP-reduction.

We conclude by another reduction arising from the DP-operator. The notion of \textbf{blocked clauses} for boolean clause-sets (see \cite{Ku97b,Ku98f}) can be generalised by calling a clause $C$ \textbf{blocked w.r.t.\ $F$} if there exists a variable $v \in \var(C)$ with $\dpi{v}(F \cup \set{C}) = \dpi{v}(F \sm \set{C})$. If $C \in F$ is blocked w.r.t.\ $F$, then $F$ is satisfiability equivalent to $F \sm \set{C}$, and such a reduction is called \textbf{elimination of blocked clauses}. If $v$ is a pure variable for $F$, then all clauses of $F$ containing variable $v$ are blocked w.r.t.\ $F$. And if $v$ is a degenerated singular variable, then $F$ has at least one blocked clause containing $v$, and so singular DP-degeneration reduction is also covered by elimination of blocked clauses. For recent practical applications of blocked clauses see \cite{JarvisaloBiereHeule2010BlockedClauseElimination}.

\section{Matching satisfiability and deficiency}
\label{sec:satviamatch}

We now generalise the basic method for satisfying (certain) clause-sets $F$ via matchings, as first systematically studied for the boolean case in \cite{FrGe98}. The idea is to choose variables $v_i$ so that each clause $C \in F$ contains at least one of them, that is, there is a literal $x_C \in C$ with $\var(x_C) = v_i$ for some $i$, and so that, \emph{whatever} these associated literal occurrences $x_C$ are, there is a partial assignment $v_i \ra \ve_i$ such that all literals $x_C$ become true --- if this can be achieved, then obviously we satisfy $F$. \textbf{Subsection \ref{sec:matchsat}} elaborates this idea, introducing ``matching satisfiable'' generalised clause-sets $F$. In \textbf{Subsection \ref{sec:defgencls}} we generalise the notion of deficiency $\delta(F)$, which has been introduced for boolean clause-sets $F$ in \cite{FrGe98} as $\delta(F) = c(F) - n(F)$. Using the maximum deficiency $\delta^*(F)$, the maximal deficiency taken over all sub-clause-sets of $F$, as for boolean clause-sets we obtain the characterisation of matching satisfiable clause-sets $F$ by the condition $\delta^*(F) = 0$. Whence matching satisfiability can be decided in polynomial time by finding a maximum matching, which also yields a satisfying assignment (called a ``matching-satisfying assignment'') in the positive case. Two technical remarks are necessary here:
\begin{enumerate}
\item Matching arguments are sensitive to repetition of clauses, and thus instead of just using clause-sets we have to use the more general notion of a \emph{multi}-clause-set (recall Subsection \ref{sec:notmulticlausesets}).
\item In case of a pure variable $v \in \var(F)$ for some $F \in \Mcls$ (that is, not all values $\ve \in D_v$ are used in $F$) we assume that $D_v$ contains exactly one value not used in $F$ (i.e., $\abs{D_v} = \abs{\val_v(F)} + 1$). In this way we are not troubled anymore by the unknown domain size $D_v$, but we can measure the size of $F$ just by $\ell(F)$, while this modification has no influence on any of the notions and procedures in this article (all autarky systems studied here are stable for unused values (recall Subsection \ref{sec:Autarkysystems})).
\end{enumerate}

\subsection{Matching satisfiable generalised clause-sets}
\label{sec:matchsat}

We wish to generalise the notion of ``matching satisfiable clause-sets'', introduced in \cite{Ku00f} for boolean clause-sets. Consider a multi-clause-set $F$ together with a decomposition $F = F_1 + \dots + F_m$ for $m \in \NNZ$ and $F_i \in \Mcls$, fulfilling the following conditions:
\begin{enumerate}[(i)]
\item\label{item:condmatchsat1} for $i \in \tb 1m$ there are variables $v_i \in \var(F_i)$ such that for all clauses $C \in F_i$ we have $v_i \in \var(C)$;
\item\label{item:condmatchsat2} the variables $v_1, \dots, v_m$ are pairwise different;
\item\label{item:condmatchsat3} for all $i \in \tb 1m$ we have $\abs{D_{v_i}} > \abs{\val_{v_i}(F_i)}$, that is, $v_i$ is a pure variable for $F_i$.
\end{enumerate}
We remark that none of the variables $v_i$ needs to be a pure variable in $F$. Given such a decomposition, we see that $F$ is satisfiable, since for each $i$ there exists $\ve_i \in D_{v_i} \sm \val_{v_i}(F_i)$, and the assignment $\pab{ v_i \ra \ve_i : i \in \tb 1m}$ is a satisfying assignment for $F$. Considering an arbitrary partial assignment $\vp$ satisfying $F$ with $\var(\vp) = \set{v_1, \dots, v_m}$, and setting $F_i$ for $i \in \tb 1m$ as the induced sub-multi-clause-set of $F$ given by the clauses $C \in F$ with $v_i \in \var(C)$ and $(v_i, \vp(v_i)) \notin C$, we obviously fulfil the above conditions, and we see that conditions (\ref{item:condmatchsat1}) - (\ref{item:condmatchsat3}) need to be restricted so that we can obtain a class of satisfiable clause-sets which is decidable in polynomial time. We observe that $c(F_i) \ge \abs{\val_{v_i}(F_i)}$ is true for arbitrary multi-clause-sets $F_i$, and thus condition
\begin{enumerate}
\item[(\ref{item:condmatchsat3})'] \label{item:condmatchsat3s} for all $i \in \tb 1m$ we have $\abs{D_{v_i}} > c(F_i)$
\end{enumerate}
strengthens condition (\ref{item:condmatchsat3}). We call multi-clause-sets $F \in \Mcls$ having a decomposition $F = F_1 + \dots + F_m$ fulfilling conditions (\ref{item:condmatchsat1}), (\ref{item:condmatchsat2}) and (\ref{item:condmatchsat3})' \textbf{matching satisfiable}, and the set of all matching satisfiable (generalised) multi-clause-sets is denoted by \bmm{\Msat}.

To understand the connection to matching problems, we introduce the bipartite graph $B(F)$ for generalised multi-clause-sets $F \in \Mcls$:
\begin{itemize}
\item Let\vspace{-4ex}
  \begin{eqnarray*}
    \ul{F} & := & \setb{(C,i) : C \in F, \, i \in \tb 1{F(C)}}\\
    \ul{V} & := & \setb{(v,j) : v \in \var(F), \, j \in \tb 1{\abs{D_v}-1}}.
  \end{eqnarray*}
  The elements of $\ul{F}$ are called the \emph{clause-nodes}, while the elements of $\ul{V}$ are called the \emph{variable-nodes}.
\item The vertex set of $B(F)$ is defined as $V(B(F)) := \ul{F} \addcup \ul{V}$ with canonical bipartition $(\ul{F}, \ul{V})$.
\item The edge set $E(B(F))$ is the set of all (undirected) edges $\set{(C,i), (v,j)}$ such that $v \in \var(C)$.
\end{itemize}
In other words, the graph $B(F)$ has as vertices $F(C)$-many copies of clauses $C \in F$ together with $(\abs{D_v} - 1)$-many copies of variables $v \in \var(F)$, while edges connect copies of variables $v$ with copies of clauses $C$ such that $v \in \var(C)$. We remark that variables $v \in \var(F)$ with trivial domain (i.e., $\abs{D_v} = 1$) do not occur in $B(F)$, and that for boolean clause-sets $F$ the graph $B(F)$ is the ordinary (bipartite) clause-variable graph. Consider for example the clause-set $F = \set{C_1,C_2,C_3}$ with $C_1 = \set{(v_1,a), (v_2,a)}$, $C_2 = \set{(v_2,b), (v_3,b)}$, $C_3 = \set{(v_3,c), (v_1,c)}$, where $\abs{D_{v_i}} = 3$. Now $B(F)$ is (suppressing the indices for the clause-copies, since here we just have a clause-set):
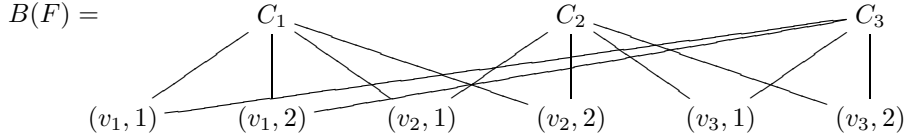
\begin{displaymath}
  B(F) = \hspace{-1em} \xymatrix {
    & C_1 \aru[dl] \aru[d] \aru[dr] \aru[drr] && C_2 \aru[dl] \aru[d] \aru[dr] \aru[drr] && C_3 \aru[dl] \aru[d] \aru[dlllll] \aru[dllll]\\
    (v_1,1) & (v_1,2) & (v_2,1) & (v_2,2) & (v_3,1) & (v_3,2)
  }
\end{displaymath}
For a set $V$ of variables we obtain $B(V * F)$ from $B(F)$ by deleting the variable-nodes $(v,j)$ of $B(F)$ with $v \in V$, while $B(F[V])$ is the induced subgraph of $B(F)$ given by the variable-nodes $(v,j)$ of $B(F)$ with $v \in V$ together with their neighbours (those clause-nodes $(C,i)$ with $\var(C) \cap V \not= \es$). Using the \textbf{weighted number of variables} 
\begin{displaymath}
  \bmm{\rd(F)} := \sum_{v \in \var(F)} (\abs{D_v} - 1) \in \NNZ,
\end{displaymath}
\begin{itemize}
\item the number of vertices of $B(F)$ is $\abs{V(B(F))} = c(F) + \rd(F)$,
\item while the number of edges is $\abs{E(B(F))} = \sum_{v \in \var(F)} \#_v(F) \cdot (\abs{D_v} - 1)$.
\end{itemize}
We have $\rd(F) = (\sum_{v \in \var(F)} \abs{D_v}) - n(F)$. If $F$ is boolean, then $\rd(F) = n(F)$.

\begin{lem}\label{def:directcharacmatchsat}
  A multi-clause-set $F$ is matching satisfiable iff there exists a matching in $B(F)$ covering all vertices of $\ul{F}$.
\end{lem}
\begin{proof}
 If $F$ is matching satisfiable, then (using the notations in the definition of matching satisfiability above) the clause-nodes corresponding to the clause-oc\-cur\-ren\-ces in $F_i$ can all be covered by the variable-nodes belonging to $v_i$ (since $c(F_i)$ does not exceed the number of copies of $v_i$), and altogether we obtain a matching covering all clause-nodes. If (for the other direction) we have a matching in $B(F)$ covering all vertices of $\ul{F}$, then for each variable $v$ involved in the matching consider a sub-multi-clause-set $F_v$ of $F$ corresponding to the clause-nodes connected via the matching to the variable-nodes associated with $v$. These $F_v$ together constitute the desired decomposition of $F$. 
\end{proof}

\subsection{The deficiency of generalised clause-sets}
\label{sec:defgencls}

Let the \textbf{deficiency} of a (generalised) multi-clause-set $F$ be defined as 
\begin{displaymath}
  \bmm{\delta(F)} := c(F) - \rd(F) \in \ZZ,
\end{displaymath}
while the \textbf{maximal deficiency} is defined as 
\begin{displaymath}
  \bmm{\delta^*(F)} := \max_{F' \le F} \delta(F') \in \NNZ.
\end{displaymath}
We have $\delta^*(F) \ge 0$ due to $\delta(\top) = 0$, and by definition we have $\delta^*(F) \le c(F)$. We remark that for a domain uniformisation $F'$ of $F$ we have $\delta(F') = \delta(F)$ as well as $\delta^*(F') = \delta^*(F)$; in principle we could consider only multi-clause-sets with uniform domains here, but the advantages in doing so seem to be negligible.

Considering $F' \le F$ as a subset of $\ul{F}$, the deficiency $\delta(F')$ of $F' \le F$ is just the deficiency of this subset in $B(F)$ (as we have defined it for arbitrary graphs). By matching theory the maximal number of nodes of $\ul{F}$ coverable by some matching thus is $c(F) - \delta^*(F)$. Summarising we have (generalising Lemma 7.2 in \cite{Ku00f}):
\begin{lem}\label{lem:characmatchsat}
  Consider a generalised multi-clause-set $F \in \Mcls$.
  \begin{enumerate}
  \item The maximal size of a matching satisfiable sub-multi-clause-set $F' \le F$ is $c(F') = c(F) - \delta^*(F)$.
  \item $F$ is matching satisfiable if and only if $\delta^*(F) = 0$.
  \end{enumerate}
\end{lem}

As an application we can generalise the well-known fact, apparently first mentioned in the literature in \cite{Tovey1984NPcomplete}, that if a boolean clause-set $F$ has minimal clause-length $k$ and maximal variable occurrence $k$ for some $k \ge 1$, then $F$ must be satisfiable (see \cite{HoorySzeider2006Tovey} for recent further developments):
\begin{corol}\label{cor:Tovey}
  Consider a generalised clause-set $F \in \Cls$ containing a non-empty clause. Then
  \begin{displaymath}
    \frac{\max_{v \in \var(F)} \#_v(F)}{\min_{C \in F} \abs{C}} \le \min_{v \in \var(F)} \abs{D_v} - 1 \Lora F \in \Msat.
  \end{displaymath}
\end{corol}
\begin{proof} Assume the condition holds and consider $F' \sse F$. We must show $\delta(F') \le 0$. Let $d := \min_{v \in \var(F)} \abs{D_v}$. Then $\delta(F') \le c(F') - (d - 1) n(F')$, and a sufficient condition for $\delta(F') \le 0$ is $\frac{c(F')}{n(F')} \le d - 1$. Let $a := \max_{v \in \var(F)} \#_v(F)$ and $b := \min_{C \in F} \abs{C}$. We know $c(F') \cdot b \le \ell(F') \le n(F') \cdot a$, and thus $\frac{c(F')}{n(F')} \le \frac ab$. 
\end{proof}

Since matchings of maximal size can be computed in polynomial time (see Chapter 16 in \cite{Schrijver2003CombOptA}), we get the following poly-time results:
\begin{lem}\label{lem:deltapoly}
  For every generalised clause-set $F \in \Mcls$, in polynomial time in $\ell(F)$ we can compute $F' \le F$ with $F' \in \Msat$ such that $c(F')$ is maximal. Since $F' = F$ iff $F$ is matching satisfiable, it follows that whether $F$ is matching satisfiable is decidable in polynomial time. And due to $c(F') = c(F) - \delta^*(F)$ furthermore the maximal deficiency $\delta^*(F)$ is also computable in polynomial time.
\end{lem}

\section{Satisfying assignments versus matching-sa\-tis\-fy\-ing assignments}
\label{sec:matchsatass}

The theme of this section is the relation between general satisfiability and matching satisfiability. From the main result, Theorem \ref{thm:Reparatur}, it follows that if a clause-set $F$ is satisfiable, then it has a matching satisfiable sub-clause-set $F'$ with at most $\delta^*(F)$ less clauses than $F$, and moreover there is a matching-satisfying assignment $\vp_0$ for $F'$ which can be extended to a satisfying assignment $\vp$ for $F$ using at most $\delta^*(F)$ additional variables. Furthermore, every satisfying assignment $\vp$ for $F$ can be modified (in polynomial time) to become such an extension by means of flips of (single) variable assignments such that throughout the whole process we always have a satisfying assignment for $F$.\footnote{This additional property yields also new information for the boolean case; it is implicitly contained in the proofs from \cite{FKS00}, which are not only generalised here, but also simplified  in such a way that the construction becomes more lucid.} As a direct application we obtain in Corollary \ref{cor:poly1} that the hierarchy of clause-sets given by the parameter $\delta^*$ allows polynomial-time SAT decision for each level. Later, in Theorem \ref{thm:MaximalerDefektFPT} (Part II), we will see that actually this hierarchy is fixed-parameter tractable, not using Theorem \ref{thm:Reparatur}, but by combining the structural results from Subsection \ref{sec:matchautsub} with the fixed-parameter tractability of the boolean case. The main use of Theorem \ref{thm:Reparatur} will be given in Theorem \ref{thm:leanKernelpoly}, namely how to compute a non-trivial autarky (if existent) in polynomial time w.r.t.\ the maximal deficiency, and this result was not known even for the boolean case.\footnote{That the lean kernel is computable in polynomial time w.r.t.\ the maximal deficiency has already been shown in \cite{Ku99dKo}, but interestingly it was not known how to find a (non-trivial) autarky (either just some, or a quasi-maximal one, yielding the lean kernel --- both tasks are basically the same).}

\subsection{The notion of matching-satisfying assignments}
\label{sec:notionmatchsat}

Consider a (generalised) multi-clause-set $F \in \Mcls$ and a partial assignment $\vp \in \Pass$. The partial graph \bmm{B_{\vp}(F)} of $B(F)$ is obtained from $B(F)$ by keeping (exactly) all edges $\set{(C,i), (v,j)}$ where $\vp$ satisfies the literal in $C$ with underlying variable $v$ (while keeping all vertices). In other words, all edges $\set{(C,i), (v,j)}$ are eliminated such that for the literal $(v, \ve) \in C$ we either have $v \notin \var(\vp)$ or $\vp((v,\ve)) = 0$ (i.e., $\vp(v) = \ve$). So the non-isolated clause-nodes in $B_{\vp}(F)$ are (exactly) the clauses satisfied by $\vp$, while the isolated variable-nodes in $B_{\vp}(F)$ are (exactly) the variables in $F$ not used by $\vp$ to satisfy any clause. Now $\vp$ is called a \textbf{matching-satisfying assignment} for $F$ if $B_{\vp}(F)$ contains a matching covering all clause-nodes (thus matching-satisfying assignments are satisfying assignments). By Lemma \ref{def:directcharacmatchsat} we get the following basic polynomial-time result, refining Lemma \ref{lem:deltapoly} by now also considering the (matching-)satisfying assignments.

\begin{lem}\label{lem:matchsatviamatchsa}
  A generalised multi-clause-set $F \in \Mcls$ is matching satisfiable if and only if there exists a matching-satisfying assignment for $F$. If $F$ is matching satisfiable, then by computing a maximum matching $M$ in $B(F)$ we can efficiently compute a matching-satisfying assignment $\vp$ for $F$ as follows: The domain of $\vp$ consists of the variables $v$ used in variable-nodes $(v,i)$ covered by $M$, while $\vp(v) := \ve$ for some (chosen) value $\ve \in D_v$ not used in the occurrences of variable $v$ in the clauses corresponding via $M$ to the variable-nodes using $v$ (note that by definition there are most $\abs{D_v}-1$ such variable-nodes).
\end{lem}

The following two lemmas give simple basic properties regarding these notions. First we consider how many variables need to be used by (matching-) satisfying assignments.

\begin{lem}\label{lem:matchc}
  Consider a generalised multi-clause-set $F \in \Mcls$ and a partial assignment $\vp \in \Pass$.
  \begin{enumerate}\NAA
  \item\label{item:matchc1} If $\vp$ is satisfying for $F$, then there exists $\vp' \sse \vp$ with $n(\vp') \le c(F)$ such that also $\vp'$ is satisfying for $F$. (So for satisfying any clause-set we never need to use more variables than there are clauses.)
  \item\label{item:matchc2} If $\vp$ is matching-satisfying for $F$, then there exists $\vp' \sse \vp$ with $n(\vp') = c(F)$ such that also $\vp'$ is matching-satisfying for $F$. (So for a matching satisfiable multi-clause-set there is a matching-satisfying assignment using exactly as many variables as there are clause-occurrences.)
  \item\label{item:matchc3} If $\vp$ is satisfying for $F$, and there is no $\vp' \sse \vp$ with $n(\vp') < c(F)$ such that $\vp'$ is satisfying for $F$, then $\vp$ is matching-satisfying for $F$.
  \item\label{item:matchc4} Consider a minimal satisfying assignment $\vp$ for $F$ w.r.t.\ the canonical partial ordering of partial assignments (that is, there is no $\vp' \subset \vp$ which still satisfies $F$). Then $\vp$ is matching-satisfying for $F$ if and only if $n(\vp) = c(F)$.
   \end{enumerate}
\end{lem}
\begin{proof} The partial assignment $\vp'$ for Part \ref{item:matchc1} is obtained by removing edges from  $B_{\vp}(F)$ until every clause-node has degree $1$, and using then only the variables from $\vp$ which are still covered. $\vp'$ for Part \ref{item:matchc2} is obtained in a similar way, only this time we remove all edges not contained in some (selected) maximum matching of $B_{\vp}(F)$. Part \ref{item:matchc3} is shown by Hall's criterion as follows: Assume that there is $F' \le F$ such that the number $\abs{\Gamma_{B_{\vp}(F)}(F')}$ of neighbours of $F'$ in $B_{\vp}(F)$ is strictly smaller than $c(F')$, and let $\vp' := \vp \mb \Gamma_{B_{\vp}(F)}(F')$ be the restriction of $\vp$ to these neighbours; by definition $\vp'$ is a satisfying assignment for $F'$ with $n(\vp') < c(F')$. Let $F'' := F - F'$. With Part \ref{item:matchc1} there is a satisfying assignment $\vp'' \sse \vp$ for $F''$ with $n(\vp'') \le c(F'') = c(F) - c(F')$. Now let $\vp^* := \vp' \cup \vp''$; by definition $\vp^*$ is a satisfying assignment for $F$ with $n(\vp^*) \le n(\vp') + n(\vp'') < c(F') + c(F) - c(F') = c(F)$ contradicting the assumption. Finally Part \ref{item:matchc4} follows by Parts \ref{item:matchc2} and \ref{item:matchc3}. 
\end{proof}

Now we show that every sub-assignment of a matching-satisfying assignment is also matching-satisfying for a suitable (multi-)clause-set.
\begin{lem}\label{lem:matchsatEig}
  Consider a generalised multi-clause-set $F \in \Mcls$ and partial assignments $\vp, \psi \in \Pass$. If $\vp \circ \psi$ is matching-satisfying for $F$, then $\vp$ is matching-satisfying for $\psi * F$.
\end{lem}
\begin{proof}
 Let $M$ be a matching in $B_{\vp \circ \psi}(F)$ covering all clause-nodes. The bipartite graph $B_{\vp}(\psi * F)$ is obtained from $B_{\vp \circ \psi}(F)$ by removing all clause-nodes satisfied by $\psi$, removing all variable-nodes assigned by $\psi$, and finally removing all variable-nodes where the variable does not occur in $\psi * F$ anymore. Now those edges from $M$ which are still in $B_{\vp \circ \psi}(F)$ yield a matching (obviously, since only edges have been removed) covering all remaining clause-nodes (they were covered before, and only useless edges have been removed). 
\end{proof}

\subsection{Matchings within satisfying assignments}
\label{sec:Matchingswithin}

As we already remarked, the matching number $\nu(B(G))$ of the clause-variable graph of $G$, the maximum size of a matching in $B(G)$, is $\nu(G) = c(F) - \delta^*(F)$. Obviously for partial assignments $\vp$ we have $\nu(B_{\vp}(F)) \le \nu(B(F))$; call $\vp$ \textbf{matching-maximum} if $\nu(B_{\vp}(F)) = \nu(B(F))$ holds. By Lemma \ref{lem:characmatchsat} we know that there exists a matching-maximum partial assignment for every clause-set. The main result of Section \ref{sec:matchsatass} is the following theorem, which states that every partial assignment can be efficiently repaired by ``conservative changes'', so that we obtain a matching-maximum partial assignment. Here by a \textbf{conservative change} of a partial assignment $\vp$ w.r.t.\ a clause-set $F$ we mean either adding some assignment $v \mapsto \ve \in D_v$ for some $v \in \var(F) \sm \var(\vp)$, or performing a \textbf{conservative flip}, that is, changing the value $v \in \var(\vp) \mapsto \vp(v)$ to some $\ve \in D_v \sm \set{\vp(v)}$, obtaining $\vp'$, such that all clauses of $F$ satisfied by $\vp$ are also satisfied by $\vp'$ (note that this property holds automatically for the first type of change, the extension of $\vp$). So if we have a sequence of conservative changes, then the corresponding sequence of sub-sets of satisfied clauses is monotonically increasing; especially if we start with a satisfying assignment, then all partial assignments in the sequence will also be satisfying assignments.

\begin{thm}\label{thm:Reparatur}
  For a generalised multi-clause-set $F \in \Mcls$ and a partial assignment $\vp_0$, in polynomial time a sequence of conservative changes w.r.t.\ $F$, starting with $\vp_0$, can be computed such that the finally obtained partial assignment $\vp$ is matching-maximum for $F$.
\end{thm}

Before proving this theorem, we derive three corollaries.
\begin{corol}\label{cor:satmatchsat}
  For a satisfiable generalised multi-clause-set $F \in \Mcls$ there exists a satisfying assignment which is matching-maximum.
\end{corol}

We obtain the following generalisation of Theorem 7.16 in \cite{Ku00f}:
\begin{corol}\label{cor:Abstandsat}
  For each satisfiable generalised multi-clause-set $F \in \Mcls$ there exists a partial assignment $\vp \in \Pass$ with $n(\vp) \le \delta^*(F)$ such that $\vp * F$ is matching satisfiable.
\end{corol}
\begin{proof}
 By Corollary \ref{cor:satmatchsat} there exists a satisfying assignment $\vp_0$ for $F$ and $F' \le F$ with $c(F') = c(F) - \delta^*(F)$, such that $\vp_0$ is matching-satisfying for $F'$. Let $F'' := F - F'$. We have $c(F'') = \delta^*(F)$ and $\vp_0$ is satisfying for $F''$, so by Lemma \ref{lem:matchc}, Part \ref{item:matchc1} there exists $\vp \sse \vp_0$ with $n(\vp) \le \delta^*(F)$ such that $\vp$ is satisfying for $F''$. Now $\vp_0 = \vp_0 \circ \vp$ is matching-satisfying for $F'$, and thus by Lemma \ref{lem:matchsatEig} $\vp_0$ is matching-satisfying for $\vp * F' = \vp * F' + \vp * F'' = \vp * (F' + F'') = \vp * F$. 
\end{proof}

\begin{corol}\label{cor:poly1}
  The satisfiability problem for generalised multi-clause-sets $F$ with $\delta^*(F) \le k$ for constant $k \in \NNZ$ is decidable in polynomial time (and if $F$ is satisfiable, then a satisfying assignment can be computed). The algorithm runs through all partial assignments $\vp$ with $\var(\vp) \sse \var(F)$ and $n(\vp) \le \delta^*(F)$ and checks whether $\vp * F$ is matching satisfiable: If yes then for a matching-satisfying assignment $\psi$ for $\vp * F$ we obtain a satisfying assignment $\vp \circ \psi$ for $F$, while if no matching satisfiable sub-instance was found in this way, then $F$ is unsatisfiable.
\end{corol}
In \cite{Szei2002FixedParam} it was shown that in the boolean case the satisfiability problem for bounded maximal deficiency actually is fixed-parameter tractable. By reducing the general case to the boolean case, we will show fixed-parameter tractability in Theorem \ref{thm:MaximalerDefektFPT} (Part II) also for generalised clause-sets. So just for satisfiability decision for bounded maxima deficiency, Theorem \ref{thm:Reparatur} is in a sense superseded by a different method, but we will later show in Theorem \ref{thm:leanKernelpoly}, that actually a stronger result can be obtained from Theorem \ref{thm:Reparatur}, namely computation of the lean kernel in polynomial time.

The remainder of this subsection is devoted to the proof of Theorem \ref{thm:Reparatur} (generalising, simplifying and also strengthening the results on ``admissible matchings'' in \cite{FKS00}).  In order to bring out the general structure of the proof we will present a more general result (which directly implies Theorem \ref{thm:Reparatur}), exploring ``parameterised maximum matching problems''. We consider the situation where a fixed graph $G$ is given together with an arbitrary set $\mc{P} \not= \es$ of ``parameters'' (which in our case are partial assignments, while $G = B(F)$) and a mapping $\vp \in \mc{P} \mapsto G_{\vp}$, where $G_{\vp}$ is a partial graph of $G$ (that is, $V(G_{\vp}) = V(G)$ and $E(G_{\vp}) \sse E(G)$). Let us call this parameterisation \emph{matching-optimal}, if there exists some $\vp \in \mc{P}$ such that $\nu(G_{\vp}) = \nu(G)$. A matching-optimal parameterisation does not establish a method to find some (``good'') $\vp \in \mc{P}$ where $\nu(G_{\vp}) = \nu(G)$ is attained. We consider the problem that we want to transform some arbitrary starting parameter $\vp_0$ into such a good $\vp$, and so we assume further that some relation $R \sse \mc{P} \times \mc{P}$ is given such that a relation $\vp R \vp'$ indicates an admissible move (in our application $\vp R \vp'$ holds if $\vp'$ results from $\vp$ by a conservative change). Denoting by $R^*$ the reflexive-transitive hull of $R$ (that is, allowing an arbitrary number of admissible moves), we call the parameterisation \emph{strongly matching-optimal}, if for every $\vp_0 \in \mc{P}$ there exists $\vp \in \mc{P}$ with $\vp_0 R^* \vp$ and $\nu(G_{\vp}) = \nu(G)$.

In the sequel we only consider bipartite $G$ with a bipartition $(A,B)$. We call the parameterisation $\mc{P}$ \emph{conditionally extensible} if for every $\vp \in \mc{P}$, every matching $M$ in $G_{\vp}$ and every edge $e \in E(G) \sm M$ such that $M' := M \cup \set{e}$ is a matching in $G$ and no matching $M^*$ in $G_{\vp}$ with $M^* \supset M$ covers the endpoint of $e$ in $B$, there exists $\vp' \in \mc{P}$ with $\vp R^* \vp'$ such that $M'$ is a matching in $G_{\vp'}$.

\begin{lem}\label{lem:condext}
  If a parameterised matching problem $(G, \mc{P})$ is conditionally extensible, then it is strongly matching-optimal. If furthermore in polynomial time in the size of $G$ a sequence of admissible moves from $\vp$ to $\vp'$ in any conditional extension can be found, then in polynomial time (in the size of $G$) for any $\vp_0 \in \mc{P}$ a sequence of admissible moves to $\vp \in \mc{P}$ with $\nu(B_{\vp}) = \nu(G)$ can be found.
\end{lem}

Before proving Lemma \ref{lem:condext}, we show that in our application, considering $B(F)$ with the parameterisation by $\vp \in \Pass \mapsto B_{\vp}(F)$ together with $R$ as the relation of conservative change, the property of conditional extensibility holds true. So we consider the situation where we have a maximal matching $M$ in $B_{\vp}(F)$, where an edge $\set{C,v}$ in $B(F)$ exists with uncovered endpoints, and we show that by just one conservative change we can extend $M$ by this additional edge.
\begin{lem}\label{lem:Hilfssatz1}
  Consider a generalised multi-clause-set $F \in \Mcls$, a partial assignment $\vp \in \Pass$, a matching $M$ in $B_{\vp}(F)$ and an edge $\set{(C_0, i), (v_0, j)} \in E(B(F))$ such that neither $C := (C_0, i)$ nor $v := (v_0, j)$ is covered by $M$. We assume furthermore that no matching $M^* \supset M$ in $B_{\vp}(F)$  covers $v$. Let $M' := M \cup \set{\set{C,v}}$. By definition $M'$ is a matching in $B(F)$, and furthermore there exists a conservative change for $\vp$ w.r.t.\ $F$, resulting in $\vp'$, such that $M'$ is a matching in $B_{\vp'}(F)$.
\end{lem}
\begin{proof} Let $(v_0, \ve_0) \in C_0$. If $v_0 \notin \var(\vp)$, then $\vp' := \vp \circ \pao{v_0}{\ve}$ for any $\ve \in D_{v_0} \sm \set{\ve_0}$ yields the required conservative change; so assume $v_0 \in \var(\vp)$. Now we have $\vp(v_0) = \ve_0$, since otherwise $M'$ would be a matching in $B_{\vp}(F)$ covering $v$. Let $E$ be the set of values $\ve \in D_{v_0}$ occurring in $M$, that is, there is some edge $\set{(C_0',i'), (v_0,j')} \in M$ with $(v_0, \ve) \in C_0'$. Since $v$ is not covered by $M$ and $M$ is a matching, $M$ can cover at most $(\abs{D_{v_0}} - 1) - 1$ many variable-nodes with underlying variable $v_0$, and so we have $\abs{E} \le \abs{D_{v_0}} - 2$. Thus there is $\ve' \in D_{v_0} \sm \set{\ve_0}$. Set $\vp' := \vp \circ \pao{v_0}{\ve'}$. By definition $M'$ is a matching in $B_{\vp'}(F)$. Now consider a clause $D \in F$ falsified by $\vp'$, and assume that $D$ is not falsified by $\vp$. Thus $(v_0, \ve') \in D$, and the literal $(v_0, \ve')$ is the only literal in $D$ satisfied by $\vp$. So no clause-node covered by $M$ has clause $D$ associated with it. It follows that $M^* := M \cup \set{\set{(D,1), v}}$ would be a matching in $B_{\vp}(F)$ extending $M$ and covering $v$, contradicting the assumption of the assertion. \end{proof}

Thus by Lemma \ref{lem:condext} now Theorem \ref{thm:Reparatur} is proven. In the remainder of this subsection we prove Lemma \ref{lem:condext}. The reader might recall the preliminaries on matchings (Subsection \ref{sec:Graphsmatching}), where the notion of an $M$-augmenting path $P$ for a matching $M$ in a graph $G$ is discussed: a larger matching $M^+$ is obtained by adding the edges from $P$ to $M$, which are not in $M$, while removing the other edges of $P$ from $M$. In order to perform the ``relinking'', necessary for the transition from $M$ to $M^+$, we show an auxiliary lemma.

\begin{lem}\label{lem:relinking}
  As in Lemma \ref{lem:condext} consider a conditionally extensible parameterised matching problem $(G, \mc{P})$, where $(A,B)$ is a bipartition of $G$. Consider furthermore $\vp \in \mc{P}$, a maximal matching $M$ in $G_{\vp}$, an edge $e  = \set{a,b} \in M$ with $a \in A, b \in B$, and an edge $e' = \set{a,b'} \in E(G)$ where $b'$ is not covered by $M$. Then by at most one conditional-extension-step we obtain a parameter $\vp' \in \mc{P}$ with $\vp R^* \vp'$ such that $M' := (M \sm \set{e}) \cup \set{e'}$ is a matching in $G_{\vp'}$.
\end{lem}
\begin{proof} If $M'$ is a matching in $G_{\vp}$, then we are done. Otherwise we have $e' \notin E(G_{\vp})$; let $M_0 := M \sm \set{e}$. We want to apply conditional extension to $M_0$ and $M' = M_0 \cup \set{e'}$ (if we succeed then we are done), so we have to show that there is no matching $M_0^*$ in $G_{\vp}$ covering $b'$ with $M_0 \subset M_0^*$. Assume the contrary, and consider the edge $\set{x,b'} \in M_0^*$: Since $G$ is a graph (no parallel edges) we have $x \not= a$, and thus $M \cup \set{\set{x,b'}}$ is a matching in $G_{\vp}$ (using bipartiteness of $G$) contradicting maximality of $M$. \end{proof}

Now we are in a position to prove Lemma \ref{lem:condext}; it suffices to show that for any given $\vp \in \mc{P}$ and a matching $M$ in $G_{\vp}$ with $\abs{M} < \nu(G)$, by a polynomial number of extension steps we can find $\vp'$ with $\vp R^* \vp'$ and a matching $M'$ in $G_{\vp'}$ with $\abs{M'} > \abs{M}$ (by repeating this process we finally obtain a maximum matching for $G$). If $M$ is not maximal in $G_{\vp}$, then we can add one edge while keeping $\vp$ and we are done, so assume that $M$ is maximal in $G_{\vp}$. Since $M$ is not a maximum matching in $G$, there exists an $M$-augmenting path $P$ in $G$. $P$ is of the form $(v_0, \dots, v_m)$ for (pairwise) different vertices $v_i$ and $m$ odd, such that $v_0, v_m$ are not covered by $M$ and such that for $0 \le i < m$ we have $\set{v_i, v_{i+1}} \notin M$ for even $i$, while for odd $i$ we have $\set{v_i, v_{i+1}} \in M$. W.l.o.g.\ $v_i \in B$ for all even $i$. The first task is for odd $i < m$ to replace the edge $\set{v_i, v_{i+1}}$ by $\set{v_i, v_{i-1}}$, using Lemma \ref{lem:relinking}; we proceed consecutively for $i = 1, 3, 5, \dots$, where if at some point Lemma \ref{lem:relinking} is not applicable, then we constructed a matching of the same size as $M$ which is not maximal w.r.t.\ its parameter, and so we obtain $M'$ by enlarging this matching. Otherwise, if the process goes through, then at the end we obtain a matching $M_0'$ in $G_{\vp_0'}$ where $\abs{M_0'} = \abs{M}$ and $v_{m-1}, v_m$ are not covered by $M_0'$. Again, if $M_0'$ is not maximal, then we get the required larger $M'$, while otherwise we can apply conditional extension, obtaining $M' := M_0' \cup \set{\set{v_{m-1}, v_m}}$. QED

\section{Matching autarkies}
\label{sec:autviamatch}

In this section we introduce the autarky system for generalised clause-sets given by ``matching autarkies'', and we develop various polynomial time procedures. ``Matching autarkies'' for clause-sets with non-boolean variables have been introduced in \cite{Ku01a}, and some basic properties have been stated regarding the direct translation of clause-sets with non-boolean variables to clause-sets with boolean variables. However, as we will discuss in Subsection \ref{sec:comparisonearlier}, this earlier version of the notion is actually too restrictive (another example for the subtleties one encounters with non-boolean variables). An outline of the content of this section is as follows.

Having a (restricted) concept $\mf{C}$ of satisfying assignments, we can ``typically'' obtain an autarky system (recall Subsection \ref{sec:prelimAutarkies}) by calling $\vp$ a ``$\mf{C}$-autarky'' for a clause-set $F$ if $\vp$ is a $\mf{C}$-satisfying assignment for $F[\var(\vp)]$ (recall Subsection \ref{sec:opsetsvar}), or, equivalently at least for general satisfiability, if $\vp$ is a $\mf{C}$-satisfying assignment for $F_{\var(\vp)}$. We have to leave such a general theory to future work, but in this article we will consider in \textbf{Subsection \ref{sec:matchautsub}} ``matching autarkies'' obtained in this way from matching-satisfying assignments. A fundamental notion is the notion of a ``tight sub-clause-set'' $F'$ of a clause-set $F$, characterised by the condition $\delta(F') = \delta^*(F)$ (that is, $F'$ realises the maximal deficiency of $F$). Translating general results of matching theory into our setting, the set of tight sub-clause-sets of $F$ form a set-lattice (i.e, union and intersection of tight sub-clause-sets are again tight), and so we have a smallest and a largest tight sub-clause-set. In \textbf{Subsection \ref{sec:matchleangcls}} we consider matching leanness (and the matching lean kernel), and in Lemma \ref{lem:charakschlank} we characterise matching lean clause-sets $F$ by the condition that all strict sub-clause-sets of $F$ have a deficiency strictly less than the deficiency of $F$, in other words, $F$ is the only tight sub-clause-set. Thus matching leanness is decidable in polynomial time, and applying the general procedure from Lemma \ref{lem:vonschlanknachred} we obtain polynomial-time computability of the matching lean kernel in Corollary \ref{cor:Mleanpoly2}.\footnote{Direct computations using matching arguments are more efficient; see Subsection \ref{sec:openApplyingmatching} for a discussion.} Since the empty sub-clause-set has deficiency $0$, we obtain $\delta(F) \ge 1$ for non-empty matching lean clause-sets. Matching lean clause-sets of minimal deficiency (i.e., deficiency $1$) are further considered in \textbf{Subsection \ref{sec:matchleangclsmindef}}. Finally, in \textbf{Subsection \ref{sec:comparisonearlier}}, we reflect on the definition of ``deficiency'' and ``matching autarky'' as defined in this article, by comparing it with an earlier version of these notions (for generalised clause-sets).

\subsection{Matching autarkies for generalised clause-sets}
\label{sec:matchautsub}

A partial assignment $\vp$ is called a \textbf{matching autarky} for $F \in \Mcls$ if $\vp$ is matching-satisfying for $F_{\var(\vp)}$, which is equivalent to $\vp$ being matching-satisfying for $F[\var(\vp)]$. Said explicitly, a partial assignment $\vp$ is a matching autarky for a (generalised) multi-clause-set $F$ iff for every clause-occurrence $C$ in $F$ touched by $\varphi$ we can choose a literal $x_C \in C$ such that $\varphi(x_C) = 1$ and such that for every variable $v \in \var(C)$ there are at most $\abs{D_v} - 1$ many (touched) clause-occurrences $C$ with $\var(x_C) = v$. The set of all matching autarkies for $F$ is denoted by \bmm{\maut(F)}. Generalising Lemma 7.1 and the remarks in Section 8 of \cite{Ku00f} we get
\begin{lem}\label{lem:matchautnormal}
  $F \in \Mcls \mapsto \maut(F) \sse \aut(F)$ is a normal autarky system.
\end{lem}
We denote by $\bmm{\nma} := \nv_{\maut}$ the normal form for multi-clause-sets obtained by eliminating all matching autarkies. According to our general results and definitions on autarky systems, the set of $\maut$-satisfiable multi-clause-sets is just $\Msat$, the set of matching satisfiable multi-clause-sets. The set of $\maut$-lean clause-sets is denoted by \bmm{\Mlean}, its elements are called \textbf{matching lean} multi-clause-sets. We now seek to characterise $\Mlean$, and to compute $\nma(F)$ in polynomial time.

A sub-multi-clause-set $F' \le F$ of a multi-clause-set $F \in \Mcls$ is called \textbf{tight} if $\delta(F') = \delta^*(F)$ holds. If $F'$ is tight for $F$, then $F'$ is an induced sub-multi-clause-set of $F$. By supermodularity of the deficiency (for graphs) we immediately get

\begin{lem}\label{lem:VereinigungDurchschnittKnapp}
  (Binary) Union and intersection of tight sub-multi-clause-sets of a given multi-clause-set are again tight. So the tight sub-clause-sets of a clause-set form a set-lattice with smallest and largest element.
\end{lem}

Generalising Lemma 7.3 in \cite{Ku00f}, we obtain the fundamental relationship between tight sub-multi-clause-sets and matching autarkies: application of matching autarkies does not reduce the deficiency, and application of suitable matching autarkies allows to realise every tight sub-multi-clause-set.
\begin{lem}\label{lem:tightdef}
  Consider a generalised multi-clause-set $F \in \Mcls$.
  \begin{enumerate}
  \item\label{item:tightdef0} For every autarky $\vp$ for $F$ we have $\delta(\vp * F) = \delta(F) - \delta(F[\var(\vp)])$.
  \item\label{item:tightdef1} For every matching autarky $\vp$ for $F$ we have $\delta(\vp * F) \ge \delta(F)$, and thus $\delta^*(\vp * F) = \delta^*(F)$.
  \item\label{item:tightdef3} If $F' \le F$ is induced, then we have $\delta^*(\var(F') * (F - F')) \le \delta^*(F) - \delta(F')$.
  \item\label{item:tightdef2} If $F' \le F$ is tight, then there is a matching autarky $\vp$ for $F$ with $\vp * F = F'$.
  \end{enumerate}
\end{lem}
\begin{proof} For Part \ref{item:tightdef0} note that by definition we have 
\begin{displaymath}
  c(F) = c(\vp * F) + c(F[\var(\vp)]), \ n(F) = n(\vp * F) + n(F[\var(\vp)])
\end{displaymath}
due to $F = \vp * F + F_{\var(\vp)}$. Part \ref{item:tightdef1} follows from Part \ref{item:tightdef0}. For Part \ref{item:tightdef3} consider $G \le \var(F') * (F - F')$. There exists $G_0 \le F - F'$ with $\var(F') * G_0 = G$ (consider the original clauses); so we have $c(G_0) = c(G)$ and $\var(G_0) \sm \var(F') = \var(G)$. Now 
\begin{multline*}
  \delta^*(F) \ge \delta(F' + G_0) = c(F' + G_0) - \rd(F' + G_0) =\\
  c(F') + c(G) - \rd(F') - \rd(G) = \delta(F') + \delta(G),
\end{multline*}
and thus $\delta(G) \le \delta^*(F) - \delta(F')$. Now Part \ref{item:tightdef2} follows immediately from Part \ref{item:tightdef3} due to $\delta^*(\var(F') * (F - F')) \le \delta^*(F) - \delta(F') = 0$, i.e., $F - F'$ is a matching autark sub-multi-clause-set of $F$. \end{proof}

\subsection{Matching lean generalised clause-sets}
\label{sec:matchleangcls}

Generalising Theorem 7.5 in \cite{Ku00f}, we now can characterise matching lean multi-clause-sets:
\begin{lem}\label{lem:charakschlank}
  Consider a generalised multi-clause-set $F \in \Mcls$. The following conditions are equivalent:
  \begin{enumerate}
  \item\label{item:charakschlank1} $F$ is matching lean;
  \item\label{item:charakschlank2} $\fa\, C \in F : \delta^*(F - \set{C}) < \delta^*(F)$;
  \item\label{item:charakschlank3} $\fa\, F' \lneq F : \delta(F') < \delta(F)$;
  \item\label{item:charakschlank4} $F$ is a tight sub-multi-clause-set of $F$, and there are no other tight sub-multi-clause-sets of $F$.
  \end{enumerate}
\end{lem}
\begin{proof} From Part \ref{item:charakschlank1} follows Part \ref{item:charakschlank4} by Lemma \ref{lem:tightdef}, Part \ref{item:tightdef2}. Obviously, Part \ref{item:charakschlank4} implies Part \ref{item:charakschlank3}, and Part \ref{item:charakschlank3} implies Part \ref{item:charakschlank2}. Finally, Part \ref{item:charakschlank1} follows from Part \ref{item:charakschlank2} by Lemma \ref{lem:tightdef}, Part \ref{item:tightdef1}. \end{proof}

\begin{corol}\label{cor:Mleanpoly}
  It is decidable in polynomial time, whether a generalised multi-clause-set $F \in \Mcls$ is matching lean.
\end{corol}
\begin{proof}
  Use Lemma \ref{lem:deltapoly} (poly-time computability of the maximal deficiency) and Lemma \ref{lem:charakschlank}, Part \ref{item:charakschlank2}.
\end{proof}

An alternative method for deciding matching leanness will be given later (by Lemma \ref{lem:surpone} together with Lemma \ref{lem:compsurp}). Applying the general procedure from Lemma \ref{lem:vonschlanknachred}, we can enhance (matching) leanness decision to the computation of the (matching) lean kernel:
\begin{corol}\label{cor:Mleanpoly2}
  The matching lean kernel $\nma(F)$ for generalised multi-clause-sets $F \in \Mcls$ is computable in polynomial time.
\end{corol}

Since matching autarkies constitute a normal autarky system, we know that there are matching autarkies $\vp$ for $F$ with $\vp * F = \nma(F)$; in Lemma \ref{lem:polytimeqmma} we show how to find such ``quasi-maximal'' matching autarkies in polynomial time, using the computation of the matching lean kernel as the main step. See also the open problems in Subsection \ref{sec:openApplyingmatching} regarding more efficient computations.

Back to the characterisation of the matching lean kernel, by Lemma \ref{lem:charakschlank}, Part \ref{item:charakschlank4} together with Lemma \ref{lem:tightdef}, Part \ref{item:tightdef2} we get
\begin{corol}\label{cor:leankerntight}
  For every generalised multi-clause-set $F \in \Mcls$ the matching lean kernel $\nma(F)$ is the intersection of all tight sub-multi-clause-sets of $F$. Thus $\nma(F)$ is the smallest tight sub-multi-clause-set of $F$, and therefore $\delta^*(F) = \delta(\nma(F))$.
\end{corol}

\subsection{Matching lean clause-sets of minimal deficiency}
\label{sec:matchleangclsmindef}

Using $\delta(\top) = 0$, from Lemma \ref{lem:charakschlank}, Part \ref{item:charakschlank3} we get the following generalisation of ``Tarsi's Lemma'' (see \cite{AhLi86}):
\begin{corol}\label{cor:TarsisMLean}
  If the generalised multi-clause-set $F \not= \top$ is matching lean, then $\delta^*(F) = \delta(F) \ge 1$.
\end{corol}

Obviously $\Musat \subset \Lean$, and thus:
\begin{corol}\label{cor:musatdelta}
  If a generalised clause-set $F \in \Cls$ is minimally unsatisfiable, then we have $\delta^*(F) = \delta(F) \ge 1$.
\end{corol}
In \cite{CreignouDaude2002ThresholdConstraints}, Theorem 4.5, arbitrary constraints over boolean variables are considered, and a lower bound on the number of clauses in terms of the number of variables for minimally unsatisfiable constraint satisfaction problems is derived, which necessarily is much weaker than Corollary \ref{cor:musatdelta}.

After these general structural results on matching autarkies and matching lean clause-sets, we conclude this section by some example classes. Considering minimally unsatisfiable clause-sets of minimal deficiency, we observe that removing any clause from a matching lean multi-clause-set $F$ with $\delta(F) = 1$ yields a matching satisfiable multi-clause-set, and thus
\begin{corol}\label{cor:musatonemlean}
  $\Musati{\delta = 1} = \Mleani{\delta = 1} \cap \Usat$.
\end{corol}
The class $\Musati{\delta = 1}$ of minimally unsatisfiable clause-sets of minimal deficiency is characterised in Theorem \ref{thm:CharakMUSATd1} (Part II), and so we have a good understanding of the unsatisfiable elements of $\Mleani{\delta = 1}$. The satisfiable elements of $\Mleani{\delta = 1}$ on the other hand seem to have a more complicated nature, where the interest in this class may be justified by the following property.

\begin{corol}\label{cor:minmatunsat}
  The class $\Mleani{\delta = 1}$ of matching lean generalised clause-sets of deficiency $1$ is exactly the class of all minimally matching unsatisfiable clause-sets (clause-sets which are not matching satisfiable, while every strict subset is matching satisfiable).
\end{corol}
An interesting example for a satisfiable boolean $F \in \Mleani{\delta = 1}$ with $n(F) = 8$, exhibited in Section 5 of \cite{Ku2003e}, is given by the following clause-variable matrix (the rows correspond to the clauses, the columns to the variables, where an entry ``$\pm$'' denotes a positive/negative occurrence, while $0$ denotes non-occurrence):
\begin{displaymath}
  M(F) =
  \begin{pmatrix}
    + & 0 & + & 0 & 0 & 0 & 0 & +\\
    + & 0 & - & + & 0 & 0 & 0 & 0\\
    - & - & 0 & 0 & 0 & + & + & 0\\
    - & - & 0 & 0 & - & - & 0 & 0\\
    0 & + & - & - & - & 0 & 0 & 0\\
    0 & 0 & 0 & - & + & 0 & - & -\\
    - & 0 & 0 & 0 & + & - & 0 & +\\
    - & 0 & 0 & + & 0 & + & - & 0\\
    0 & + & + & 0 & 0 & 0 & + & -
  \end{pmatrix}.
\end{displaymath}
Obviously $\delta(F) = 1$. Since every two different rows clash in exactly one element, $F$ is a $1$-regular hitting clause-set. Every column contains at least two ``$-$'' and two ``$+$'', and thus every variable occurs at least two times negatively as well as positively (the purpose of this example was to refute a conjecture by Endre Boros that every boolean $1$-regular hitting clause-set of deficiency $1$ must contain a variable occurring in one sign only once). To demonstrate that $F$ is matching lean, consider the following subgraph of $B(F)$:
{\nc{\cn}[1]{\mr{c}^{#1}} \nc{\vn}[1]{\bmm{\mr{v}_{#1}}}
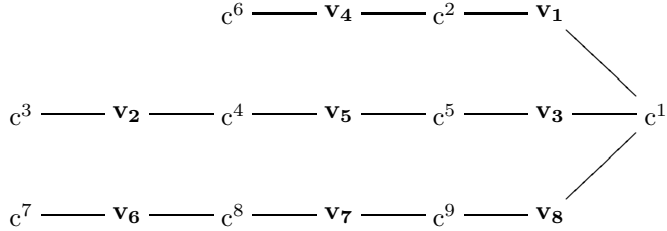
\begin{displaymath}
  \xymatrix {
    && \cn{6} \aru[r] & \vn{4} \aru[r] & \cn{2} \aru[r] & \vn{1}\\
    \cn{3} \aru[r] & \vn{2} \aru[r] & \cn{4} \aru[r] & \vn{5} \aru[r] & \cn{5} \aru[r] & \vn{3} \aru[r] & \cn{1} \aru[ul] \aru[dl]\\
    \cn{7} \aru[r] & \vn{6} \aru[r] & \cn{8} \aru[r] & \vn{7} \aru[r] & \cn{9} \aru[r] & \vn{8}
  }
\end{displaymath}
}
Here variable-nodes corresponding to row $j$ are denoted by $v_j$, and clause-nodes corresponding to row $i$ by $c^i$. This subgraph has the special property that it is a spanning tree, where the variable-nodes all have degree $2$. From these properties by Corollary \ref{cor:CharakMLEANWald} in the next subsection it follows that $F$ is matching lean.\footnote{This can also be seen directly as follows. For a graph $G$ let $\delta(G) := \abs{E(G)} - \abs{V(G)}$. Using $\kappa(G)$ for the number of connected components of $G$, we have that $G$ is a forest iff $\delta(G) = - \kappa(G)$ (while for every graph $G$ we have $\delta(G) \ge - \kappa(G)$). Now consider any non-empty set $V'$ of variable-nodes in the above forest $\mc{F}$ (that actually we have a tree is irrelevant) together with the induced sub-graph $\mc{F}'$ given by $V' \cup \Gamma_{\mc{F}}(V')$. Since also $\mc{F'}$ is a forest we have $\delta(\mc{F}') \le - \kappa(\mc{F}') \le -1$, where $\abs{E(\mc{F}')} = 2 \abs{V'}$ and $\abs{V(\mc{F}')} = \abs{V'} + \abs{\Gamma_{\mc{F}'}(V')}$, and thus $\abs{\Gamma_{\mc{F}'}(V')} \ge \abs{V'} + 1$, where $\abs{\Gamma_{\mc{F}'}(V')} \le \delta(F[V']) + \abs{V'}$, so that $\delta(F[V']) \ge 1$. Since this holds for every non-empty $V'$, it follows that $F$ can not have a non-trivial matching autarky.}

Finally we mention that a natural class of (generalised) multi-clause-sets, where every (general) autarky is a matching autarky, is the class of multi-clause-sets $F$ such that for all variables $v \in \var(F)$ and all values $\ve \in D_v$ we have $\#_{(v,\ve)}(F) \le 1$; we call such (generalised) clause-sets \textbf{totally singular}. Thus totally singular $F$ are satisfiable iff they are matching satisfiable; an alternative method for quick SAT decision here is given by the observation that the class of totally singular multi-clause-sets is stable under DP-reduction (recall Subsection \ref{sec:Reduction}), and furthermore every variable is either pure or a singular variable, and thus elimination of pure variables together with repeated application of singular DP-reduction also solves this formula class. We remark that the minimally unsatisfiable $F$ in this class are exactly the marginal minimally unsatisfiable clause-sets of deficiency $1$, as characterised in Corollary \ref{cor:CharakMMUSATd1} (Part II).

\subsection{Comparison with an earlier version of ``matching autarkies''}
\label{sec:comparisonearlier}

In \cite{Ku01a} an earlier version of matching autarkies has been introduced, which we will call here ``non-repetitive matching autarkies'': A partial assignment $\vp$ is called \emph{non-repetitive matching-satisfying} for a multi-clause-set $F \in \Mcls$, if for every clause-occurrence $C$ in $F$ (taking multiple occurrences into account) a literal $x_C \in C$ can be chosen with $\vp(x_C) = 1$ such that for different clause-occurrences $C, C'$ we have $x_C \not= x_{C'}$. And $\vp$ is called a \emph{non-repetitive matching autarky} for $F$ if $\vp$ is non-repetitive matching-satisfying for $F_{\var(\vp)}$.

Recalling the three conditions (\ref{item:condmatchsat1}) - (\ref{item:condmatchsat3}) from Subsection \ref{sec:matchsat} and strengthening condition (\ref{item:condmatchsat1}) to
\begin{enumerate}
\item[(\ref{item:condmatchsat1})'] for $i \in \tb 1m$ there are variables $v_i \in \var(F_i)$ such that for all clause-occurrence $C$ in $F_i$ there are literals $x_C \in C$ with $\var(x_C) = v_i$, and such that for different clause-occurrences $C, C'$ we have $x_C \not= x_{C'}$;
\end{enumerate}
we get that $F$ is non-repetitive matching satisfiable iff $F$ has a decomposition fulfilling (\ref{item:condmatchsat1})', (\ref{item:condmatchsat2}) and (\ref{item:condmatchsat3}). By (\ref{item:condmatchsat1})' we get $c(F_i) = \abs{\val_{v_i}(F_i)}$, and thus from (\ref{item:condmatchsat3}) follows (\ref{item:condmatchsat3})'. Whence a non-repetitive matching-satisfying assignment $\vp$ for $F$ is matching-satisfying for $F$, and a non-repetitive matching autarky for $F$ is also a matching autarky for $F$. 

For boolean clause-sets, non-repetitive matching autarkies are identical with matching autarkies, but in general non-repetitive matching autarkies are more restrictive than matching autarkies. Consider a variable $v$ with $D_v = \set{0,1,2}$, and consider the multi-clause-set $F_1 := 2 \cdot \set{v \not= 0}$. $F_1$ is matching satisfiable, while being lean w.r.t.\ non-repetitive matching autarkies. Neither pureness of $v$ nor the existence of unit-clauses is essential here, as the following extended example shows (using additional boolean variables $a,b,c,d$, while $D_v = \set{0,1,2,3}$):
\begin{multline*}
  F_1' := \setb{ \set{v \not= 0, a, b}, \set{v \not= 0, a, \ol{b}}, \set{\ol{a},b}, \set{\ol{a},\ol{b}},\\
    \set{v \not= 1, c, d}, \set{v \not= 2, \ol{c}, \ol{d}}, \set{v \not= 3, c, d} }.
\end{multline*}
$F_1'$ is matching satisfiable via $\pab{v \ra 1, a \ra 0, b \ra 0, c \ra 1, d \ra 0}$, while $F_1'$ is again lean w.r.t.\ non-repetitive matching autarkies.

Further examples are discussed in Subsection \ref{sec:Preservationmatchingstructure} (Part II). There we actually show that non-repetitive matching autarkies are preserved by the direct translation of (generalised) clause-sets into boolean clause-sets, which in general is not the case for matching autarkies, and so perhaps non-repetitive matching autarkies nevertheless might have some advantages over matching autarkies? The main problem with the notion of non-repetitive matching autarkies is that it does not seem to support a natural notion of related deficiency (with the same nice properties as the combination of matching autarkies and (standard) deficiency), and, related to this problem, it does not seem obvious how to achieve polynomial time decision of the class of non-repetitive matching lean (multi-)clause-sets. The whole problem boils down to the point that non-repetitive matching autarkies do not seem to be given solely by a matching condition, but require some other form of a more global  condition. Thus, to conclude, the generalisation of (boolean) matching autarkies together with the generalisation of (boolean) deficiency introduced in this section seems to be the right choice, as demonstrated by the theory built up in this section, and as further validated by the applications in the following sections.

\section{Finding autarkies in poly-time for bounded maximal deficiency}
\label{sec:leanKernelpoly}

We already know that the satisfiability problem is decidable in polynomial time for bounded maximal deficiency (Corollary \ref{cor:poly1}). In this section we strengthen this by computing a maximal autarky in polynomial time (for bounded maximal deficiency). Several first attempts fail to solve the problem:
\begin{itemize}
\item In Sections 11.10.1 -- 11.10.3 of \cite{Kullmann2007HandbuchMU} general methods for computing the lean kernel are discussed. Especially the general method from Corollary \ref{cor:polytimeLKern} (in the present article) could yield a poly-time algorithm, since for bounded maximal deficiency we can find either a satisfying assignment or a resolution refutation tree in polynomial time (the latter follows by the methods used to prove Theorem \ref{thm:MaximalerDefektFPT}) (Part II).
\item But the problem is that crossing out variables can increase the maximal deficiency, and so actually these methods seem to fail for bounded maximal deficiency.
\item Based on matroid theory, in \cite{Ku99dKo} it has been shown how to compute the lean kernel in polynomial time for boolean clause-sets, which via the direct translation to boolean clause-sets (see Subsection \ref{sec:Preservationgeneralstructure} in Part II) can be generalised to non-boolean clause-sets as considered here.
\item However we do not obtain an autarky, and self-reducibility (\cite{KullmannMarekTruszczynski2007Autarkien}) seems not to be applicable --- again the problem is that the maximal deficiency might be increased. 
\end{itemize}
 Fortunately we are able to strengthen the method from Corollary \ref{cor:poly1}, which will be done in the following. First we have to strengthen Corollary \ref{cor:Mleanpoly2}, by efficiently computing for multi-clause-sets $F$ a partial assignment $\vp$ with $\vp * F = \nma(F)$. Again, for finding such a ``quasi-maximal'' matching autarky $\vp$ for $F$ the general procedure given in \cite{KullmannMarekTruszczynski2007Autarkien} (just based on deciding matching leanness, via an oracle) can not be used, since the matching condition does not allow addition of unit-clauses. The following method was already mentioned in Lemma \ref{lem:vonschlanknachred}, but for completeness we state this special case explicitly.

\begin{lem}\label{lem:polytimeqmma}
  For a generalised multi-clause-sets $F \in \Mcls$ we can compute a matching autarky $\vp$ for $F$ with $\vp * F = \nma(F)$ in polynomial time as follows:
  \begin{itemize}
  \item First compute $\nma(F)$ by Corollary \ref{cor:Mleanpoly2}.
  \item Thus by Lemma \ref{lem:autdecomp} there is a matching autarky $\vp$ for $F$ with $\vp * F = \nma(F)$, namely $\vp$ is any matching-satisfying assignment for $F[V]$ of $V := \var(F) \sm \var(\nma(F))$ (where $\var(\vp) \sse V$), and those can be found by Lemma \ref{lem:matchsatviamatchsa}.
  \end{itemize}
\end{lem}

Second, we need to strengthen the general decomposition result Lemma \ref{lem:autdecomp} by considering the (maximal) deficiency of the satisfiable part of the decomposition:
\begin{lem}\label{lem:strengthdecomp}
  Consider a normal autarky system $\A$, and the decomposition $F = F_1 + F_2$ of a generalised multi-clause-set $F \in \Mcls$ according to Lemma \ref{lem:autdecomp} into the $\A$-lean kernel $F_1 = \nA(F)$ and the largest $\A$-autark sub-clause-set $F_2 = F_V$ for $V := \var(F) \sm \var(F_1)$. Assume furthermore that $F_1$ is matching lean. Then we have $\delta^*(F[V]) \le \delta^*(F)$, and if $F_1 \not= \top$, then $\delta^*(F[V]) < \delta^*(F)$ holds.
\end{lem}
\begin{proof}
  Consider $F' := F_1$, and assume that $F_1 \not= \top$. Lean kernels w.r.t.\ autarky systems are always induced sub-multi-clause-sets, and thus by Lemma \ref{lem:tightdef}, Part \ref{item:tightdef3} we have $\delta^*(\var(F') * (F - F')) \le \delta^*(F) - \delta(F')$, where $\var(F') * (F - F') = \var(F_1) * F_2 = F[V]$, and $\delta(F') = \delta(F_1) \ge 1$ by Corollary \ref{cor:TarsisMLean}.
\end{proof}

Now, strengthening Corollary \ref{cor:poly1}, we are not just able to find a satisfying assignment in polynomial time (if existent) for bounded maximal deficiency, but we can also find a non-trivial autarky (if existent):
\begin{thm}\label{thm:leanKernelpoly}
   Consider a constant $k \in \NNZ$. Then for a generalised multi-clause-set $F \in \Mcls$ with $\delta^*(F) \le k$ in polynomial time we can find a non-trivial autarky if existent (and thus, by repetition, we can compute $\nma(F)$ in polynomial time). The procedure for finding a non-trivial autarky is as follows:
   \begin{enumerate}
   \item As in Corollary \ref{cor:poly1} we run through all non-trivial partial assignments $\vp$ with $\var(\vp) \sse \var(F)$ and $n(\vp) \le \delta^*(F)$.
   \item Instead of just checking whether $\vp * F$ is matching satisfiable, by Lemma \ref{lem:polytimeqmma} we compute a partial assignment $\psi$ with $\psi * (\vp * F) = \nma(\vp * F)$ (and $\var(\psi) \sse \var(\vp * F)$).
   \item Then we check whether $\psi \circ \vp$ is a non-trivial autarky for $F$ --- if this is the case, then we found a non-trivial autarky for $F$, while if this is never the case, then $F$ is lean.
   \end{enumerate}
\end{thm}
\begin{proof}
  Let $V := \var(F) \sm \var(\na(F))$. So $F[V]$ is satisfiable, and by Corollary \ref{cor:Abstandsat} there exists a partial assignment $\vp$ with $n(\vp) \le \delta^*(F[V])$ and $\var(\vp) \sse V$ such that $\vp * F[V]$ is matching satisfiable; note that we read $\vp * F[V]$ as $\vp * (F[V])$, and that $\vp * (F[V]) = (\vp * F)[V]$ holds. By Lemma \ref{lem:strengthdecomp} we know that $\delta^*(F[V]) \le \delta^*(F)$, and thus $\vp$ will be considered by the procedure, and the assignment $\psi$ computed satisfies $(\vp * F)[V]$. So we have a satisfying assignment $\psi \circ \vp$ for $F[V]$, i.e., we have an autarky $\psi \circ \vp$ for $F$ with $\psi * (\vp * F) = \na(F)$.
\end{proof}

As the proof of Theorem \ref{thm:leanKernelpoly} shows, by running through all $\vp$ we will actually find a quasi-maximal autarky $\theta$ (with $\theta * F = \na(F)$), however it is somewhat easier to stop when the first non-trivial autarky has been found.

\section{Expansion and the surplus}
\label{sec:expansion}

In this section the notions of ``expansion'' and ``surplus'' are studied. The basic question is: if we consider arbitrary non-empty sets of variables and all clauses containing at least one of them, how many more clauses than variables do we have at least? This minimum (over all ``expansions'') is the surplus. Having at least a surplus of $1$ turns out to be equivalent to being matching lean (Lemma \ref{lem:surpone}).

One important application of the surplus is in establishing that by fixing one variable to a value, the maximal deficiency must get smaller, and this matter is discussed in Subsection \ref{sec:expansionreductions} (generalising the results from \cite{Szei2002FixedParam}). We obtain a simplified proof of fixed-parameter tractability of satisfiability in the maximal deficiency $\delta^*(F)$ for boolean clause-sets, using the method of ``bounded search trees''. The main tool for this application is a poly-time reduction $S: \Mcls \ra \Mcls$ (see Lemma \ref{lem:Sred}) with the properties, that the maximal deficiency is not increased while a surplus of at least $2$ is established. In general we have for all variables $v \in \var(F)$ and all $\ve \in D_v$ the upper bound 
\begin{displaymath}
 \delta^*(\pao{v}{\ve} * S(F)) \le \delta^*(S(F)) - 1
\end{displaymath}
(see Corollary \ref{cor:HinrKritAbbauDef}), \emph{given} that the surplus of $F$ is at least as big as the domain size of $v$ (and some further reduction condition holds, namely $F$ does not contain singular variables). Thus by a trivial DLL branching algorithm, using additionally only the reduction $F \mapsto S(F)$ at each node, we obtain SAT decision for $F$ in time $2^{\delta^*(F)} \cdot \mathrm{poly}(\ell(F))$ for boolean clause-sets. Later, in Theorem \ref{thm:MaximalerDefektFPT} (Part II), we show how to achieve the same bound also for generalised clause-sets, by exploiting the direct translation to boolean clause-sets.

\subsection{Basic properties}
\label{sec:expansionbasics}

For multi-clause-sets $F$ we have defined the bipartite graph $B(F)$ together with its canonical bipartition $(\ul{F}, \ul{V})$. The general definition of deficiency (for arbitrary graphs) then yields the deficiency $\delta_{B(F)}(F') = \abs{F'} - \abs{\Gamma_{B(F)}(F')}$ for sets $F' \sse \ul{F}$ of clause-nodes, as well as the deficiency $\delta_{B(F)}(V') = \abs{V'} -  \abs{\Gamma_{B(F)}(V')}$ for sets $V' \sse \ul{V}$ of variable-nodes.\footnote{Recall that $\Gamma_G(A)$ for a graph $G$ and a vertex-set $A$ is the set of all neighbours of $A$ in $G$.} So we have a ``clause-based deficiency'' as well as a ``variable-based deficiency''. Identifying $F'$ with a sub-multi-clause-set of $F$, we have $\abs{F'} = c(F')$ and $\abs{\Gamma_{B(F)}(F')} = \rd(F')$, and thus $\delta_{B(F)}(F')$ (the graph-theoretical deficiency of sets of clause-nodes) is the same as the deficiency of multi-clause-sets as we have defined it in Subsection \ref{sec:matchsat}. At first sight, the situation for $V'$ seems not to be naturally interpretable on the level of multi-clause-sets, since $V'$ may contain for some variable $v$ only some of the $\abs{D_v}-1$ copies of $v$. To consider this problem, let $V''$ the the set of variable-nodes obtained from $V'$ by adding for $(v,i) \in V'$ all $(v,j)$ for $j \in \tb{1}{\abs{D_v}-1}$. Now we have
\begin{displaymath}
  \delta_{B(F)}(V'') = - \delta(F[V'_0]),
\end{displaymath}
where $V_0$ is the set of variables in $V'$ (or $V''$). Since $\delta_{B(F)}(V') \le \delta_{B(F)}(V'')$ due to $\Gamma_{B(F)}(V') = \Gamma_{B(F)}(V'')$, we see that actually, since we are only interested in \emph{maximising} the deficiency, also the variable-based deficiency has a sensible interpretation at the (conceptual) level of multi-clause-sets. Now two changes are applied to the variable-based deficiency, resulting in the notion of ``expansion'', related to the deficiency of a set of variables (analogous to the deficiency of a (sub-)multi-clause-set), and in the notion of ``surplus'' related to the maximal deficiency over all sets of variables (analogous to the maximal deficiency of a multi-clause-set). The first change is just to switch signs, so that we can use $\delta(F[V'_0])$ instead of $- \delta(F[V'_0])$. More substantially, we exclude the empty set of variables for the surplus: The maximal deficiency $\delta^*(F)$ of a multi-clause-set is only used to determine the size of a maximal matching in $B(F)$, and so negative deficiencies are not of interest (they indicate that a bigger matching number is possible --- if only there would be more clauses), whence the empty clause-set is taken into account in $\delta^*(F)$ for convenience, to force the maximal deficiency to be at least $0$. But now for the notion of surplus actually we are only interested in the negative values, that is, in the ``surplus'' which can not be realised, and thus the empty set of variables has to be excluded. After these motivations, let us now start with the formal definitions.

For a multi-clause-set $F$ and a set $V$ of variables let the \textbf{expansion} be defined as $\delta(F[V])$. As explained above, using the deficiency $\delta_{B(F)}$ in the (bipartite) graph $B(F)$, the expansion equals $- \delta_{B(F)}(V')$, where $V'$ is the set of variable nodes of $B(F)$ associated with some variable in $V$. The \textbf{surplus} of $F$ is defined as
\begin{displaymath}
  \bmm{\surp(F)} := 
  \begin{cases}
    \min_{\es \not= V \sse \var(F)} \delta(F[V]) & \text{if } \var(F) \not= \es\\
    0 & \text{if } \var(F) = \es
  \end{cases}.
\end{displaymath}
The surplus of $F$ equals the surplus of $B(F)$ as defined in Subsection 1.3 of \cite{LP86} (but with the sides of the bipartition switched(!)). By definition we have $\surp(F) \le \delta(F[\var(F)]) = \delta(F) - F(\bot) \le \delta(F)$. Generalising Lemma 7.7 in \cite{Ku00f} we have:
\begin{lem}\label{lem:surpone}
  A generalised multi-clause-set $F \in \Mcls$ with $n(F) > 0$ is matching lean if and only if $\surp(F) \ge 1$. More specifically, for any generalised multi-clause-set $F$ we have:

  For $\es \not= V \sse \var(F)$ in case of $\delta(F[V]) \le 0$ the generalised multi-clause-set $F[V]$ has a non-trivial matching autarky; such a non-trivial matching autarky yields a non-trivial matching autarky for $F$, and every non-trivial matching autarky of $F$ can be obtained in this way.
\end{lem}
\begin{proof} If $F$ is matching lean, then by Lemma \ref{lem:matchautnormal} for $V \sse \var(F)$ also $F[V]$ is matching lean, and thus by Corollary \ref{cor:TarsisMLean} we get $\surp(F) \ge 1$. If on the other hand $F$ is not matching lean, then there exists $\es \not= V \sse \var(F)$ such that $F[V]$ is matching satisfiable (where $V$ is the variable set of any non-trivial matching autarky), i.e., $\delta^*(F[V]) = 0$, and thus $\surp(F) \le 0$. \end{proof}

By Theorem 1.3.8 in \cite{LP86} we obtain the following characterisation of matching lean clause-sets:
\begin{corol}\label{cor:CharakMLEANWald}
  A generalised multi-clause-set $F$ is matching lean if and only if there exists a subgraph $\mc{F}$ of $B(F)$ with the following properties:
  \begin{enumerate}[(i)]
  \item $\mc{F}$ is a forest;
  \item $\mc{F}$ covers all variable-nodes;
  \item every variable-node has degree $2$ in $\mc{F}$.
  \end{enumerate}
\end{corol}
An example for the application of Corollary \ref{cor:CharakMLEANWald} has been given at the end of Subsection \ref{sec:matchautsub}. The problem of computing $\surp(F)$ can be solved by Theorem 1.3.6 in \cite{LP86} as follows (compare Lemma 15 in \cite{Szei2002FixedParam} for the case of boolean clause-sets):
\begin{lem}\label{lem:compsurp}
  Consider a generalised multi-clause-set $F$.
  \begin{enumerate}
  \item\label{lem:compsurp1}  Let $M$ be a maximum matching for $B(F)$. If $M$ does not cover all variable-nodes, then $\surp(F) = \abs{M} - \rd(F) < 0$ (that is, $\sigma(F)$ is the number of uncovered variable-nodes, multiplied by $-1$). Otherwise we have $\surp(F) \ge 0$.
  \item\label{lem:compsurp2} Assume $\surp(F) \ge 0$, and consider $s \in \NNZ$, $s \le \delta(F)$. For $v \in \var(F)$ with $\abs{D_v} \ge 2$ (trivial variables are ignored) let $F_{s,v}$ be the multi-clause-set obtained from $F$ by adding $s$ new elements to the domain of $v$ (that is, a new variable $v'$ with $D_v \sse D_{v'}$ and $\abs{D_{v'}} = \abs{D_v} + s$ is chosen, and $F_{s,v}$ is obtained by replacing $v$ by $v'$ in $F$). Let $M_{s,v}$ be a maximum matching for $B(F_{s,v})$. Then we have:
    \begin{enumerate}
    \item If $M_{s,v}$ does not cover all variable-nodes in $B(F_{s,v})$, then $\surp(F) < s$, and moreover, from $M_{s,v}$ in linear time in $\ell(F)$ a set $\es \not= V \sse \var(F)$ with $\delta(F[V]) < s$ can be computed .
    \item If for all $v$ the maximum matching $M_{s,v}$ covers all variable-nodes in $B(F_{s,v})$, then $\surp(F) \ge s$.
    \end{enumerate}
  \end{enumerate}
\end{lem}

Lemma \ref{lem:surpone} together with Lemma \ref{lem:compsurp} yields an alternative to the decision procedure for matching leanness as given in Corollary \ref{cor:Mleanpoly}.

\subsection{Decreasing the maximal deficiency}
\label{sec:decmaxdef}

The next lemma tackles the problem of giving a sufficient criterion for $\delta^*(\pao{v}{\ve} * F) < \delta^*(F)$; Part \ref{item:HauptanwendungSurplus2} generalises Lemma 7.10 in \cite{Ku00f} (the proof there is technically not fully correct).
\begin{lem}\label{lem:HauptanwendungSurplus}
  Consider a generalised multi-clause-set $F \in \Mcls$.
  \begin{enumerate}
  \item\label{item:HauptanwendungSurplus0} For $F' \le F$ we have 
    \begin{displaymath}
      \delta(F') = \delta(F) - \delta(F[\var(F) \sm \var(F')]) - \eta(F, F')\\ \le \delta(F) - \delta(F[\var(F) \sm \var(F')]),
    \end{displaymath}
    where $\eta(F, F') := c(F) - c(F_{\var(F) \sm \var(F')}) - c(F') \ge 0$ is the number of clause occurrences in $F$ of clauses $C$ with $\var(C) \sse \var(F')$ but not occurring in $F'$.
  \item\label{item:HauptanwendungSurplus1} For $F' \le F$ we have $\delta(F') \le \delta(F) - \min(c(F) - c(F'), \,\surp(F))$.
  \item\label{item:HauptanwendungSurplus2} Consider $v \in \var(F)$. Then for $\ve \in D_v$ we have
    \begin{displaymath}
       \delta^*(\pao{v}{\ve} * F) \le \delta(F) - \min(s_{(v,\ve)}(F), \surp(F)) + \abs{D_v} - 1.
     \end{displaymath}
     Thus, using $m_v(F) := \min_{\ve \in D_v} s_{(v,\ve)}(F)$, we obtain
    \begin{displaymath}
      \delta^*(\pao{v}{\ve} * F) \le \delta(F) - \min(m_v(F), \surp(F)) + \abs{D_v} - 1.
    \end{displaymath}
  \end{enumerate}
\end{lem}
\begin{proof} For Parts \ref{item:HauptanwendungSurplus0} and \ref{item:HauptanwendungSurplus1} let $V := \var(F) \sm \var(F')$.

The equation in Part \ref{item:HauptanwendungSurplus0} follows immediately with $\delta(F) = c(F) - \rd(F)$ and $\delta(F[V]) = c(F[V]) - \rd(F[V])$, where due to $\var(F[V]) = V$ we have $\rd(F[V]) = \rd(F) - \rd(F')$. And that $\eta(F, F') \ge 0$ holds follows with the explanation given.

For Part \ref{item:HauptanwendungSurplus1} we consider two cases. If $n(F') = n(F)$, then $\delta(F') = \delta(F) - (c(F) - c(F'))$. So assume $n(F') < n(F)$ (and thus $V \not= \es$). By Part \ref{item:HauptanwendungSurplus0} we have $\delta(F') \le \delta(F) - \delta(F[V])$, and thus by $\delta(F[V]) \ge \surp(F)$ we get $\delta(F') \le \delta(F) - \surp(F)$.

For Part \ref{item:HauptanwendungSurplus2} consider an induced $F' \le \pao{v}{\ve} * F$, and let $F'' \le F$ be the unique sub-multi-clause-set of $F$ with $c(F'') = c(F')$ and $\pao{v}{\ve} * F'' = F'$. We have $\rd(F'') \le \rd(F') + \abs{D_v} - 1$, and thus $\delta(F'') \ge \delta(F') - \abs{D_v} + 1$. 

By part \ref{item:HauptanwendungSurplus1} we get $\delta(F'') \le \delta(F) - \min(c(F) - c(F''), \surp(F))$, where $c(F'') \le c(F) - s_{(v,\ve)}(F)$, and the assertion follows. \end{proof}

Considering a matching lean boolean clause-set $F$, Part \ref{item:HauptanwendungSurplus2} of Lemma \ref{lem:HauptanwendungSurplus} yields the upper bound $\delta^*(\pao{v}{\ve} * F) \le \delta(F)$ for non-pure variables $v$ (using Lemma \ref{lem:surpone}); for algorithmic purposes we are interested in cases where the maximal deficiency actually decreases:

\begin{corol}\label{cor:HinrKritAbbauDef}
  Consider a generalised clause-set $F$ and a variable $v \in \var(F)$ with $m_v(F) \ge \abs{D_v}$ (using the definition of $m_v(F)$ from Lemma \ref{lem:HauptanwendungSurplus}, Part \ref{item:HauptanwendungSurplus2}) and $\surp(F) \ge \abs{D_v}$. Then for each $\ve \in D_v$ we have $\delta^*(\pao{v}{\ve} * F) \le \delta(F) - 1$.
\end{corol}
It is unclear, whether Corollary \ref{cor:HinrKritAbbauDef} is best possible --- the condition $\surp(F) \ge \abs{D_v}$ is hard to establish for larger domain sizes. The key seems to be to improve the estimation used in the proof of Lemma \ref{lem:HauptanwendungSurplus}, Part \ref{item:HauptanwendungSurplus2}.

\subsection{Reductions}
\label{sec:expansionreductions}

In this final subsection we study how to establish the prerequisites of Corollary \ref{cor:HinrKritAbbauDef}. By singular DP-reduction we can eliminate cases with $m_v < \abs{D_v}$ as follows.
\begin{lem}\label{lem:dpausnahme}
  Consider a generalised multi-clause-set $F \in \Mcls$ and a variable $v \in \var(F)$.
  \begin{enumerate}
  \item\label{lem:dpausnahme0} Assume $v$ is not pure in $F$. Then $v$ is a singular variable for $F$ if and only if $m_v(F) < \abs{D_v}$ (where $m_v(F)$ is as defined in Part \ref{item:HauptanwendungSurplus2} of Lemma \ref{lem:HauptanwendungSurplus}).
  \item Assume $v$ is a singular variable.
    \begin{enumerate}
    \item\label{lem:dpausnahmea} $\delta(\dpi{v}(F)) \le \delta(F)$.
    \item\label{item:dpausnahmeb} If $v$ is non-degenerated, then we have
      \begin{enumerate}
      \item\label{item:dpausnahmeb1} $\delta(\dpi{v}(F)) = \delta(F)$.
      \item\label{item:dpausnahmeb2} $F$ is matching lean if and only if $\dpi{v}(F)$ is matching lean.
      \item\label{item:dpausnahmeb3} If $F$ is matching lean, then we have $\delta^*(\dpi{v}(F)) = \delta^*(F)$.
      \end{enumerate}
    \end{enumerate}
  \end{enumerate}
\end{lem}
\begin{proof} The only non-obvious assertion is \ref{item:dpausnahmeb2} (using the remarks made before Lemma \ref{lem:singDPMUSAT}, and Lemma \ref{lem:charakschlank}). Let $D_v = \set{\ve_1, \dots, \ve_k}$ ($k = \abs{D_v}$), and assume w.l.o.g.\ that $\#_{(v,\ve_i)}(F) = 1$ for $i < k$; consider $C_1, \dots, C_{k-1} \in F$ with $(v,\ve_i) \in C_i$ for $i < k$, and let $D_1, \dots, D_m \in F$ be the clauses containing $(v,\ve_k)$ ($m = \#_{(v,\ve_k)}(F)$). Thus $F_{\set{v}} = \set{C_1, \dots, C_{k-1}} + \sum_{i=1}^m \set{D_i}$. Now with $F' := F - F_{\set{v}}$ we have
\begin{displaymath}
  \dpi{v}(F) = F' + R,
\end{displaymath}
where $R := \sum_{i=1}^m \setb{ R_i }$ and $R_i := C_1 \cup \dots \cup C_{k-1} \cup D_i$ for $i \in \tb 1m$. First assume that $F$ is matching lean, but that we have a non-trivial matching autarky $\vp$ for $\dpi{v}(F)$ with $\var(\vp) \sse \var(\dpi{v}(F)) = \var(F) \sm \set{v}$. Let $V := \var(\vp)$. If $V \cap \var(R) = \es$, then $\vp$ would also be a matching autarky for $F$, since $\var(R) = \var(F_{\set{v}}) \sm \set{v}$. So assume $V \cap \var(R) \not= \es$. If there exists $i \in \tb 1m$ with $\var(D_i) \cap V = \es$, then for all $j < k$ we have $\var(C_j) \cap V = \es$, and it follows that $\vp$ would also be a matching autarky for $F$. So assume that for all $i \in \tb 1m$ we have $\var(D_i) \cap V \not= \es$. Now $\pao{v}{\ve_k} \circ \vp$ is a matching autarky for $F$, contradicting matching leanness of $F$.

For the reverse direction assume that $\dpi{v}(F)$ is matching lean, but that we have a non-trivial matching autarky $\vp$ for $F$. Now it is not hard to see that $\vp$ is also a matching autarky for $\dpi{v}(F)$. \end{proof}

\begin{corol}\label{cor:RedsingVar}
  There is a polynomial time computable map $r: \Mcls \ra \Mcls$, such that for a generalised multi-clause-set $F$ we have:
  \begin{enumerate}[(i)]
  \item $n(r(F)) \le n(F)$, $c(r(F)) \le c(F)$ and $\delta^*(r(F)) \le \delta^*(F)$.
  \item $r(F)$ is satisfiability-equivalent to $F$.
  \item $r(F)$ is matching lean.
  \item $r(F)$ is lean w.r.t.\ pure autarkies (i.e., $r(F)$ does not contain pure variables).
  \item $r(F)$ does not contain singular variables.
  \end{enumerate}
  Computation of $r(F)$ is as follows:
  \begin{enumerate}
  \item\label{item:Reduktion11} Apply singular DP-degeneration reduction and reduction by pure autarkies and matching autarkies as long as possible.
  \item If there exists a singular variable, then it must be non-degenerated, thus applying DP-reduction does not increase the maximal defect by Part \ref{item:dpausnahmeb3} of Lemma \ref{lem:dpausnahme}, so apply this reduction and go to Step \ref{item:Reduktion11}. Otherwise output $r(F)$ and stop.
  \end{enumerate}
\end{corol}

The next lemma contains the main idea for establishing a surplus of two.
\begin{lem}\label{lem:etabsurptwo}
  Consider a multi-clause-set $F \in \Mlean$, such that for all variables $v \in \var(F)$ we have $\#_v(F) \ge \abs{D_v} + 1$, and assume that $V \sse \var(F)$ is given with $\delta(F[V]) = 1$. Then $F[V]$ is satisfiable, and a satisfying assignment $\vp$ with $\var(\vp) \sse \var(F[V])$ can be found in polynomial time. (With $\vp$ thus we have found a non-trivial autarky for $F$.)
\end{lem}
\begin{proof} Assume that $F[V]$ is unsatisfiable. By Corollary \ref{cor:musatonemlean} thus we have $F[V] \in \Musati{\delta=1}$ which would contain a variable occurring in all signs exactly once (see the later Lemma \ref{lem:GrundtatsacheDefektEins} in Part II, whose proof does not make use of the results of this subsection), contradicting the assumption. So $F[V]$ is satisfiable, and a satisfying assignment $\vp$ can be found by Corollary \ref{cor:poly1}. \end{proof}

Strengthening Corollary \ref{cor:RedsingVar} we get now:
\begin{lem}\label{lem:Sred}
  There is a polynomial time computable map $S: \Mcls \ra \Mcls$ such that for a generalised multi-clause-set $F$ we have:
  \begin{enumerate}[(i)]
  \item $n(S(F)) \le n(F)$, $c(S(F)) \le c(F)$ and $\delta^*(S(F)) \le \delta^*(F)$.
  \item $S(F)$ is satisfiability-equivalent to $F$.
  \item $S(F)$ is matching lean and lean w.r.t.\ pure autarkies, and does not contain singular variables.
  \item If $\var(S(F)) \not= \es$, then $\surp(S(F)) \ge 2$.
  \end{enumerate}
  Computation of $S(F)$ is as follows:
\begin{enumerate}
\item\label{item:Reda1} First reduce $F := r(F)$ (see Corollary \ref{cor:RedsingVar}).
\item If $\var(F) = \es$ or $\surp(F) \ge 2$ then stop.
\item Otherwise find some $\es \not= V \sse \var(F)$ with $\delta(F[V]) = \surp(F) = 1$ by Lemma \ref{lem:compsurp}, Part \ref{lem:compsurp2}; by Lemma \ref{lem:etabsurptwo} we can now find a non-trivial autarky $\vp$ for $F$: reduce $F := \vp * F$, and go to Step \ref{item:Reda1}.
\end{enumerate}
\end{lem}

If we only allow boolean clause-sets, then, as explained at the beginning of this subsection, we obtain fixed-parameter tractability of satisfiability decision w.r.t.\ the parameter $\delta^*(F)$ by Lemma \ref{lem:Sred} and Corollary \ref{cor:HinrKritAbbauDef} (together with Lemma \ref{lem:dpausnahme}, Part \ref{lem:dpausnahme0}). Finally we mention that a good possibility for further improvements is to generalise and strengthen the approach from \cite{Ku99dKo} based on matroid theory.

\section{Conclusion and open problems}
\label{sec:open1}

The first purpose of this article was to set the stage for the study of generalised clause-sets as sets of ``no-goods'', where literals are given by one ``forbidden value'': We defined and summarised the basic properties of syntax, semantics, resolution calculus and autarky systems. Then we considered the generalisation of the notion of deficiency for these generalised clause-sets, and we studied the basic autarky system related to this notion, namely matching autarkies.

\subsection{Matching autarkies}
\label{sec:openApplyingmatching}

Basic tasks are the computation of the matching lean kernel and the computation of a quasi-maximal matching autarky (realising the matching lean kernel). In this article we did not explore further algorithmic details, and so the computation of the matching lean kernel in Corollary \ref{cor:Mleanpoly2} just uses the general procedure of Lemma \ref{lem:vonschlanknachred}, based on the simple decision procedure for matching lean clause-sets from Corollary \ref{cor:Mleanpoly} (alternatively, Lemma \ref{lem:surpone} together with Lemma \ref{lem:compsurp} can be used); the computation of a quasi-maximal matching autarky in Lemma \ref{lem:polytimeqmma} in turn is based on Corollary \ref{cor:Mleanpoly2} in the obvious way. It should be possible to use the characterisation of the matching lean kernel in Corollary \ref{cor:leankerntight} (as the smallest tight sub-clause-set) for an alternative procedure for computing $\nma(F)$ directly, exploiting matching theory.

Many questions are still open regarding polynomial-time satisfiability decision in the maximal deficiency (for generalised clause-sets):
\begin{enumerate}
\item In Corollary \ref{cor:poly1} polynomial-time satisfiability decision is achieved by using the existence of matching-maximum satisfying assignments. Is it possible to show also \emph{fixed-parameter tractability} in this way, using perhaps some form of local search which considers only matching-maximum assignments in some form?
\item Is it possible to directly generalise the approach which yields FPT for the boolean case, as given at the end of Section \ref{sec:expansion}?
\item In Section \ref{sec:leanKernelpoly} it was shown how to find a non-trivial autarky in polynomial time in the maximal deficiency, while the approach from \cite{Ku99dKo} based on matroid theory yields computation of the lean kernel in polynomial time in the maximal deficiency. The question here is, whether for these harder tasks we also have fixed-parameter tractability, or whether the problem of deciding leanness is fixed-parameter intractable in the maximum deficiency.
\item Implementations of the various poly-time algorithms here (especially the algorithm exploited in the proof of Theorem \ref{thm:MaximalerDefektFPT} in Part II) need to be carried out in order to study whether interesting applications exist (and also to judge whether ``native algorithms'' for generalised clause-sets are preferable here, or whether the boolean translation is superior).
\item As an alternative to the direct translation, can we also use the (new) nested translation from Subsection \ref{sec:newtrans} (Part II) here?
\end{enumerate}
 
Regarding  Lemma \ref{lem:etabsurptwo} one should study the underlying autarky system, and whether there is an underlying hierarchy of autarky systems (and accompanying reductions).

W.r.t.\ practical applications, one can add the following:
\begin{enumerate}
\item One needs to implement (and to study) reduction by matching autarkies, as a preprocessing step and also at each node of the search tree of a DPLL solver.
\item The extension by Lemma \ref{lem:etabsurptwo} might have potential.
\end{enumerate}
In general it seems that applications of ``expensive'' algorithms might be more fruitful for harder problems like QBF (where autarkies for quantified boolean formulas consider only the existential variables, and in the most general setting substitute them with boolean functions in the preceeding universal variables such that every clause becomes a tautology).

\subsection[Satisfiable min\-i\-mal\-ly matching-unsatisfiable clause-sets]{Satisfiable clause-sets which are min\-i\-mal\-ly matching-unsatisfiable}
\label{sec:openminmatchuns}

The class of minimally matching-unsatisfiable clause-sets is exactly $\Mleani{\delta=1}$, the set of matching lean clause-sets of deficiency one. The unsatisfiable elements we know quite well by $\Usat \cap \Mleani{\delta=1} = \Musati{\delta=1}$, which is characterised in Subsection \ref{sec:MUSATdefone} (Part II). However the satisfiable elements given by $\Sat \cap \Mleani{\delta=1}$ seem to exhibit a much richer structure: See for example the special clause-set at the end of Subsection \ref{sec:matchautsub}, which in fact is a $1$-regular hitting clause-set (see Section \ref{sec:construgcls} in Part II for further studies) --- even for this special sub-class of satisfiable matching lean $1$-regular hitting clause-sets we know not much.

\chapter{Minimal unsatisfiability and conflict structure}
\label{cha:II}

\section{Introduction}
\label{sec:intro2}

We conclude here the study of the basic properties of ``generalised clause-sets'', started in \cite{Kullmann2007ClausalFormZI} (see Chapter \ref{cha:I} in this report); please see there for the basic definitions. The main focus here is on irredundant and minimally unsatisfiable clause-sets, and on translations to boolean clause-sets (while the first part focused on the basic framework and on notions related to autarkies).

\subsection{Translations}
\label{sec:introTranslations}

Given that much more is known about boolean clause-sets than about generalised clause-sets, and that also basically all SAT solvers only accept boolean clause-sets, it is desirable to have a rich toolkit regarding translations from generalised clause-sets to boolean clause-sets. We start with the best-known translation in Section \ref{sec:translating}, the \emph{direct translation}, where for every generalised literal $(v,\ve)$ (meaning ``$v \not= \ve$''), consisting of variable $v$ and value $\ve \in D_v$, we have a corresponding boolean variable $\transl((v,\ve))$. Since we only have ``negative literals'', not positive literals ``$v = \ve$'', we don't need to enforce that variables get only one value. This \emph{weak form} of the direct translation is in most cases inferior to the \emph{strong form} (which uses the ``AMO''-constraints for ``at most one value'') w.r.t.\ SAT solving, however regarding theoretical purposes it is far easier to handle, since much of the structure of the clause-sets is preserved by the translation. For example this translation preserves the deficiency, and in Theorem \ref{thm:MaximalerDefektFPT} as an application we can lift fixed-parameter tractability in the maximal deficiency from boolean clause-sets to generalised clause-sets.

The direct translation (in the weak form) preserves basic autarky structures, however not the conflict structure, given by the conflict multigraph, which has as vertices the clauses, and as many edges between vertices (clauses) as they have clashing literals, that is, literals with the same underlying variable but with different values. In Subsection \ref{sec:newtrans} we introduce a new translation, the \emph{nested translation}, which preserves the deficiency as well as the conflict structure. For a variable $v$ with $k = \abs{D_v}$ values it uses $k-1$ boolean variables, while each value of $v$ is represented by one of the clauses of the unique (up to isomorphism) saturated minimally unsatisfiable Horn clause-set with $k-1$ variables (and $k$ clauses). This new translation is applied in Subsection \ref{sec:Regularhitting}, allowing to lift various fundamental results regarding the conflict structure of boolean clause-sets to generalised clause-sets.

Finally a very general translation, the \emph{generic translation}, is outlined in Subsection \ref{sec:opentranslations}. The generic translation covers all known translations, and allows to use arbitrary unsatisfiable boolean clause-sets to represent the values of the non-boolean variables (and this individually for every variable). The first empirical study \cite{Kullmann2010GreenTao} using these translations applied the instances of the generic translation as listed in Subsection \ref{sec:opentranslations} to the variation of van-der-Waerden problems (recall Subsection \ref{intro:hypergraphcolouring} in Part I) given by using the first $n$ prime numbers instead of the first $n$ natural numbers, computing \emph{Green-Tao numbers}, based on the Green-Tao Theorem (see \cite{GreenTao2005Primes}). The nested translation turned out to be always far superior over the direct translation, for all solver types, and only for larger domain sizes the logarithmic translation performed even better.

\subsection{Irredundant and minimally unsatisfiable clause-sets}
\label{sec:introirr}

``Irredundant'' (generalised) clause-sets $F$ are clause-sets where every clause contributes some unique falsifying assignment; if $F$ is also unsatisfiable, then one speaks of ``minimal unsatisfiability''. In Section \ref{sec:Irredundant} we provide the basic definitions and properties. In a sense the ``most'' irredundant clause-sets (see Corollary \ref{cor:charakhitirr}) are ``hitting clause-sets'', where all falsifying assignments are unique for each clause.\footnote{Considering DNF instead of CNF, in the literature on boolean functions boolean hitting clause-sets are also known as \emph{disjoint DNF}.} A larger class is given by ``multihitting'' clause-sets, which are actually in general no longer irredundant, but they have exactly one ``irredundant core'', which can be found quickly.

Minimally unsatisfiable clause-sets are then the subject of Section \ref{sec:MUSATneu}. Recall that in Part I (\cite{Kullmann2007ClausalFormZI}) the notion of \emph{deficiency} has been generalised, and its main properties have been established. Now we generalise the characterisation of the base-case of deficiency $1$ known from the boolean case.

In Section \ref{sec:construgcls} we return to the study of hitting clause-sets, now from the point of view of the ``conflict structure'', that is, considering the (multi-)graph of conflicts (``clashes'') between clauses. The main result is that via the so-called ``hermitian defect'' (obtained by considering quadratic forms, or, equivalently, positive and negative eigenvalues) we obtain an upper bound on the deficiency. For certain ``regular'' clause-sets the hermitian defect is easy to compute, and we can generalise various characterisations from the boolean case.

\subsection{Overview and main results}
\label{sec:introover2}

In \textbf{Section \ref{sec:translating}} the so-called direct translation of generalised clause-sets into boolean clause-sets is studied under the point of view of structure preservation, taking advantage of the fact that due to the restriction to ``negative literals'' we do not need the AMO clauses (incorporating them would destroy the structures the translation should preserve). Besides preservation of satisfiability, minimal unsatisfiability and leanness, in Subsection \ref{sec:Preservationmatchingstructure} we show that also a good deal of the matching structure is preserved by the translation (including for example the deficiency). Equipped with these tools, in Theorem \ref{thm:MaximalerDefektFPT} then we obtain FPT for SAT decision in the maximal deficiency.

In \textbf{Section \ref{sec:MUSATneu}} we turn to the study of generalised clause-sets which are minimally unsatisfiable. Considering first the larger class of irredundant generalised clause-sets (no clause is implied by the others) in \textbf{Section \ref{sec:Irredundant}}, we study the question when irredundancy is preserved by applying partial assignments. The class of irredundant clause-sets which stay irredundant for all partial assignments is characterised in Corollary \ref{cor:charakhitirr} as the class of hitting clause-sets, while in Lemma \ref{lem:multihittingmusat} we consider the bigger class of multihitting (generalised) clause-sets and show that they have a unique minimally unsatisfiable core (if they are unsatisfiable). In Subsection \ref{sec:saturated} we then discuss the process of ``saturation'' as introduced in \cite{FlRe94}; for generalised clause-sets we have to face a considerably more complicated situation here than in the boolean case, and thus it seems that for generalised clause-sets saturation does not play the role it does for boolean clause-sets. Without the saturation tool, proving the basic Lemma \ref{lem:GrundtatsacheDefektEins} for the characterisation of $\Musat(1)$ needs a different trick; we use the good properties of the (direct) boolean translation. The main result of Subsection \ref{sec:MUSATdefone} then follows in Theorem \ref{thm:CharakMUSATd1} (the characterisation of minimally unsatisfiable clause-sets of deficiency $1$), and its two corollaries (the characterisation of saturated and marginal minimally unsatisfiable clause-sets of deficiency $1$). A short review on properties related to minimal variable occurrences given in Subsection \ref{sec:stabmmvd} concludes the section.

\textbf{Section \ref{sec:construgcls}} rounds off the picture, by generalising the basic results from linear algebra as applied to boolean clause-sets in \cite{Ku2003e,GalesiKullmann2003bHermitian}, regarding the conflict structure of clauses. A new translation to boolean clause-sets especially suited for these considerations is presented in Subsection \ref{sec:newtrans}, and basic invariance properties are shown. Then the basic results on the hermitian rank and on (boolean) hitting clause-sets are generalised in Subsection \ref{sec:Regularhitting}, providing the basic upper bound on the deficiency by the hermitian defect in Theorem \ref{thm:GPvKlm}, while corollaries apply these results to regular hitting clause-sets.

Finally we present a collection of open problems and conjectures in \textbf{Section \ref{sec:open2}}.

\section{Preliminaries}
\label{sec:Preliminaries}

For the basic notions and notations regarding variables, partial assignments and graphs see Section \ref{sec:prelim} in Part I, while for notions and notations regarding (generalised) (multi-)clause-sets see Section \ref{sec:prelimgeneralcls} in Part I. A quick review of the most basic definitions is as follows.

\emph{Variables} $v$ have finite and non-empty domains $D_v$. \emph{Generalised clause-sets} denote conjunctive normal forms (CNFs). The \emph{literals} of generalised clause-sets are of the form ``$v \not= \ve$'' for values $\ve \in D_v$, which are as mathematical objects defined just as pairs $(v,\ve)$; the disequalities state their meaning, which for DNFs (not considered here) would be the equalities ``$v = \ve$''. \emph{Clauses} are finite and non-tautological sets of literals, that is, they do not contain clashing literals, which are literals $(v,\ve), (v,\ve')$ with identical underlying variables but different values $\ve \not= \ve'$. (Finite) clause-sets are finite sets of clauses (there is no principle problem with infinite clause-sets, however they are not considered here). So boolean variables $v$ have domains $D_v = \set{0,1}$, where ``$v \not= 0$'' is the positive literal, and ``$v \not= 1$'' the negative literal for $v$. Typical examples for the use of (generalised) clause-sets arise from (hyper-)graph colouring problems and from problems from (exact) Ramsey theory; see Subsection \ref{intro:hypergraphcolouring} in Part I. If clauses might occur several times, then we use (finite) multi-clause-sets, which are maps $F: \Cl \ra \NNZ$, where $\Cl$ is the set of all clauses, such that $F(C) > 0$ is the case only for finitely many clauses $C$.

A \emph{partial assignment} is a map $\vp$ whose domain $\var(\vp) := \dom(\vp)$ is a finite set $V$ of variables, such that for all $v \in V$ we have $\vp(v) \in D_v$. A partial assignments $\vp$ satisfies a literal $(v,\ve)$ iff $v \in \var(\vp)$ and $\vp(v) \not= \ve$, while $\vp$ falsifies $(v,\ve)$ if $v \in \var(\vp)$ and $\vp(v) = \ve$. A partial assignment satisfies a clause iff it satisfies at least one literal in it, and it satisfies a clause-sets iff it satisfies all of its clauses. The operation of partial assignments $\vp$ on clause-sets $F$ is denoted by $\vp * F$, which is the clause-set which results from $F$ by removing all satisfied clauses, and removing all falsified literals from the remaining clauses.

The \emph{conflict graph} of a clause-set $F$ has the clauses of $F$ as vertices, while two (different) clauses are adjacent iff they clash in at least one literal-pair. The conflict multigraph has the same vertices as the conflict graph, but now there are as many (parallel) edges between clauses as there are conflicts between them.

Generalising the notion of \emph{hitting clause-sets} (every pair of different clauses clashes in at least one literal-pair; see Subsection \ref{sec:conflictstructure} in Part I), a clause-set $F$ is called \textbf{at most $k$-multihitting} for some $k \in \NNZ$ if the conflict graph of $F$ is complete $k$-partite, while $F$ is called \textbf{multihitting} if it is at most $k$-multihitting for some $k$; let \bmm{\Mclash} denote the set of all multihitting clause-sets. While ``at most $k$-multihitting'' implies that the chromatic number of the conflict graph is at most $k$, if we speak of \textbf{$k$-multihitting} then the chromatic number of the conflict graph must be equal to $k$ (so that $F$ is hitting iff $F$ is $c(F)$-multihitting). For a given multihitting clause-set $F$ there is a unique partition $\FF$ of $F$ (that is, $\FF$ is a set of sub-clause-sets of $F$ which are non-empty and pairwise disjoint, such that their union is $F$), so that for any clauses $C_1, C_2 \in F$ with $C_1 \in G_1$ and $C_2 \in G_2$ for some $G_1, G_2 \in \FF$ the clauses $C_1$ and $C_2$ clash if and only if $G_1 \not= G_2$ (so $F$ is $\abs{\FF}$-multihitting). We call $\FF$ the \textbf{multipartition} of $F$ (if $F$ is bihitting, then $\FF$ is also called the \textbf{bipartition} of $F$).

A basic example class is given as follows. Consider a uniform domain $D = \tb 1k$ and variables $v_1, \dots, v_n$, $n \ge 1$. Let the clause-set $F_{\ve} := \set{\set{(v_i,\ve)} : i \in \tb 1n}$ for $\ve \in \tb 1k$ consist of all the unit-clauses for value $\ve$. Then $F := \bc_{\ve=1}^k F_{\ve}$ is $k$-multihitting with multipartition $\set{F_{\ve} : \ve \in \tb ik}$. We can add to $F$ all full clauses (containing all variables $v_1, \dots, v_n$) containing at least two values, obtaining a $k^n$-multihitting clause-set (with $k^n + (n-1) k$ clauses).

\section{The direct translation}
\label{sec:translating}

In this section we investigate the direct translation of generalised clause-sets into boolean clause-sets. Different from previous research (for an overview see \cite{Pre09HBSAT}), here we are not interested in experimental results (and how good different translations perform in various experiments for different SAT solvers), but we are interest in \emph{structure-preserving} translations. At least regarding our focus on (matching) autarkies and the deficiency, the only reasonable possibility here amongst the known translations seems to be what in \cite{Prestwich2003Encodings} has been coined the ``multivalued encoding'', which is the ``direct translation'', but without AMO (``at most one'') clauses (since these binary clauses would destroy the combinatorial structures we are considering):
\begin{itemize}
\item For every literal $(v,\ve)$ we consider a boolean variable $\transl((v,\ve))$ expressing that $v$ shall not get value $\ve$.
\item Clauses $C$ are translated into (positive) boolean clauses $\tau(C)$ by replacing each literal $x \in C$ with the (positive) boolean literal $\tau(x)$.
\item We add ``ALO clauses'' requiring that each variable gets at least one value (if it gets more than one value, then one of the values can be chosen).
\end{itemize}
In Subsection \ref{sec:newtrans} we will present a new ``nested'' translation with some strong preservation property, but which comes at the price of breaking the symmetry between the possible values of a variable. A very general translation scheme is discussed in Subsection \ref{sec:opentranslations}, covering all known translations.

In \cite{Ku01a} in Subsection 4.5 (``An autarky preserving reduction to boolean clause-sets'') it has already been stated that the direct translation not only preserves satisfiability, unsatisfiability and minimally unsatisfiability, but also leanness. We have to expand these results especially regarding the notions of matching autarkies and deficiency, since in \cite{Ku01a} only a restricted notion of ``matching autarkies'' has been used (recall Subsection \ref{sec:comparisonearlier} from Part I), without an associated notion of deficiency

Another source relevant here is \cite{AnsoteguiManya2004FiniteDomainBooleanDomain}, where ``monosigned CNF formulas'' are translated, a generalisation of ``generalised clause-sets'' allowing also to express that a variable must get a certain value; in other words, where our literals $(v, \ve)$ express ``$v \not= \ve$'', for monosigned formulas also ``positive'' literals ``$v = \ve$'' are allowed. We have discussed this and further generalisations in Subsection \ref{sec:introsignedformulas} in Part I, and we quickly recall the relevant points here. 
\begin{itemize}
\item The generalisation to monosigned CNF formulas can be motivated by the fact that these formulas are exactly those which can be translated by the direct translation; however the price which has to be paid here is that now the AMO clauses are necessary in the direct translation! This adds further to the point we want to make, that generalised clause-sets in our definition (allowing only ``negative literals'') are the appropriate generalisation of boolean conjunctive normal forms, while further generalisations (like ``monosigned formulas'') enter new areas, where the combinatorics of clause-sets no longer can be applied. For a local search algorithm working directly with ``monosigned CNF formulas'' see \cite{FrischPeugniez2001NBLocalSearch} (using the notion of ``nb-formulas'' (for ``non-boolean'')).
\item It is worth to mention here that in \cite{HwangMitchell2005MehrereZweige} it has been shown that resolution which works only with generalised clause-sets, that is, where in the corresponding branching approach for a variable $v$ only a branching of width $\abs{D_v}$ assigning in each branch one of the possible values to $v$ (see \cite{Ku00g}) is considered, can be exponentially worse than resolution on the translation into boolean logic, where now branchings ``$v$ gets value $\ve$'' and ``$v$ does \emph{not} get value $\ve$'' are possible. From this is follows that generalised DPLL-algorithms should not be restricted to branchings where in each branch a variable needs to be fixed to some value; however the focus of this article is not generalisation of SAT solvers, but generalisation of \emph{combinatorial structure}, and thus we do not further pursue these (important) investigations.
\end{itemize}

\subsection{The details of the translation}
\label{sec:detailstranslation}

Formally, the translation proceeds as follows. We consider some bijection $\transl: \Lit \ra \Bva$ from the set of all (generalised) literals to the set of all boolean variables.\footnote{Such a bijection exists due to our assumption on $\Va$, since the set of all literals has the same cardinality as the set of variables, as it is well known from elementary set theory.} The intended meaning of the (positive) boolean literal $\transl((v, \ve))$ for a literal $(v,\ve) \in \Lit$ is the same as the interpretation of the original (generalised) literal, namely ``$v$ shall not get value $\ve$''. We obtain an injection $\transl: \Cl \ra \Cl(\Bva)$ by setting $\transl(C) := \set{\transl(x) : x \in C}$ for $C \in \Cl$. Actually $\transl: \Cl \ra \Cl(\Bva)$ constitutes a bijection from $\Cl$ to the set of all positive boolean clauses. The translation $\transl$ can be further extended to an injection $\transl: \Cls \ra \Bcls$ by $\transl(F) := \set{\transl(C) : C \in F}$ for $F \in \Cls$. Again, $\transl: \Cls \ra \Bcls$ constitutes a bijection from the set of (generalised) clause-sets to the set of boolean clause-sets containing only positive clauses. Finally, for $v \in \Va$ let 
\begin{displaymath}
  \bmm{\alo_v} := \set{\ol{\transl((v, \ve))} : \ve \in D_v} \in \Cl(\Bva)
\end{displaymath}
be the (negative, non-empty boolean) clause expressing that $v$ gets assigned at least one of the values $\ve \in D_v$ (that is, not all (positive) literals $\tau((v,\ve))$ for $\ve \in D_v$ can be true), and let the full translation $\ftrans: \Cls \ra \Bcls$ (which again is an injection) by given as 
\begin{displaymath}
  \bmm{\ftrans(F)} := \transl(F) \cup \set{\alo_v : v \in \var(F)}.
\end{displaymath}
Note that the union in the definition of $\ftrans(F)$ is disjoint, since $\transl(F)$ consists only of positive clauses, while $\set{\alo_v : v \in \var(F)}$ consists only of non-empty negative clauses (and thus $\ftrans(F)$ is a ``PN-clause-set'' as defined in \cite{GalesiKullmann2003bHermitian}). As an example, consider $F = \setb{\set{v \not= 0, w \not= 1}, \set{v \not= 1, w \not= 0}, \set{v \not= 2, w \not= 2}}$ for variables $v, w$ with $D_v = D_w = \set{0,1,2}$. Now, using $a_{\ve} := \tau((v, \ve))$ and $b_{\ve} := \tau((w, \ve))$ for $\ve \in  \set{0,1,2}$ (so altogether we get six boolean variables here), we have
\begin{displaymath}
  \ftrans(F) = \setb{ \set{a_0, b_1}, \set{a_1, b_0}, \set{a_2, b_2}, \set{\ol{a_0}, \ol{a_1}, \ol{a_2}}, \set{\ol{b_0}, \ol{b_1}, \ol{b_2}}}.
\end{displaymath}
In general the sub-clause-sets of $\ftrans(F)$ not containing pure variables (recall Subsection \ref{sec:Reduction} from Part I) are exactly the $\ftrans(F')$ for $F' \sse F$ not containing pure variables.

\subsection{Preservation of general structure}
\label{sec:Preservationgeneralstructure}

Regarding set-theoretical operations we have that $\ftrans$ is an embedding of the semilattice $(\Cls, \cup)$ into $(\Bcls, \cup)$, that is, for $F_1, F_2 \in \Cls$ we have 
\begin{displaymath}
  \ftrans(F_1 \cup F_2) = \ftrans(F_1) \cup \ftrans(F_2).
\end{displaymath}
Thus $\ftrans$ is also an order embedding, i.e., $F_1 \sse F_2 \Lra \ftrans(F_1) \sse \ftrans(F_2)$. By definition we have for $F \in \Cls$ the equalities
\begin{eqnarray*}
  c(\ftrans(F)) & = & c(F) + n(F)\\
  n(\ftrans(F)) & = & \sum_{v \in \var(F)} \abs{D_v}\\
  \delta(\ftrans(F)) & = & c(\ftrans(F)) - n(\ftrans(F)) = c(F) - \rd(F) = \delta(F),
\end{eqnarray*}
and thus the translation $\ftrans$ preserves the deficiency of clause--sets as defined in Subsection \ref{sec:matchsat} (Part I). It follows immediately that $\delta^*(\ftrans(F)) \ge \delta^*(F)$ holds for all $F \in \Cls$, but inequality can occur here (see Subsection \ref{sec:Preservationmatchingstructure}).

We consider now the relations between partial assignments $\vp \in \Pass$ for $F \in \Cls$ and partial assignments $\psi \in \Pass(\Bva)$ for $\ftrans(F) \in \Bcls$. For $\vp \in \Pass$ we define the partial assignment $\bmm{\transl(\vp)} \in \Pass(\Bva)$ by letting $\var(\transl(\vp)) := \set{\transl((v,\ve)) : v \in \var(\vp), \ve \in D_v}$ be the set of all boolean variables corresponding via the translation to literals over the variables in $\var(\vp)$, while $\transl(\vp)((v,\ve)) = 0$ iff $\vp(v) = \ve$. If we consider for example the partial assignment $\pab{v \ra 1, w \ra 2}$ for variables $v,w$ with $D_v = D_w = \set{0,1,2}$, then, using as above $a_{\ve} := \tau((v, \ve))$ and $b_{\ve} := \tau((w, \ve))$ for $\ve \in  \set{0,1,2}$, we get
\begin{displaymath}
  \tau(\pab{v \ra 1, w \ra 2}) = \pab{a_0 \ra 1, a_1 \ra 0, a_2 \ra 1, b_0 \ra 1, b_2 \ra 1, b_2 \ra 0}.
\end{displaymath}
 The partial assignments in $\Pass(\Bva)$ of the form $\transl(\vp)$ for some $\vp \in \Pass$ are called \textbf{standard partial assignments (w.r.t.\ $\transl$)}. So $\transl$ constitutes a bijection between $\Pass$ and the standard partial assignments (which are always boolean), and standard partial assignments $\vp \in \Pass(\Bva)$ are characterised by the condition, that whenever some $\transl((v,\ve)) \in \var(\vp)$, then for all $\ve' \in D_v$ we have $\transl((v,\ve')) \in \var(\vp)$, and there is exactly one $\ve_0 \in D_v$ with $\vp(\transl((v,\ve_0))) = 0$; for the corresponding partial assignment $\transl^{-1}(\vp) \in \Pass$ we then have $\transl^{-1}(\vp)(v) = \ve_0$.

In the following lemma we see that the properties of $\vp$ regarding touching or satisfying clauses are well reflected by $\transl(\vp)$, and hence the translation is invariant regarding the autarky property and the property of satisfying a clause-set.
\begin{lem}\label{lem:Eigtrans1}
  For $\vp \in \Pass$, $C \in \Cl$ and $F \in \Cls$ we have
  \begin{enumerate}
  \item\label{item:Eigtrans11} $\vp$ touches resp.\ satisfies $C$ if and only if $\transl(\vp)$ touches resp.\ satisfies $\transl(C)$. Thus
    \begin{eqnarray*}
      \transl(F_{\var(\vp)}) & = & \transl(F)_{\var(\transl(\vp))}\\
      \ftrans(F[\var(\vp)]) & = & \ftrans(F)[\var(\transl(\vp))].
    \end{eqnarray*}
  \item\label{item:Eigtrans12} $\transl(\vp)$ is an autarky for the set of clauses $\set{\alo_v : v \in \Va}$.
 \item\label{item:Eigtrans13} $\vp$ is an autarky for $F$ if and only if $\transl(\vp)$ is an autarky for $\ftrans(F)$.
  \item\label{item:Eigtrans14} If $\transl(\vp)$ satisfies $\ftrans(F)$, then $\vp$ satisfies $F$. If on the other hand $\vp$ satisfies $F$ and $\var(\vp) \supseteq \var(F)$ holds, then $\transl(\vp)$ satisfies $\ftrans(F)$.
  \end{enumerate}
\end{lem}
\begin{proof} Parts \ref{item:Eigtrans11}, \ref{item:Eigtrans12} follow directly from the definitions, while Part \ref{item:Eigtrans13} follows from Parts \ref{item:Eigtrans11}, \ref{item:Eigtrans12}, and Part \ref{item:Eigtrans14} follows from Parts \ref{item:Eigtrans11}, \ref{item:Eigtrans13}. \end{proof}

For the reverse direction, translating partial assignments in $\Pass(\Bva)$ to partial assignments in $\Pass$, call $\vp \in \Pass(\Bva)$ \textbf{admissible} if $\vp$ is an autarky for the set of clauses $\set{\alo_v : v \in \Va}$, that is, if $\transl((v, \ve)) \in \var(\vp)$, then there is $\ve_0 \in D_v$ with $\vp(\transl((v,\ve_0))) = 0$. In words: a partial assignment $\vp$ for the boolean variables is admissible iff for every variable $\transl((v,\ve))$ in its domain there exists a value $\ve_0 \in D_v$ such that $\transl((v,\ve_0))$ is in the domain of $\vp$ as well with $\vp(\transl((v,\ve_0))) = 0$. Note that an autarky $\vp \in \Pass(\Bva)$ for $\ftrans(F)$ (this includes satisfying assignments) is admissible in case of $\var(\vp) \sse \var(\ftrans(F))$.

Call a standard partial assignment $\psi \in \Pass(\Bva)$ a \textbf{standard completion} of an admissible $\vp \in \Pass(\Bva)$ if $\psi$ touches (satisfies) exactly the same ALO-clauses as $\vp$, and if from $\psi(\transl((v,\ve))) = 0$ always follows $\vp(\transl((v,\ve))) = 0$. In other words, a standard completion $\psi$ of an admissible $\vp$ is obtained from $\vp$ by considering all variables $v$ such that $\ve \in D_v$ with $\transl((v,\ve)) \in \var(\vp)$ exists, choosing $\ve_0(v) \in D_v$ with $\vp(\transl((v,\ve_0(v)))) = 0$, and setting $\psi(\transl((v,\ve'))) := 1$ for $\ve' \in D_v \sm \set{\ve_0(v)}$, while $\psi(\transl((v,\ve_0(v)))) := 0$.

The standard completions of admissible partial assignments in $\Pass(\Bva)$ are exactly the $\transl(\theta)$ for $\theta \in \Pass$. The purpose of standard completions $\psi$ of admissible partial assignments $\vp$ is to transform as follows an (arbitrary) autarky $\vp$ for $\ftrans(F)$ into an autarky $\psi$ for $\ftrans(F)$ which is also a standard partial assignment:
\begin{enumerate}\NAA
\item First we restrict $\vp$ to $\var(\ftrans(F))$ to be sure it is admissible.
\item Then we choose a standard completion $\psi$.
\end{enumerate}
Now by Lemma \ref{lem:Eigtrans1} (Part \ref{item:Eigtrans13}) we obtain from $\psi$ an autarky for $F$. The following lemma (with obvious proofs) states the basic properties of standard completions.

\begin{lem}\label{lem:Eigtrans2}
  For $C \in \Cl$ and $F \in \Cls$, an admissible $\vp \in \Pass(\Bva)$ and a standard completion $\psi \in \Pass(\Bva)$ of $\vp$ we have
  \begin{enumerate}
  \item\label{item:Eigtrans21} If $\vp$ touches resp.\ satisfies $\transl(C)$ then $\psi$ touches resp.\ satisfies $\transl(C)$.
 \item\label{item:Eigtrans22} If $\vp$ is an autarky for $\ftrans(F)$ then $\psi$ is an autarky for $\ftrans(F)$.
  \end{enumerate}
\end{lem}

\begin{lem}\label{lem:Eigtrans3}
  For a (generalised) clause-set $F \in \Cls$ we have:
  \begin{enumerate}
  \item\label{item:Eigtrans31} $F \in \Sat \Lra \ftrans(F) \in \Sat$.
  \item\label{item:Eigtrans32} $F \in \Musat \Lra \ftrans(F) \in \Musat$.
  \item\label{item:Eigtrans33} $F \in \Lean \Lra \ftrans(F) \in \Lean$.
  \end{enumerate}
\end{lem}
\begin{proof} If $F \in \Sat$ then $\ftrans(F) \in \Sat$ with Lemma \ref{lem:Eigtrans1}, Part \ref{item:Eigtrans14}, and if $\ftrans(F) \in \Sat$, then $F \in \Sat$ with Lemma \ref{lem:Eigtrans2}, Part \ref{item:Eigtrans21} and Lemma \ref{lem:Eigtrans1}, Part \ref{item:Eigtrans14}.

If $F \in \Musat$, but there were some minimally unsatisfiable $F^* \subset \ftrans(F)$, then there would be $F' \subset F$ with $\ftrans(F') = F^*$ (since $F^*$ does not contain pure variables), and thus $F'$ would be unsatisfiable by Part \ref{item:Eigtrans31}. If on the other hand $\ftrans(F) \in \Musat$, but there were some unsatisfiable $F' \subset F$, then $\ftrans(F')$ would be unsatisfiable as well by Part \ref{item:Eigtrans31}.

Finally, if $F \in \Lean$ then $\ftrans(F) \in \Lean$ by  Lemma \ref{lem:Eigtrans2}, Part \ref{item:Eigtrans22} and Lemma \ref{lem:Eigtrans1}, Part \ref{item:Eigtrans13}, and if $\ftrans(F) \in \Lean$ then $F \in \Lean$ by  Lemma \ref{lem:Eigtrans1}, Part \ref{item:Eigtrans13} (the other direction). \end{proof}

Parts \ref{item:Eigtrans31} and \ref{item:Eigtrans33} have been concluded in Corollary 20 in \cite{Ku01a} from the stronger property $\na(\ftrans(F)) = \ftrans(\na(F))$ (recall that $\na$ is the lean kernel operator); in this article we do not go further with the study of the translation $\ftrans$, but we restrict ourselves to the minimum required to understand our applications.

\subsection{Preservation of matching structure}
\label{sec:Preservationmatchingstructure}

\begin{lem}\label{lem:Eigtransmatch1}
  For $\vp \in \Pass$, $C \in \Cl$ and $F \in \Cls$ we have
  \begin{enumerate}
  \item\label{lem:Eigtransmatch11} If $\transl(\vp)$ matching satisfies $\ftrans(F)$, then $\vp$ matching satisfies $F$.
   \item\label{lem:Eigtransmatch12} If $\transl(\vp)$ is a matching autarky for $\ftrans(F)$, then $\vp$ is a matching autarky for $F$.
  \end{enumerate}
\end{lem}
\begin{proof} If the partial assignment $\transl(\vp)$ matching satisfies $\ftrans(F)$, then (by definition) for each clause $D \in \ftrans(F)$ one can choose a literal $x_D \in D$ with $\vp(x_D) = 1$, such that for the variables $\var(x_D) = \transl((v_D,\ve_D))$ the map $D \in \ftrans(F) \mapsto \transl((v_D, \ve_D))$ is injective (whence $D \in \ftrans(F) \mapsto (v_D, \ve_D)$ is injective). Now the map $C \in F \mapsto v_{\transl(C)}$ has for each image $v_{\transl(C)}$ at most $(\abs{D_v} - 1)$-many inverse images, since for each $\ve \in D_v$ there is at most one $D \in \ftrans(F)$ with $v_D = v_{\transl(C)}$ and $\ve_D = \ve$, and exactly one of these $D$ is the clause $\alo_{v_D}$.

For Part \ref{lem:Eigtransmatch12} recall that $\vp$ is a matching autarky for $F$ iff $\vp$ matching satisfies $F[\var(\vp)]$, which by Part \ref{lem:Eigtransmatch11} follows from $\transl(\vp)$ matching satisfying $\ftrans(F[\var(\vp)])$, where by Lemma \ref{lem:Eigtrans1}, Part \ref{item:Eigtrans11} we have $\ftrans(F[\var(\vp)]) = \ftrans(F)[\var(\transl(\vp))]$, and thus the latter assertion is equivalent to $\transl(\vp)$ being a matching autarky for $\ftrans(F)$. \end{proof}

Lemma 19, Part (1)(d) of \cite{Ku01a} rephrased in the terminology of Subsection \ref{sec:comparisonearlier} (Part I) says that if $\vp$ is a \emph{non-repetitive} matching autarky for $F$, then $\transl(\vp)$ is a matching autarky for $\ftrans(F)$; in follows then in Corollary 20 of \cite{Ku01a} that if $\ftrans(F)$ is matching lean, then $F$ is lean w.r.t.\ non-repetitive matching autarkies. These properties do not hold for matching autarkies in general (in the presence of non-boolean variables), as the following examples show.

An example, where a matching autarky $\vp$ for a (generalised) clause-set $F \in \Cls$ does not yield a matching autarky $\transl(\vp)$ for $\ftrans(F)$, is given by \emph{multi}-clause-sets as follows:
\begin{enumerate}
\item Consider the multi-clause-set $F_1 := 2 \cdot \set{(v,0)}$ for a variable $v$ with $D_v = \set{0,1,2}$.
\item $F_1$ is matching satisfiable (via $\pao v1$ or $\pao v2$).
\item $F_1$ is lean w.r.t.\ non-repetitive matching autarkies.
\item The translation is $\ftrans(F_1) = 2 \cdot \set{\transl((v,0))} + \set{\ol{\transl((v,0))}, \ol{\transl((v,1))},\ol{\transl((v,2))}}$, where $\var(\ftrans(F_1)) = \set{\transl((v,0)),\transl((v,1)),\transl((v,2))}$.
\item $\nma(\ftrans(F_1)) = 2 \cdot \set{\transl((v,0))} = \transl(F_1)$ (via matching autarkies we can only eliminate the ALO-clause).
\item Thus $\ftrans(F_1)$ is not matching satisfiable.
\end{enumerate}
One sees that the problem with transferring matching autarkies from generalised (multi-)clause-sets to their boolean translation lies in the possibility that a matching in the clause-variable graph $B(F)$ might use the same literal several times, which is not possible for the translated literals. To obtain an example using clause-\emph{sets},  consider additionally two boolean variables $w, w'$ and let
\begin{displaymath}
  F_2 = \setb{ \set{v \not= 0, w \not= 0}, \, \set{v \not= 0, w' \not= 0}, \, \set{w \not= 1}, \, \set{w' \not= 1} }.
\end{displaymath}
The partial assignment $\vp := \pab{v \ra 1, w \ra 0, w' \ra 0}$ is matching satisfying for $F_2$ (note that again $F_2$ is lean w.r.t.\ non-repetitive matching autarkies), but $\transl(\vp)$ is not a matching autarky for $\ftrans(F_2)$, and moreover the matching lean kernel of $\ftrans(F_2)$ is $\ftrans(F_2) \sm \set{\alo_v}$ (again only the ALO-clause for $v$ can be eliminated via matching autarkies), and thus $\ftrans(F_2)$ is not matching satisfiable. Furthermore we have in this example $\delta^*(F_2) = 0$ and $\delta^*(\ftrans(F_2)) = \delta(\ftrans(F_2) \sm \set{\alo_v}) = \delta(\ftrans(F_2)) - (1 - 2) = \delta(F_2) + 1 = 1$.

Now consider the transfer of matching autarkies in the other direction, that is, we have given a matching autarky $\vp \in \Pass(\Bva)$ for $\ftrans(F)$, and we want to obtain some associated matching autarky for $F$. The problem here is that $\vp$ might use some variable $\transl((v,\ve))$, but not a variable $\transl((v,\ve'))$ for some $\ve' \in D_v \sm \set{\ve}$, and such situations can not be translated back to $F$. The simplest example for this phenomenon is (again) given by a \emph{multi}-clause-set as follows:
\begin{enumerate}
\item Let $F_3 := \set{(v,1)} + 2 \cdot \set{(v,2)}$ for a variable $v$ with $D_v = \set{0,1,2}$.
\item Obviously $F_3$ is matching lean.
\item The translation is
  \begin{displaymath}
    \ftrans(F_3) = \set{\transl((v,1))} + 2 \cdot \set{\transl((v,2))} + \set{\ol{\transl((v,0))},\ol{\transl((v,1))},\ol{\transl((v,2))}},
  \end{displaymath}
  where $\var(\ftrans(F_3)) = \set{\transl((v,0)),\transl((v,1)),\transl((v,2))}$.
\item Thus $\nma(\ftrans(F_3)) = \transl(2 \cdot \set{(v,2)})$ via the matching autarky $\pab{\transl((v,0)) \ra 0, \transl((v,1)) \ra 1}$.
\end{enumerate}
A clause-\emph{set} $F_4$, where $F_4$ is matching lean but $\ftrans(F_4)$ is not, is given by 
\begin{displaymath}
  F_4 := \setb{ \set{v \not= 1}, \set{v \not= 2}, \set{v \not= 2, w \not= 0}, \set{w \not= 1} }
\end{displaymath}
for an additional boolean variable $w$, since here 
\begin{displaymath}
  \nma(\ftrans(F_4)) = \transl(\set{\set{v \not= 2}, \set{v \not= 2, w \not= 0}, \set{w \not= 1} }) \cup \set{\alo_w}
\end{displaymath}
via the matching autarky $\pab{\transl((v,0)) \ra 0, \transl((v,1)) \ra 1}$ for $\ftrans(F_4)$.

\subsection{Polynomial time SAT decision in the maximal deficiency}
\label{sec:PTimeSATmd}

As we have seen now, matching autarkies for (generalised) clause-sets $F \in \Cls$ and matching autarkies for $\ftrans(F) \in \Cls(\Bva)$ in general are incomparable. Nevertheless we can use them to show fixed-parameter tractability for generalised clause-sets w.r.t.\ the parameter $\delta^*(F)$ as follows.

\begin{thm}\label{thm:MaximalerDefektFPT}
  SAT decision for (generalised) clause-sets $F \in \Cls$ can be done in time $O \big (2^{\delta^*(F)} \cdot (\sum_{v \in \var(F)} \abs{D_v})^3 \big )$
\end{thm}
\begin{proof} Consider $F \in \Cls$ and let $F^*$ be the result of reducing $\ftrans(F)$ w.r.t.\ matching autarkies and pure autarkies (thus $F^*$ is the unique maximal sub-clause-set of $F$ which is matching lean and does not contain pure variables). We can compute $F^*$ in polynomial time, and $F^*$ is satisfiability equivalent to $F$. Since $F^*$ contains no pure literals, it corresponds to a sub-clause-set of $F$, and thus we have $\delta(F^*) \le \delta^*(F)$, and since $F^*$ is matching lean we have $\delta^*(F^*) = \delta(F^*)$. Theorem 4 in \cite{Szei2002FixedParam} says that satisfiability of $F^*$ can be tested in time $O(2^{\delta^*(F^*)} \cdot n(F^*)^3)$, where in this procedure actually already the cost of reducing $\ftrans(F)$ to $F^*$ is included if we use $n(\ftrans(F))$ instead of $n(F^*)$ in the big-Oh expression (see Section 5 in  \cite{Szei2002FixedParam}, or use the argumentation of Subsection \ref{sec:expansion} of Part I), and the theorem follows. \end{proof}

\section{Irredundant generalised clause-sets}
\label{sec:Irredundant}

One of the motivations behind the notion of lean clause-sets can be seen in ``approximating'' the fundamental notion of minimally unsatisfiable clause-sets. In this section we will now consider some of the basic facts on the more general notion of ``irredundant clause-sets'' (that is, every clause contributes something unique, allowing also to consider satisfiable clause-sets) in our generalised setting. Detailed studies of irredundant clause-sets in the boolean case can be found in the following references:
\begin{enumerate}
\item \cite{BueningZhao2005Clauseminimal} (speaking of ``clause minimal formula'') focuses on questions related to the problem (from a complexity theoretical perspective) when for given clause-sets $F, H$ there exists a clause-set $G$ such that $F \cup G$ is equivalent to $H$.
\item \cite{Liberatore2005Redundanz} considers in various forms (also mostly from a complexity-theoretical perspective) the problem of finding an irredundant core in a given clause-set.
\end{enumerate}

 We start in Subsection \ref{sec:irredundant} with a discussion of the notion of ``irredundant clause-sets'', concentrating on the basic question of preservation of irredundancy under application of partial assignments. In Subsection \ref{sec:hittinggen} we consider the in some sense most extreme case of irredundant clause-sets, namely ``hitting clause-sets'': every two different clauses clash, that is, have no common falsifying assignment, or, in other words, the conflict graph is complete. Furthermore we consider the natural generalisation to ``multihitting clause-sets'' (the conflict graph is multipartite; see Section \ref{sec:Preliminaries}). In Corollary \ref{cor:charakhitirr} we show that hitting clause-sets are exactly those clause-sets which are irredundant after application of every partial assignment, and thus unsatisfiable hitting clause-sets are exactly those clause-sets which are minimally unsatisfiable after application of every partial assignment (Corollary \ref{cor:charakhitting}). For unsatisfiable multihitting clause-sets we show in Lemma \ref{lem:multihittingmusat} that they have exactly one minimally unsatisfiable sub-clause-set (which can be computed efficiently by subsumption-elimination), and in Lemma \ref{lem:polyzeitbitreff} we show that the satisfiability problem for bihitting clause-sets (where the conflict graph is bipartite) is solvable in quasi-polynomial time (this problem is essentially the same problem as the hypergraph transversal problem). We conclude in Subsection \ref{sec:irrcores} by considering ``irredundant cores'', that is, minimal equivalent sub-clause-sets, in general. The main result here is Lemma \ref{lem:meqhypgr}, expressing the duality between irredundant cores and maximal non-equivalent sub-clause-sets via the formation of hypergraph transversals.

The special case of unsatisfiable irredundant clause-sets (i.e., minimally unsatisfiable clause-sets) is considered in Chapter \ref{sec:MUSATneu}.

\subsection{Irredundant clause-sets}
\label{sec:irredundant}

A clause $C \in F$ is called \textbf{redundant} (or \textbf{unnecessary}) for clause-set $F \in \Cls$ if $F \sm \set{C} \models C$ holds, while otherwise $C$ is called \textbf{irredundant} (or \textbf{necessary}) for $F$. The following conditions are equivalent for a clause $C \in F$:
\begin{itemize}
\item $C$ is redundant for $F$.
\item $F \sm \set{C}$ is equivalent to $F$.
\item The set $\modf_{\var(F)}(C)$ of falsifying assignments for $C$ (w.r.t.\ the variables of $F$) is covered by the family $(\modf_{\var(F)}(C'))_{C' \in F \sm \set{C}}$ of sets of falsifying assignments for the remaining clauses.
\end{itemize}
A (generalised) clause-set $F \in \Cls$ is called \textbf{irredundant} if all $C \in F$ are irredundant for $F$, otherwise $F$ is called \textbf{redundant}.  A clause-set $F$ is minimally unsatisfiable if and only if $F$ is unsatisfiable and irredundant. Regarding complexity classifications of decision problems related to (ir)redundancy we have the following:
\begin{enumerate}
\item In \cite{PW88} it is shown that the decision problem whether a (boolean) clause-set is irredundant is NP-complete, while the decision problem whether a (boolean) clause-set is minimally unsatisfiable is $D^P$-complete. Trivially these results also hold for generalised clause-sets.
\item As we have seen in Theorem \ref{thm:MaximalerDefektFPT}, SAT decision for (generalised) clause-sets is fixed-parameter tractable in the maximal deficiency, and thus also irredundancy decision is fixed-parameter tractable in the maximal deficiency.  Since for minimally unsatisfiable (generalised) clause-sets maximal deficiency and deficiency coincide (Corollary \ref{cor:musatdelta} in Part I), minimally unsatisfiability decision is also fixed-parameter tractable in the deficiency; however, as shown in Proposition 1 in \cite{BueningZhao2005Clauseminimal}, the decision whether a (boolean) clause-set is irredundant with deficiency $k$ is NP-complete for every fixed $k \in \NN$ (different from minimally unsatisfiable clause-sets, irredundant clause-sets of deficiency $k$ can contain sub-clause-sets of arbitrary deficiency). Obviously the same holds for generalised clause-sets.
\end{enumerate}

We are interested here in the question, given a partial assignment $\vp$ and a clause $C \in F$ with $\vp * \set{C} \not= \top$ (i.e., $C$ is not satisfied by $\vp$), under what circumstances is the clause $\vp * C = C \sm C_{\vp}$ redundant for $\vp * F$ ? We will see that this question is closely related to the question, how ``much irredundant'' $C$ is for $F$, that is, how much of $\modf_{\var(F)}(C)$ is covered by $(\modf_{\var(F)}(C'))_{C' \in F \sm \set{C}}$, which can be recast as the question, whether for some $C' \supseteq C$ we have $F \sm \set{C} \models C'$.

Assume that $\vp * C$ is redundant for $\vp * F$, that is, $(\vp * F) \sm (\vp * \set{C}) \models \vp * C$ holds. Due to $(\vp * F) \sm (\vp * \set{C}) \sse \vp * (F \sm \set{C})$ it follows $\vp * (F \sm \set{C}) \models \vp * C$, which is equivalent to $F \sm \set{C} \models C \cup C_{\vp}$. Let us call $C$ \textbf{\bm$\vp$-redundant} for $F$ if $F \sm \set{C} \models C \cup C_{\vp}$ holds, and otherwise \textbf{\bm$\vp$-irredundant}. In other words, $C$ is $\vp$-redundant for $F$ iff the part of $\modf_{\var(F)}(C)$ which consists of assignments compatible with $\vp$ is covered by $(\modf_{\var(F)}(C'))_{C' \in F \sm \set{C}}$. Obviously, $C$ is redundant for $F$ iff $C$ is $\es$-redundant for $F$, and if $C$ is $\vp$-redundant for $F$, then $C$ is also $\vp'$-redundant for $F$ for every partial assignment $\vp'$ with $\vp \sse \vp'$. So $\vp$-redundancy generalises (ordinary) redundancy by weakening it, while $\vp$-irredundancy strengthens (ordinary) redundancy. For an example consider boolean variables $a,b$ and the irredundant clause-set $F = \set{\set{a}, \set{b}}$: $\set{b}$ is $\pao a0$-redundant for $F$.

If $C$ is $\vp$-irredundant for $F$, then $\vp * C$ is irredundant for $\vp * F$, but the reverse direction is not true in general due to the fact that there might be other clauses $C' \in F$ with $\vp * C' = \vp * C$. To repair this, let us call clause $C$ \textbf{contraction-\bm$\vp$-redundant} for $F$ if
\begin{displaymath}
  F \sm \set{ C' \in F : \vp * \set{C'} = \vp * \set{C} } \models C \cup C_{\vp},
\end{displaymath}
while otherwise we call $C$ \textbf{contraction-\bm$\vp$-irredundant} for $F$. We summarise (and extend) the foregoing discussion in Lemma \ref{lem:TatsachenRedundanz}, whose proof should be obvious by now.

\begin{lem}\label{lem:TatsachenRedundanz}
  Consider a generalised clause-set $F \in \Cls$, a clause $C \in F$ and a partial assignment $\vp \in \Pass$ such that $\vp * \set{C} \not= \top$.
  \begin{enumerate}
  \item $\vp * C$ is (ir)redundant for $\vp * F$ if and only if $C$ is contraction-$\vp$-(ir)redundant for $F$.
  \item\label{item:TatsachenRedundanz2} 
    \begin{enumerate}
    \item\label{item:TatsachenRedundanz2a} If $C$ is $\vp$-irredundant for $F$, then $C$ is contraction-$\vp$-irredundant for $F$.
    \item If there is no clause $C' \in F \sm \set{C}$ with $\vp * \set{C'} = \vp * \set{C}$ (that is, $C$ is ``contraction-free'' in $F$ w.r.t.\ $\vp$), then also the reverse direction holds, that is, if $C$ is contraction-$\vp$-irredundant for $F$ then $C$ is $\vp$-irredundant for $F$. Clause $C$ is contraction-free in $F$ w.r.t.\ $\vp$ in the following cases:
      \begin{enumerate}[(i)]
      \item $n(\vp) = 0$ (i.e., $\vp$ is the empty partial assignment);
      \item $n(\vp) = 1$ and $F$ is subsumption-free;
      \item $C$ clashes with every $C' \in F \sm \set{C}$.
      \end{enumerate}
    \end{enumerate}
  \end{enumerate}
\end{lem}

\begin{corol}\label{cor:TatsachenRedundanz}
  Consider a generalised clause-set $F \in \Cls$ which is subsumption-free, a clause $C \in F$ and a variable $v \in \Va$ together with a value $\ve \in D_v$ such that for all $\ve' \in D_v \sm \set{\ve}$ we have $(v, \ve') \notin C$. Then $\pao{v}{\ve} * C = C \sm \set{(v,\ve)}$ is irredundant for $\pao{v}{\ve} * F$ if and only if $C$ is $\pao{v}{\ve}$-irredundant for $F$, that is, iff $F \sm \set{C} \not\models C \cup \set{(v,\ve)}$. In other words, a non-satisfied clause is irredundant after application of an elementary partial assignment (for the new clause-set) iff the original clause is irredundant (for the original clause-set) even after addition of the falsified literal.
\end{corol}

Obviously irredundant clause-sets are subsumption-free, and from Corollary \ref{cor:TatsachenRedundanz} we get immediately:

\begin{corol}\label{cor:TatsachenRedundanz2}
  Consider an irredundant generalised clause-set $F \in \Cls$, a clause $C \in F$ and a variable $v \in \Va$ together with $\ve \in D_v$.
  \begin{enumerate}
  \item If there exists $\ve' \in D_v \sm \set{\ve}$ with $(v, \ve') \in C$, then clause $C$ vanishes when applying $\pao{v}{\ve}$ to $F$ (and in that sense it becomes redundant in $\pao{v}{\ve}$). So assume $\val_v(\set{C}) \sse \set{\ve}$ in the sequel (that is, $C$ is not satisfied by $\pao{v}{\ve}$).
  \item If $(v, \ve) \in C$, then $\pao{v}{\ve} * C = C \sm \set{(v,\ve)}$ is irredundant for $\pao{v}{\ve} * F$.
  \item If $(v, \ve) \notin C$, then $C$ is irredundant for $\pao{v}{\ve} * F$ if and only if $C$ is $\pao{v}{\ve}$-irredundant for $F$, i.e., iff $F \sm \set{C} \not\models C \cup \set{(v, \ve)}$.
  \end{enumerate}
  In other words, considering the application of an elementary partial assignment to an irredundant clause-set, new clauses (produced by the application) are definitely irredundant, while untouched clauses stay irredundant iff they stay irredundant in the original clause-set after addition of the falsified literal (so that they fall then actually under the first case).
\end{corol}

Considering a clause $C \in F$, we called $C$ redundant for $F$ iff $F \sm \set{C} \models C$; now for arbitrary clauses $C$ we can call $C$ ``dependent'' on $F$ if $F \models C$ holds (that is, if the set of falsifying assignments of $F$ covers the set of falsifying assignments of $C$), and otherwise ``independent''. If $C \in F$, then $C$ is dependent on $F$, while $C$ is redundant for $F$ iff $C$ is dependent on $F \sm \set{C}$. The relation of $C$ depending on $F$ allows two dimensions for minimisation: Considering a minimal clause $C$ which is dependent on $F$ we arrive at the notion of a \emph{prime implicate} of $F$ (an implied clause, which ceases to be implied after removal of any literal), while considering a minimal clause-set $F$ such that $C$ depends on $F$ we arrive at a \emph{minimal premise set} for $C$. The following lemma states the relation between minimal premise sets and minimally unsatisfiable clause-sets.

\begin{lem}\label{lem:minpremisses}
  Consider a generalised clause-set $F \in \Cls$ and a clause $C \in \Cl$. Then the following assertions are equivalent:
  \begin{enumerate}
  \item $F$ is a minimal premise set for $C$.
  \item $\vp_C * F$ is minimally unsatisfiable, no clause of $F$ is satisfied by $\vp_C$, and $F$ is $\vp_C$-contraction free, that is, there are no clauses $D, D' \in F$, $D \not= D'$, with $\vp_C * \set{D} = \vp_C * \set{D'}$.
  \end{enumerate}
\end{lem}

\begin{corol}\label{cor:EindPrim}
  If clauses $C, D$ are prime implicates of the generalised clause-set $F$, and if $F$ is a minimal premise set for $C$ as well as for $D$, then $C = D$ holds.
\end{corol}
\begin{proof}
  Assume w.l.o.g.\ that there is a literal $x \in C \sm D$. Since $\vp_C$ does not satisfy any clause of $F$, while $C$ is minimal, the literal $x$ occurs in $F$, and all occurrences of $\var(x)$ in $F$ are equal to $x$ (while $\abs{D_{\var(X)}} \ge 2$). Furthermore $\var(x) \notin \var(D)$ (since $\vp_D$ does not satisfy any clause of $F$), and then $\vp_D * F$ contains the pure variable $\var(x)$ and thus can not be unsatisfiable.
\end{proof}

Even an irredundant clause-set $F$ may not be a minimal premise set of any clause $C$ (consider for example the boolean clause-set $\set{\set{a},\set{b}}$), and for a prime implicate $C$ of $F$ there might be several minimal premise sub-sets of $F$ (consider for example $C = \set{a,b}$ and the boolean clause-set $\set{\set{a,x},\set{\ol{x},b},\set{a,y},\set{\ol{y},b}}$). However by Corollary \ref{cor:EindPrim} for a clause-set $F$ there can be at most one prime implicate $C$ such that $F$ is a minimal premise set of $C$, and thus we get (by considering all non-empty sub-clause-sets):

\begin{corol}\label{cor:MaxPrim}
  A generalised clause-set $F$ can have at most $2^{c(F) - 1}$ many prime implicates.
\end{corol}
See the end of Subsection \ref{sec:Regularhitting} for some comments on the sharpness of this bound.

\subsection{Hitting and multihitting clause-sets}
\label{sec:hittinggen}

The next lemma answers the question which clauses $C$ remain irredundant for a clause-set $F$ under \emph{all} applications of partial assignments; this strongest form of irredundancy of $C$ for $F$ turns out to be equivalent to the condition that the set of falsifying assignments for $C$ is not covered at all by $(\modf_{\var(F)}(C'))_{C' \in F \sm \set{C}}$. A simple but important observation here is that for two clauses $C, C'$ and $\var(C) \cup \var(C') \sse V$ we have $\modf_V(C) \cap \modf_V(C') = \es$ iff $C$ and $C'$ clash.
\begin{lem}\label{lem:charakallpart}
  Consider a generalised clause-set $F \in \Cls$ and a clause $C \in \Cl$. Then the following assertions are equivalent:
  \begin{enumerate}[(i)]
  \item\label{item:charakallpart1} $C$ is $\vp$-irredundant for $F$ for all $\vp \in \Pass$.
  \item\label{item:charakallpart2} $C$ is contraction-$\vp$-irredundant for $F$ for all $\vp \in \Pass$.
  \item\label{item:charakallpart3} $\modf_{\var(F)}(C) \cap \bc_{C' \in F \sm \set{C}} \modf_{\var(F)}(C') = \es$.
  \item\label{item:charakallpart4} $C$ clashes with every $C' \in F \sm \set{C}$, i.e., clause $C$ is connected in the conflict graph $\cg(F \cup \set{C})$ to every other vertex.
  \end{enumerate}
\end{lem}
\begin{proof} By the above remark we see that (\ref{item:charakallpart3}) and (\ref{item:charakallpart4}) are equivalent. By definition (\ref{item:charakallpart3}) is equivalent to (\ref{item:charakallpart1}), while by Lemma \ref{lem:TatsachenRedundanz}, part \ref{item:TatsachenRedundanz2} it is (\ref{item:charakallpart1}) equivalent to (\ref{item:charakallpart2}). \end{proof}

We remind at the notion of a ``hitting clause-set'' (see Subsection \ref{sec:conflictstructure} in Part I), where every two different clauses clash in at least one variable (have literals with the same variables but with different values). This is equivalent to the sets of falsifying assignments for different clauses being disjoint.

\begin{corol}\label{cor:charakhitirr}
  A generalised clause-set $F \in \Cls$ is a hitting clause-set if and only if for all $\vp \in \Pass$ the clause-set $\vp * F$ is irredundant.
\end{corol}
Obviously the application of a partial assignment to a hitting clause-set produces again a hitting clause-set. So we also have the simple fact, that $F$ is hitting iff for all partial assignments $\vp * F$ is hitting. Generalising Theorem 32 in \cite{Ku2003e}:
\begin{corol}\label{cor:charakhitting}
  A generalised clause-set $F \in \Cls$ is unsatisfiable hitting if and only if $\vp * F$ is minimally unsatisfiable for every $\vp \in \Pass$.
\end{corol}

Hitting clause-sets are irredundant; the more general class of \emph{multihitting clause-sets} (clause-sets with complete multipartite conflict graph) contains redundant clause-sets, but all redundancies can be removed efficiently (and canonically), as the following lemma shows. We use the notion of an \textbf{irredundant core} of a clause-set $F \in \Cls$ which is an irredundant $F' \sse F$ such that $F'$ is equivalent to $F$ (in \cite{Liberatore2005Redundanz} the notion ``irredundant equivalent subset'' is used). An irredundant core of an unsatisfiable clause-set is called a \textbf{minimally unsatisfiable core} (or simply ``minimally unsatisfiable sub-(clause)-set'', often abbreviated by ``mus'').

\begin{lem}\label{lem:multihittingmusat}
  Consider a generalised clause-set $F \in \Cls$ without trivial variables which is multihitting. Let $\FF$ be the multipartition of $F$, and $V := \var(F)$.
  \begin{enumerate}
  \item\label{item:multihittingmusat1} For $F_1, F_2 \in \FF$, $F_1 \not= F_2$ we have $\modf_V(F_1) \cap \modf_V(F_2) = \es$.
  \item\label{item:multihittingmusat2} If for $F' \sse F$ and $C \in F \sm F'$ we have $F' \models C$, then there must be some $C' \in F'$ with $C' \subset C$.
  \item\label{item:multihittingmusat3} $F$ has exactly one irredundant core, which is obtained from $F$ by subsumption-elimination. Thus if $F$ is unsatisfiable, then $F$ has exactly one minimally unsatisfiable core, which is obtained from $F$ by subsumption-elimination.
  \item\label{item:multihittingmusat4} A hitting clause-set $F$ is unsatisfiable iff $\sum_{C \in F} \abs{\modf_V(\set{C})} = \abs{\Pass(V)}$.
  \end{enumerate}
\end{lem}
\begin{proof} Part \ref{item:multihittingmusat1} follows by definition. In Part \ref{item:multihittingmusat2} it is $\modf_V(\set{C})$ covered by $\modf_V(F')$, and thus by Part \ref{item:multihittingmusat1} in fact $\modf_V(\set{C})$ is covered by $\modf_V(F' \cap F_C)$, where $F_C \in \FF$ with $C \in F_C$; i.e., $F_C \cap F' \models \set{C}$. By the strong completeness of resolution and the fact that within $F_C$ no clashes exist, it follows that there must be $C' \in F' \cap F_C$ with $C' \subset C$. Part \ref{item:multihittingmusat3} follows immediately from Part \ref{item:multihittingmusat2}. Finally Part \ref{item:multihittingmusat4} follows immediately from Part \ref{item:multihittingmusat1}. \end{proof}

\begin{corol}\label{cor:Charakmusatmulth}
  A multihitting clause-set is irredundant if and only if $F$ is sub\-sump\-tion-free. Thus an unsatisfiable multihitting clause-set is minimally unsatisfiable if and only if $F$ is subsumption-free.
\end{corol}

By Corollary \ref{cor:Charakmusatmulth} we know that deciding whether a multihitting clause-set is minimally unsatisfiable is the same task (up to subsumption elimination) as deciding whether it is unsatisfiable. Obviously $\Mclash \cap \Usat$ is in co-NP (and thus also $\Mclash \cap \Musat$). We have more precise information only for special cases:
\begin{itemize}
\item Using $\abs{\modf_{\var(F)}(C)} = \prod_{v \in \var(F) \sm \var(C)} \abs{D_v}$ for $C \in F$ it follows that satisfiability for hitting clause-sets is decidable in polynomial time (generalising the well-known special case for boolean clause-sets).
\item For boolean \emph{bihitting} clause-sets in \cite{GalesiKullmann2003bHermitian} it was shown that satisfiability decision can be done in quasi-polynomial time (where ``quasi-polynomial'' means a ``polynomial'' upper bound but where the exponent is allowed logarithmic growth in the size of the input), since satisfiability decision for bihitting clause-sets is essentially the same as deciding whether for two given hypergraphs one is the transversal hypergraph of the other.
\end{itemize}
The second point can immediately be generalised as follows.

\begin{lem}\label{lem:polyzeitbitreff}
  Satisfiability for bihitting generalised clause-sets is decidable in quasi-polynomial time.
\end{lem}
\begin{proof} Variables with a domain size greater than two appearing in a bihitting clause-set must be pure variables, since if a generalised clause-set contains a variable of domain size $k$, then the conflict graph contains the complete graph $K_k$ (which is not bipartite). \end{proof}

It seems to be a very interesting question, to what degree (generalised) multihitting clause-sets have efficient satisfiability decision (see Subsection \ref{sec:openmultihitting} for further discussion, and see \cite{Kullmann2004c} for more information in the boolean case).

\subsection{Irredundant cores}
\label{sec:irrcores}

A ``core'' of an unsatisfiable (boolean) clause-set is just some unsatisfiable sub-clause-set; typically one is interested in minimally unsatisfiable cores, and in this short section we generalise some basic observations in \cite{KullmannLynceSilva2005Autarkies} from the boolean context to generalised clause-sets, and now also considering satisfiable clause-sets.

Recall that a (finite) hypergraph is a pair $(V, \EE)$ such that $V$ is a (finite) set (the ``vertex set'') and $\EE$ is a set of subsets of $V$ (the set of ``hyperedges''). For a hypergraph $G$ by $\min(G)$ resp.\ $\max(G)$ we denote the hypergraph with the same vertex set and with all inclusion-minimal resp.\ maximal hyperedges from $G$. Consider a (generalised) clause-set $F$. Let $\eqs(F)$ be the hypergraph with vertex set $F$ (the clauses of $F$), while the hyperedges are all subsets of $F$ which are equivalent to $F$, and let $\neqs(F)$ be the hypergraph with vertex set $F$ and hyperedges the subsets of $F$ which are not equivalent to $F$. If $F$ is unsatisfiable, then $\eqs(F) = \Usat(F)$, the hypergraph consisting of all unsatisfiable sub-cause-sets of $F$, while $\neqs(F) = \Sat(F)$, the hypergraph of all satisfiable sub-clause-sets of $F$. Now $\min(\eqs(F))$ is the hypergraph consisting of all irredundant cores of $F$; if $F$ is unsatisfiable then $\min(\eqs(F)) = \mus(F)$, the hypergraph of all minimally unsatisfiable cores of $F$.

Generalising \cite{Liberatore2005Redundanz,KullmannLynceSilva2005Autarkies}, the elements of $\bca \eqs(F) = \bca \min(\eqs(F))$, the clauses which are in every irredundant core of $F$, are called \emph{necessary clauses}. Following \cite{KullmannLynceSilva2005Autarkies}, the elements of $\bc \min(\eqs(F))$, the clauses which are in some irredundant core, are called \emph{potentially necessary clauses} (in \cite{Liberatore2005Redundanz} such clauses are called ``useful''). We see that necessary clauses are exactly the irredundant clauses as defined before. Regarding decision complexity we have:
\begin{enumerate}
\item A clause-set $F$ is satisfiable iff $\bot$ is necessary for $F \cup \set{\bot}$, and thus already for (boolean) unsatisfiable clause-sets decision whether a clause is necessary is NP-complete (this was noticed for (arbitrary) boolean clause-sets in Theorem 3 in \cite{Liberatore2005Redundanz}, and trivially also the decision problem whether some clause is necessary for an generalised clause-sets is NP-complete as well).
\item By Theorem 4 in \cite{Liberatore2005Redundanz} we have that decision whether a clause $C$ is potentially necessary for a (boolean) clause-set $F$ is $\polysigma{2}$-complete (where $\polysigma{2}$ is the class of problems reducible to the decision problem whether a quantified boolean formula with quantifier-prefix $\ex^* \fa^*$ is true). Trivially this holds also for all generalised clause-sets, and due to $\bc \min(\eqs(F \cup \set{\bot})) = \set{\bot} \cup \bc \min(\eqs(F))$ we can restrict $F$ here again to unsatisfiable clause-sets.
\item In Theorem 5 in \cite{Liberatore2005Redundanz} it is shown that decision whether a (boolean) clause-set has a unique irredundant core is $\polydelta{2}[\log n]$-complete (where $\polydelta{2}$ is the class of problems decidable in polynomial time by arbitrary use of an NP-oracle, while for $\polydelta{2}[\log n]$ only logarithmically many oracle calls are allowed). Obviously this carries over to generalised clause-sets, however whether again restriction to unsatisfiable clause-sets is possible (that is, deciding whether an unsatisfiable clause-set has a unique minimally unsatisfiable core) is not clear.
\end{enumerate}
Finally we can also generalise the observation of Bailey and Stuckey, independently also made in \cite{Bruni2003ApprMU} (Theorem 2), and exploited in \cite{LiffitonSakallah2005AllMUS}, where we use the same (simplified) proof as in \cite{KullmannLynceSilva2005Autarkies} (Section 2): For a hypergraph $G$ denote by $\Tr(G)$ the hypergraph with the same vertex set $V(G)$, while the hyperedges are the minimal transversals of $G$ (minimal subsets of $V(G)$ intersecting every hyperedge), and denote by $\complement(G)$ the hypergraph with vertex set $V(G)$ and hyperedges $V(G) \sm H$ for $H \in E(G)$.
\begin{lem}\label{lem:meqhypgr}
  For every (generalised) clause-set $F$ we have
\begin{displaymath}
  \min(\eqs(F)) = \Tr(\complement(\max(\neqs(F)))).
\end{displaymath}
\end{lem}
\begin{proof} The assertion is equivalent to $\complement(\Tr(\min(\eqs(F)))) = \max(\neqs(F))$, which just states that the maximal non-equivalent sub-clause-sets of $F$ are exactly the maximal independent vertex sets of $\min(\eqs(F))$, i.e., those maximal sets of clauses not containing an irredundant core. \end{proof}

\section{Minimally unsatisfiable generalised clause-sets}
\label{sec:MUSATneu}

This chapter is about the basic facts regarding (generalised) unsatisfiable irredundant clauses-sets, that is, minimally unsatisfiable (generalised) clause-sets. See \cite{Kullmann2007HandbuchMU} for an overview on minimal unsatisfiability (and extensions) in the boolean case. The basic parameter structuring our considerations is the deficiency, and especially the lowest possible deficiency is considered.

In Subsection \ref{sec:saturated} ``saturated minimally unsatisfiable clause-sets'' are discussed (no literal occurrence can be added without destroying minimal unsatisfiability); this is a concrete example where generalised clause-sets behave essentially more complicated than boolean clause-sets. In Subsection \ref{sec:MUSATdefone} we characterise minimally unsatisfiable generalised clause-sets of deficiency one as well as the special cases of saturated and marginal clause-sets. Finally in Subsection \ref{sec:stabmmvd} we collect some observations which might serve for further progress in the characterisation of minimally unsatisfiable clause-sets.

\subsection{Saturated minimally unsatisfiable clause-sets}
\label{sec:saturated}

A clause-set $F \in \Cls$ is called \textbf{saturated minimally unsatisfiable} if $F$ is unsatisfiable, but for any clause $C \in F$ replacing $C$ in $F$ by $C \cup \set{x}$ for any literal $x$ with $\var(x) \notin \var(C)$ and $\abs{D_{\var(x)}} \ge 2$ yields a satisfiable clause-set.\footnote{Instead of ``saturated'' in \cite{AhLi86} ``strong'' is used, and in \cite{KleineBueningZhao2007ComplexitySomeSubclassesMU} ``maximal''; we follow \cite{FlRe94}.} Saturated minimally unsatisfiable clause-sets are minimally unsatisfiable (consider $x$ such that $\var(x) \notin \var(F)$), and actually a clause-set $F$ is saturated minimally unsatisfiable iff it is minimally unsatisfiable and addition of a literal $x$ with $\var(x) \in \var(F)$ to any clause $C$ with $\var(x) \notin \var(C)$ yields a satisfiable clause-sets. The set of all saturated minimally unsatisfiable clause-sets is called \bmm{\Smusat}. By Lemma \ref{lem:multihittingmusat}, part \ref{item:multihittingmusat4} we see that unsatisfiable hitting clause-sets are in $\Smusat$.

\begin{lem}\label{lem:Saturierung}
  Every minimally unsatisfiable clause-set $F \in \Musat$ can be \textbf{saturated}, that is, there exists $F^* \in \Smusat$ with $\var(F^*) = \var(F)$ and a bijection $\pi: F \ra F^*$ such that for all $C \in F$ we have $C \sse \pi(C)$.
\end{lem}
\begin{proof} The observation needed here is that if for a minimally unsatisfiable clause-set $F$ we replace some clause $C \in F$ by a clause $C' \supset C$, obtaining $F' := (F \sm \set{C}) \cup \set{C'}$, then $F'$ is minimally unsatisfiable if $F'$ is unsatisfiable (the only possibly redundant clause in $F'$ is $C'$, and if $C'$ is redundant in $F'$, then $F'$ is satisfiable, since $F' \sm \set{C'} = F \sm \set{C} \in \Sat$). So we can add literals $x$ with $\var(x) \in \var(F)$ to clauses such that we maintain (minimally) unsatisfiability, and finally we will end up with a saturated $F^*$. \end{proof}

For boolean clause-sets the characterisation of $\Smusat$ from Lemma C.1 in \cite{Ku99dKo} is fundamental: A minimally unsatisfiable boolean clause-set $F$ is saturated if and only if for every variable $v \in \var(F)$ and each $\ve \in D_v = \set{0,1}$ it is $\pao{v}{\ve} * F$ minimally unsatisfiable. Together with saturation this characterisation provides a powerful method for proving properties of minimally unsatisfiable clause-sets via induction on the number of variables.  For generalised clause-sets saturatedness is weaker, and the above condition is only sufficient for being minimally unsatisfiable, but is no longer necessary. The following lemma develops these fundamental facts, using the following notion: We say that \emph{addition of literal $x$ renders clause-set $F$ satisfiable} iff for all clauses $C \in F$ with $\var(x) \notin \var(C)$ the clause-set $(F \sm \set{C}) \cup \set{C \cup \set{x}}$ is satisfiable (thus a clause-set $F$ is saturated minimally unsatisfiable iff $F$ is unsatisfiable and addition of any literal renders $F$ satisfiable).

\begin{lem}\label{lem:saturierungmusatbelegung}
  Consider a generalised clause-set $F \in \Musat$ and a literal $(v, \ve) \in \Lit$.
  \begin{enumerate}
  \item\label{item:saturierungmusatbelegung1} If $\pao{v}{\ve} * F \in \Musat$, then for all $\ve' \in D_v \sm \set{\ve}$ addition of literal $(v,\ve')$ renders $F$ satisfiable.
  \item\label{item:saturierungmusatbelegung2} If $v$ is boolean, and for $\ve' \in D_v \sm \set{\ve}$ addition of literal $(v,\ve')$ renders $F$ satisfiable, then we have $\pao{v}{\ve} * F \in \Musat$.
  \end{enumerate}
\end{lem}
\begin{proof} For Part \ref{item:saturierungmusatbelegung1} assume that there is $C \in F$, $v \notin \var(C)$ and $\ve' \in D_v \sm \set{\ve}$ such that $F' := (F \sm \set{C}) \cup \set{C \cup \set{(v,\ve')}}$ is unsatisfiable. Then $\pao{v}{\ve} * F' \in \Usat$ with $\pao{v}{\ve} * F' = (\pao{v}{\ve} * F) \sm \set{C}$, and thus $C$ would be redundant in $\pao{v}{\ve} * F$.

For Part \ref{item:saturierungmusatbelegung2} assume that $\pao{v}{\ve} * F$ is not minimally unsatisfiable; by Corollary \ref{cor:TatsachenRedundanz2} thus there is a clause $C \in F$, $v \notin \var(C)$ such that $F \sm \set{C} \models C \cup \set{(v,\ve)}$. It follows that for $F' := (F \sm \set{C}) \cup \set{C \cup \set{(v,\ve')}}$ we have $F' \models C$ (using one resolution step), and thus $F'$ would be unsatisfiable. \end{proof}

\begin{corol}\label{cor:saturierungmusatbelegung}
  If for the generalised clause-set $F \in \Cls$ for every partial assignment $\vp \in \Pass$ with $n(\vp) \le 1$ we have $\vp * F \in \Musat$, then $F \in \Smusat$. If $F$ is boolean, then also the reverse direction holds, that is, $F \in \Smusat$ if and only if for every partial assignment $\vp \in \Pass$ with $n(\vp) \le 1$ we have $\vp * F \in \Musat$.
\end{corol}

An example, showing that the implication ``$F \in \Smusat \Ra \pao{v}{\ve} * F \in \Musat$'' in Corollary \ref{cor:saturierungmusatbelegung} does not hold for generalised clause-sets, is as follows: Consider variables $a, b$ with $D_a = D_b = \set{0,1,2}$, and let $F$ be the following clause-set with $4$ binary clauses and $2$ unary clauses:
\begin{multline*}
  F := \setb{ \set{a \not= 0, b \not= 0}, \, \set{a \not= 1, b \not= 0}, \, \set{a \not= 0, b \not= 1}, \, \set{a \not= 1, b \not= 1}, \\
    \set{a \not= 2}, \, \set{b \not= 2} }.
\end{multline*}
We have $F \in \Smusat$ (after unit-clause elimination we obtain a boolean clause-set with all possible (full) clauses), while $\pao a2 * F = \set{\bot, \set{b \not= 2}} \notin \Musat$ (as well as $\pao b2 * F = \set{\bot, \set{a \not= 2}} \notin \Musat$). It might be worth investigating the class of (generalised) clause-sets $F$ such that for all partial assignments $\vp$ with $n(\vp) \le 1$ we have $\vp * F \in \Musat$ (a strict subset of $\Smusat$); see Subsection \ref{sec:stabmmvd} for first observations.

An important application of the process of saturation for \emph{boolean} clause-sets is given by Lemma C.2 in \cite{Ku99dKo}, proving that for every $F \in \Musat$, $F \not= \set{\bot}$ there is a variable $v \in \var(F)$ such that for both $\ve \in D_v = \set{0,1}$ we have $\#_{(v,\ve)}(F) \le \delta(F)$. The proof is based on the characterisation of saturated minimally unsatisfiable boolean clause-sets in Corollary \ref{cor:saturierungmusatbelegung} and uses $\delta(F) \ge 1$ for $F' \in \Musat$, where $F'$ is obtained from $F$ by applying suitable partial assignments $\vp$ with $n(\vp) = 1$. There are various possibilities to obtain a generalisation for generalised clause-sets (the problem is that saturation is not that powerful anymore):
\begin{itemize}
\item in Lemma \ref{lem:GrundtatsacheDefektEins} we obtain the generalisation to generalised clause-sets in the special case of deficiency one,
\item while in Subsection \ref{sec:stabmmvd} we consider the class of minimally unsatisfiable clause-sets stable under application of partial assignments with at most one variable, for which then the general bound can be shown (in Corollary \ref{cor:mmvdobSchr}).
\end{itemize}
The existence of a variable $v$ in the boolean case with $\#_{(v,0)}(F), \#_{(v,1)}(F) \le \delta(F)$ yields that the minimal variable-degree of $F$ is at most $2 \delta(F)$. Even for lean (boolean) clause-sets this can be strengthened considerably, as shown in \cite{KullmannZhao2010Bounds}, while the proper generalisation to generalised clause-sets is open; see Corollary \ref{cor:mvdobSchr} for a first result in this direction.

\subsection{Characterisation of the basic case of deficiency one}
\label{sec:MUSATdefone}

Generalising the tree construction from \cite{Ku99dKo} (exploiting a formula class introduced by Stephen Cook and communicated to me by Alasdair Urquhart), let a \textbf{deficiency-$1$ tree representation} (in the remainder of this section just called ``tree representation'') be a $4$-tuple $(T, r, v, \ve)$, where
\begin{itemize}
\item $(T, r)$ is a finite tree with root $r$ (inner nodes (that is, nodes which are not leaves) can have an arbitrary number of children).
\item $v$ labels each inner node $w$ of $(T, r)$ with a unique variable $v(w)$.
\item $\ve$ labels each edge $e$ leading from a node $w$ to a node $w'$ (edges are directed from the root towards the leaves) with a value $\ve(e) \in D_{v(w)}$ such that the labelling of the edges going out from $w$ yields a bijection to $D_{v(w)}$.
\end{itemize}
If an order on the value set $D_{v(w)}$ is given, then also the outgoing edges are ordered by the same order; in the special case of boolean variables thus we can speak of ``left'' and ``right'' branches, corresponding to the positive and negative literal. An example $R$ is given as follows.

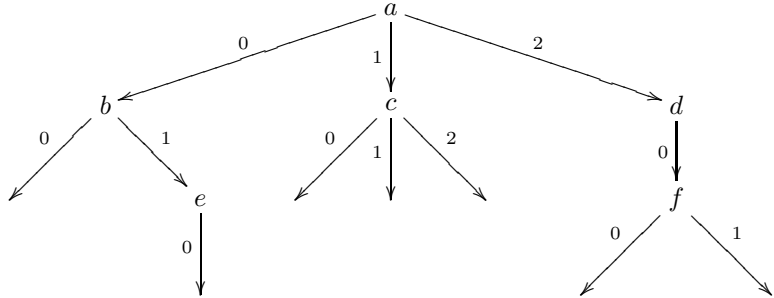
\begin{displaymath}
  \xygraph{
    a (
      :[dlll]  {b}_0 ( :[dl]_0, :[dr]{e}^1  (:[d]_0)), :[d] {c}_1 (:[dl]_0, :[d]_1, :[dr]^2), :[drrr] {d}^2 (:[d] {f}_0 (:[dl]_0, :[dr]^1))
    )
  }
\end{displaymath}

This tree representation $R$ uses six variables $a,\dots,f$ with $D_a = D_c = \set{0,1,2}$, $D_b = D_f = \set{0,1}$ (thus $b, f$ are boolean variables), and $D_d = D_e = \set{0}$.

Given a tree representation $(T, r, v, \ve)$, to every node $w$ of $(T, r)$ we associate a clause $C_w$ by considering the path $w_0, e_1, w_1, \dots, e_m, w_m$ from the root to $w$ in $T$ (thus $w_0 = r$, $w_m = w$, and the $e_i$ are the connecting edges from $w_{i-1}$ to $w_i$, while $m$ is the length of the path), and setting $C_w := \set{(v(w_i), \ve(e_{i+1})) : i \in \tb 0{m-1}}$. The clause-set \bmm{F(T, r, v, \ve)} is defined as the set of all clauses $C_w$ for leaves $w$ of $(T,r)$. For the above example $R$ we get
\begin{multline*}
  F(R) = \setb{
    \set{a \not= 0, b \not= 0}, \ \set{a \not= 0,b \not= 1,e \not= 0},\\
    \set{a \not= 1,c \not= 0}, \ \set{a \not= 1,c \not= 1}, \ \set{a \not= 1c \not= 2},\\
    \set{a \not= 2,d \not= 0,f \not= 0}, \ \set{a \not= 2,d \not= 0,f \not= 1} }.
\end{multline*}
We list some basic properties of the clause-sets $F(T, r, v, \ve)$:
\begin{enumerate}
\item\label{item:propsmusatdone1} The rooted tree $(T,r)$ yields a resolution tree for $F(T, r, v, \ve)$ by labelling the nodes $w$ with clauses $C_w$ and considering the variables $v(w)$ for inner nodes $w$ as resolution variables; since $C_r = \bot$ we see that $F(T, r, v, \ve)$ is unsatisfiable.
\item\label{item:propsmusatdone2} $F(T, r, v, \ve)$ is a $1$-regular hitting clause-set (for two different clauses $C_{w_1}, C_{w_2}$ the unique clashing variable is $v(w_0)$ for the root $w_0$ of the smallest subtree of $(T,r)$ containing $w_1$ and $w_2$). It follows that $F(T, r, v, \ve)$ is saturated minimally unsatisfiable.
\item\label{item:propsmusatdone3} $\delta(F(T, r, v, \ve)) = 1$, since $c(F(T, r, v, \ve))$ is the number of leaves of $(T,r)$, while $\rd(F(T, r, v, \ve))$ is the number of edges of $T$ minus the number of inner nodes of $(T,r)$, and thus $\delta(F(T, r, v, \ve))$ is the difference of the number of vertices and the number of edges of $T$, which is $1$ for every tree.
\item\label{item:propsmusatdone4} If $n(F(T, r, v, \ve)) > 0$ (that is, if $(T,r)$ is not trivial), then we have:
  \begin{enumerate}
  \item\label{item:propsmusatdone4a} There is exactly one variable occurring in every clause of $F(T, r, v, \ve)$ (namely $v(r)$).
  \item\label{item:propsmusatdone4b} Every clause $C \in F(T, r, v, \ve)$ contains a literal $x \in C$ with $\#_x(F) = 1$ (namely with $\var(x) = v(w_0)$, where $C = C_w$ and $w_0$ is the parent node of $w$).
  \item\label{item:propsmusatdone4c} There exists a variable $v \in \var(F(T, r, v, \ve))$ such that for all values $\ve \in D_v$ we have $\#_{(v,\ve)}(F(T, r, v, \ve)) = 1$ (choose $v = v(w)$ for an inner node $w$ of $(T,r)$ such that all children of $w$ are leaves).
  \end{enumerate}
\end{enumerate}
We can read off many more properties of $F(T, r, v, \ve)$ directly from the tree representation, for example the minimal resp.\ maximal clause-length is the minimal resp.\ maximal depth of a leaf, but we need here only the above listed properties. Using $\Sclash$ for the set of regular hitting clause-sets and $\hitdeg$ for the hitting degree, as introduced before, we have $F(T, r, v, \ve) \in \Sclashk{\hitdeg=1,\delta=1}^{\sat=0}$.

We say that a clause-set $F' \in \Cls$ is obtained from $F(T, r, v, \ve)$ by \textbf{literal elimination} if $F'$ is obtained from $F(T, r, v, \ve)$ by eliminating some literal occurrences (at least one) without ever creating a pure variable.  Replacing ``$F(T, r, v, \ve)$'' by $F'$, Properties \ref{item:propsmusatdone1}, \ref{item:propsmusatdone3}, \ref{item:propsmusatdone4b}, \ref{item:propsmusatdone4c} are still valid, while Properties \ref{item:propsmusatdone2}, \ref{item:propsmusatdone4a} are lost: $F'$ is definitely not a hitting clause-set anymore, and there does not need to exist a variable occurring in every clause. It is furthermore $F'$ definitely not saturated anymore (by the definition of $F'$), however $F'$ is still minimally unsatisfiable (since removal of any clause either creates a pure variable or removes the only clause).

In Lemma C.5 from \cite{Ku99dKo} it is shown that the boolean elements of $\Smusati{\delta=1}$ are exactly the clause-sets $F(T, r, v, \ve)$ using only boolean variables, while the elements of $\Musati{\delta=1} \sm \Smusati{\delta=1}$ are exactly the clause-sets obtained from such $F(T, r, v, \ve)$ by literal elimination. To generalise this characterisation, the following lemma is central (compare Property \ref{item:propsmusatdone4c} from above).

\begin{lem}\label{lem:GrundtatsacheDefektEins}
  For every (generalised) clause-set $F \in \Musati{\delta=1}$ with $n(F) > 0$ there exists a variable $v \in \var(F)$ such that for all $\ve \in D_v$ we have $\#_{(v,\ve)}(F) = 1$.
\end{lem}
\begin{proof} Consider $F \in \Musati{\delta=1}$. We investigate the structure of $\ftrans(F)$ (recall Section \ref{sec:translating}). As we remarked in Subsection \ref{sec:Preservationgeneralstructure}, we have $\delta(\ftrans(F)) = 1$, and thus by Lemma \ref{lem:Eigtrans3} we have $\ftrans(F) \in \Musati{\delta=1}$. Since $\ftrans(F)$ is a boolean clause-set, we can conclude that $\ftrans(F)$ is obtained by literal elimination from some tree representation $(T,r,v,\ve)$ as defined above (using only boolean variables). $\ftrans(F)$ always has the following special properties:
\begin{enumerate}[(i)]
\item $\ftrans(F)$ is a PN-clause-set, that is, every clause is either positive or negative.
\item For every negative clause $N \in \ftrans(F)$ we have $\fa\, x \in N : \#_x(F) = 1$ (recall that the negative clauses are the $\alo$-clauses introduced by the translation $\ftrans$).
\end{enumerate}
Call a boolean $F \in \Musati{\delta=1}$ \emph{special} if these two conditions are fulfilled. (These ``special'' boolean clause-sets constitute exactly the image $\ftrans(\Musati{\delta=1})$ of the translation, but we do not need this simple fact here.) Consider a tree representation $(T, r, v, \ve)$ of a special $F$; obviously also all clause-sets given by the subtrees of $(T, r)$ are special again. Now we proof by induction on the height of the tree representation of special formulas $F$ with $n(F) > 0$ that there always exists a negative clause $N \in F$ such that $\fa\, x \in N : \#_{\ol{x}}(F) = 1$, using the standard complement notation for boolean literals here; this proves the lemma by definition of the translation $\ftrans$.

If the height of $(T, r)$ is $1$, then $F$ is $\set{\set{v(r)},\set{\ol{v(r)}}}$, and the assertion is true. So assume the height of $(T, r)$ is greater than $1$, and consider the left subtree $T_0$ and the right subtree $T_1$ of $T$ with associated special $F_0, F_1 \in \Musati{\delta=1}$. If $T_0$ is not the trivial tree (has more than one node), then by the induction hypothesis there exists a negative clause (non-empty) $N_0 \in F_0$ with $\fa\, x \in N_0 : \#_{\ol{x}}(F_0) = 1$. Now we must have $N_0 \in F$, since otherwise $N_0 \cup \set{v} \in F$, where this clause is neither positive nor negative; using $N := N_0$ proves the assertion (since none of the variables in $N_0$ occurs in $T_1$ in this case). So the remaining case is that $T_0$ is the trivial tree. Again by the induction hypothesis there is a negative clause (non-empty) $N_1 \in F_1$ with $\fa\, x \in N_1 : \#_{\ol{x}}(F_1) = 1$. Either we have $N := N_1 \in F$ or $N := N_1 \cup \set{\ol{v}} \in F$, proving the assertion (in the second case due to the triviality of $T_0$). \end{proof}

Now we are able to generalise Lemma C.5 in \cite{Ku99dKo} (Part (\ref{item:CharakMUSATd1a}) of Theorem \ref{thm:CharakMUSATd1} has been shown for boolean clause-sets in \cite{DDK98}):

\begin{thm}\label{thm:CharakMUSATd1}
  The class $\Musati{\delta=1}$ of minimally unsatisfiable (generalised) clause-sets of deficiency $1$ has the following two characterisations:
  \begin{enumerate}[(i)]
  \item\label{item:CharakMUSATd1a} For $F \in \Cls$ we have $F \in \Musati{\delta=1}$ if and only if $F$ can be reduced to the clause-set $\set{\bot}$ by applying non-degenerated singular DP-reduction (as long as possible, in any order).
  \item\label{item:CharakMUSATd1b} $\Musati{\delta=1}$ is the class of all clause-sets $F(T, r, v, \ve)$ together with all clause-sets $F'$ derived by literal elimination from such clause-sets.
  \end{enumerate}
 \end{thm}
\begin{proof} Part (\ref{item:CharakMUSATd1a}) follows from Lemma \ref{lem:GrundtatsacheDefektEins} together with Lemma \ref{lem:singDPMUSAT} (Part I) and Lemma \ref{lem:dpausnahme}, Part \ref{item:dpausnahmeb1} (also Part I). For Part (\ref{item:CharakMUSATd1b}) it remains to show that every $F \in \Musati{\delta=1}$ can be obtained from some $F(T,r,v,\ve)$ by a (possibly empty) sequence of literal eliminations. We show this by induction on $n(F)$.  If $n(F) = 0$, then $F = \set{\bot}$, and we can take the trivial rooted tree. So assume $n(F) > 0$. By Lemma \ref{lem:GrundtatsacheDefektEins} there exists a variable $v \in \var(F)$ such that for all $\ve \in D_v$ we have $\#_{(v,\ve)}(F) = 1$; let $C_{\ve} \in F$ be the unique clause with $(v,\ve) \in C_{\ve}$. Thus $v$ is a singular DP-variable w.r.t.\ $F$. Let $G := \dpi{v}(F)$; we have $G = (F \sm \set{C_{\ve}}_{\ve \in D_v}) \cup \set{R}$, where $R = \bc_{\ve \in D_v} (C \sm \set{(v,\ve)})$. As already argued for Part (\ref{item:CharakMUSATd1a}) we have $G \in \Musati{\delta=1}$, and thus we can apply the induction hypothesis to $G$; the assertion follows now immediately by extending the tree representation of $G$ at the leaf labelled by $R$ by adding new leaves $C_{\ve}$ for $\ve \in D_v$. \end{proof}

Theorem \ref{thm:CharakMUSATd1} yields also two further poly-time decision procedures for the class $\Musati{\delta=1}$ (while two general poly-time decision procedures for the classes $\Musati{\delta=k}$ for $k \in \NN$ are given by Corollary \ref{cor:poly1} (Part I) and Theorem \ref{thm:MaximalerDefektFPT}). To conclude, we characterise the saturated and the marginal elements of $\Musati{\delta=1}$.

\begin{corol}\label{cor:CharakSMUSATd1}
  The class $\Smusati{\delta=1}$ of saturated minimally unsatisfiable (generalised) clause-sets of deficiency $1$ is exactly the class of all clause-sets $F(T, r, v, \ve)$. It follows that the following conditions are equivalent for a clause-set $F \in \Cls$:
  \begin{enumerate}
  \item $F = F(T, r, v, \ve)$ for some deficiency-$1$ tree representation $(T,r,v,\ve)$.
  \item $F$ is an unsatisfiable $1$-regular hitting clause-set of deficiency $1$ (i.e., $F \in \Sclashk{\hitdeg=1,\delta=1}^{\sat=0}$).
  \item $F$ is an unsatisfiable regular hitting clause-set of deficiency $1$ (i.e., $F \in \Sclashk{\delta=1}^{\sat=0}$).
  \item $F$ is an unsatisfiable hitting clause-set of deficiency $1$ (i.e., $F \in \Clashi{\delta=1}^{\sat=0}$).
  \item $F$ is a saturated minimally unsatisfiable clause-set of deficiency $1$ (i.e., $F \in \Smusati{\delta=1}$).
  \end{enumerate}
\end{corol}

If a minimally unsatisfiable clause-set is hitting, then it is saturated; Corollary \ref{cor:CharakSMUSATd1} proves the reverse for deficiency $1$ (which does not hold for higher deficiencies). Corollary \ref{cor:CharakSMUSATd1} shows that $\Smusati{\delta=1}$ covers some unsatisfiable regular hitting clause-sets --- in Corollary \ref{cor:CharakUnsatreghit} we will see that it actually covers all of them.

While saturated minimally unsatisfiable clause-sets do not allow addition of any literal occurrence to any clause without destroying the property of being minimally \emph{unsatisfiable}, on the other end of the spectrum we have \textbf{marginal minimally unsatisfiable clause-sets}, which are minimally unsatisfiable clause-sets such that removing any literal occurrence destroys the property of being \emph{minimally} unsatisfiable.

\begin{corol}\label{cor:CharakMMUSATd1}
  The class $\Mmusati{\delta=1}$ of marginal minimally unsatisfiable (generalised) clause-sets of deficiency $1$ is exactly the class of all $F \in \Musati{\delta=1}$ for which no further literal eliminations are possible, which is equivalent to the property that $F$ is totally singular (recall Subsection \ref{sec:matchautsub} from Part I), that is, for every variable $v \in \var(F)$ and every $\ve \in D_v$ we have $\#_{(v,\ve)}(F) = 1$.

  Thus $\Mmusati{\delta=1}$ is the set of totally singular elements of $\Musati{\delta=1}$. Actually, $\Mmusati{\delta=1}$ is the set of totally singular elements of all of $\Musat$.
\end{corol}
\begin{proof} If for a minimally unsatisfiable clause-set $F$ every literal in it occurs exactly once, then obviously it is marginal; by Theorem \ref{thm:CharakMUSATd1} we obtain the reverse direction (since if for some value more than one occurrence of a variable is left, then literal elimination is still applicable). 

If on the other hand a totally singular $F$ is minimally unsatisfiable, then it is minimally matching unsatisfiable, and thus by Corollary \ref{cor:minmatunsat} (Part I) it has deficiency $1$. \end{proof}

We remark that in \cite{Ku2003d} it is shown that in the boolean case the class of conflict graphs of $F \in \Musati{\delta=1}$ is exactly the class of all connected graphs, while the conflict graphs of saturated (boolean) $F \in \Musati{\delta=1}$ are exactly all complete graphs, and the conflict graphs of marginal (boolean) $F \in \Musati{\delta=1}$ are exactly all trees; furthermore a boolean element of $\Musati{\delta=1}$ is saturated resp.\ marginal iff the conflict graph is complete respectively a tree. And a totally singular multi-clause-set is minimally unsatisfiable iff its conflict multigraph is a tree.

Finally we remark that in general the decision problem whether a (generalised) clause-set is saturated resp.\ marginal is $D^P$-complete as shown in \cite{KleineBueningZhao2007ComplexitySomeSubclassesMU} (there for boolean clause-sets, which obviously immediately generalises).

\subsection{Stability parameter and minimal variable degree}
\label{sec:stabmmvd}

Let us conclude the chapter by some considerations which summarise certain observations we made, and which could serve as basis for further investigations.

For a multi-clause-set $F$ let $\bmm{\stabpar(F)}$, the \textbf{(substitution) stability parameter regarding irredundancy}, be the supremum in $\ZZ_{\ge -1} \cup \set{+\infty}$ of $n \in \NNZ$ such that for all $\vp \in \Pass$ with $n(\vp) \le n$ the multi-clause-set $\vp * F$ is irredundant. We have the following basic properties.
\begin{enumerate}
\item $\stabpar(F) = -1$ iff $F$ is redundant, $\stabpar(F) \ge 0$ iff $F$ is irredundant.
\item $\stabpar(F) = + \infty$ iff $\stabpar(F) \ge n(F)$ iff $F$ is a hitting clause-set (see Corollary \ref{cor:charakhitirr}).
\item Assume that $F$ is unsatisfiable. Then $\stabpar(F) \ge 1 \Ra F \in \Smusat$ by Corollary \ref{cor:saturierungmusatbelegung}, and for boolean $F$ we have equivalence.
\end{enumerate}
For $n(F) > 0$ let the \textbf{min-max var-degree} resp.\ the \textbf{minimal var-degree} be defined by
\begin{eqnarray*}
  \bmm{\mmvd}(F) & := & \min_{v \in \var(F)} \max_{\ve \in D_v} \: \#_{(v,\ve)}(F) \in \NN\\
  \bmm{\mvd}(F) & := & \min_{v \in \var(F)} \#_v(F) \in \NN.
\end{eqnarray*}
An upper bound $\mmvd(F) \le k$ says that there exists a variable $v$ such that for all values $\ve$ of $v$ we have at most $k$ occurrences of $v$ with that ``polarity''. So Lemma \ref{lem:GrundtatsacheDefektEins} now can be reformulated as the statement $\fa\, F \in \Musati{\delta=1} : F \not= \set{\bot} \Ra \mmvd(F) = 1$.

\begin{lem}\label{lem:mmvdobSchr}
  Consider a generalised multi-clause-set $F$ and a variable $v \in \var(F)$ which is not pure for $F$ (i.e., $\fa\, \ve \in D_v : \#_{(v,\ve)}(F) \ge 1$), such that $\#_v(F) = \mvd(F)$. Then we have for $\ve \in D_v$ the following.
  \begin{enumerate}
  \item\label{lem:mmvdobSchr1} $\delta(\pao{v}{\ve} * F) = \delta(F) - s_{(v,\ve)}(F) + \abs{D_v} - 1$. (Compare Lemma \ref{lem:HauptanwendungSurplus}, Part \ref{item:HauptanwendungSurplus2} from Part I.)
  \item\label{lem:mmvdobSchr2} Assume $\stabpar(F) \ge 1$ and that $F$ is unsatisfiable.
    \begin{enumerate}
    \item\label{lem:mmvdobSchr2a} $s_{(v,\ve)}(F) \le \delta(F) + \abs{D_v} - 2$.
    \item\label{lem:mmvdobSchr2b} If $v$ is non-trivial (i.e., $\abs{D_v} \ge 2$), then $\#_{(v,\ve)}(F) \le \delta(F)$.
    \end{enumerate}
  \end{enumerate}
\end{lem}
\begin{proof} Part \ref{lem:mmvdobSchr1} follows by the observation that $\var(\pao{v}{\ve} * F) = \var(F) \sm \set{v}$ (if another variable $w$ would vanish, then every occurrence of $w$ would be in a clause $C$ with some $(v,\ve') \in C$ for $\ve' \in D_v \sm \set{\ve}$, and so $\#_w(F) \le s_{(v,\ve)}(F) = \#_v(F) - \#_{(v,\ve)}(F) < \#_v(F) = \mvd(F) \le \#_w(F)$). For Part \ref{lem:mmvdobSchr2} we have $\delta(\pao{v}{\ve} * F) \ge 1$, and thus Part \ref{lem:mmvdobSchr2a} follows. For Part \ref{lem:mmvdobSchr2b} consider $\ve' \in D_v \sm \set{\ve}$. By Part \ref{lem:mmvdobSchr2a} we have $s_{(v,\ve')}(F) \le \delta(F) + \abs{D_v} - 2$, where $s_{(v,\ve')}(F) = \#_{(v,\ve)}(F) + \sum_{\ve'' \in D_v \sm \set{\ve,\ve'}} \#_{\ve''}(F) \ge \#_{(v,\ve)}(F) + \abs{D_v} - 2$. \end{proof}

\begin{corol}\label{cor:mmvdobSchr}
  For an unsatisfiable generalised clause-set $F$ with $\stabpar(F) \ge 1$, $n(F) > 0$, we have 
  \begin{displaymath}
    \mmvd(F) \le \delta(F).
  \end{displaymath}
\end{corol}
\begin{proof} Eliminating all trivial variables from $F$ we obtain the clause-set $F'$ with $\delta(F') = \delta(F)$ and $c(F') = c(F)$; now the assertion follows by Part \ref{lem:mmvdobSchr2b} of Lemma \ref{lem:mmvdobSchr}. \end{proof}

Since every minimally unsatisfiable (generalised) clause-set can be saturated (see Lemma \ref{lem:Saturierung}), and every boolean saturated minimally unsatisfiable clause-set $F$ fulfils $\stabpar(F) \ge 1$, we get for arbitrary \emph{boolean} $F \in \Musat$ the upper bound $\mmvd(F) \le \delta(F)$ (as shown in \cite{Ku99dKo}). For generalised minimally unsatisfiable clause-sets we showed this upper bound in Lemma \ref{lem:GrundtatsacheDefektEins} for the simplest case $\delta(F) = 1$, while the general case is open.

\begin{corol}\label{cor:mvdobSchr}
  Consider an unsatisfiable generalised clause-set $F$ with $\stabpar(F) \ge 1$ and $n(F) > 0$, and let $D \in \NN$ be the maximal domain-size. Then we have
  \begin{displaymath}
    \mvd(F) \le D \cdot \delta(F).
  \end{displaymath}
\end{corol}

For boolean clause-sets, a stronger and also more general bound is shown in \cite{KullmannZhao2010Bounds}.

\section{The conflict structure of generalised clause-sets}
\label{sec:construgcls}

In this final section we conclude the basic combinatorial theory of generalised clause-sets by regarding their conflict structure, continuing \cite{Ku2003c,Ku2003e,GalesiKullmann2003bHermitian,Ku2003d}. A fundamental tool again is translation to boolean clause-sets, but introducing here a new translation in Subsection \ref{sec:newtrans}, which, different from the direct translation, leaves the conflict structure invariant. Then in Subsection \ref{sec:Regularhitting} the basic facts from linear algebra regarding the ``hermitian rank/defect'' are summarised and generalised to generalised clause-sets, discussing the close relation to the theory of multiclique partitions of multigraphs and the applications to regular hitting clause-sets.

Recall from Subsection \ref{sec:conflictstructure} (Part I) that an $r$-regular hitting clause-set $F$ is a (generalised) clause-set $F$ such that any two different clauses clash in exactly $r \in \NNZ$ clauses. Using the \textbf{conflict multigraph} $\cmg(F)$ of $F$, the multigraph with the clauses of $F$ as vertices and as many edges between vertices as the clauses have conflicts, we see that $F$ is $r$-regular hitting iff $\cmg(F)$ is isomorphic to $r \cdot K_{c(F)}$, where $K_m$ denotes the complete graph with $m$ vertices, while the factor $r$ makes $r$ parallel edges out of every edge. The \textbf{conflict matrix} $\scf(F)$ of $F$ is the adjacency matrix of $\cmg(F)$, that is, the square matrix of order $c(F)$ where the entry at position $(i,j)$ denotes the number of clashes between clauses $i$ and $j$. Every conflict matrix is symmetric, non-negative with a zero diagonal and with integral entries; $F$ is $r$-regular hitting iff all entries of $\scf(F)$ except of diagonal-entries are equal to $r$.

\subsection{A new boolean translation}
\label{sec:newtrans}

We introduce a new translation of generalised multi-clause-sets $F$ into boolean clause-sets $F \mapsto \ntrans(F)$, called the \textbf{nested translation}, with the advantage over the direct translation (recall Section \ref{sec:translating}) that the conflict structure is preserved, that is, we have $\scf(F) = \scf(\ntrans(F))$; the price we have to pay here is the introduction of a ``nesting structure'', which breaks the symmetry between the different values of a variable.

Consider a variable $v$ with domain $D_v = \set{\ve_1, \dots, \ve_k}$ (where $k = \abs{D_v}$). The conflict multigraph of the clause-set $\set{\set{v \not= \ve_i} : i \in \tb 1k} \in \Smusati{\delta=1}$ is isomorphic to the complete graph $K_k$ with $k$ vertices. The nested translation is based on the following exact realisation of this conflict graph by a boolean clause-set\footnote{other minimally unsatisfiable realisations lead to other translations, which could be used here as well; see Subsection \ref{sec:opentranslations}}, namely by the minimally unsatisfiable boolean Horn clause-set $H_v = \set{E_{v,1}, \dots, E_{v,k}} \in \Smusati{\delta=1}$ with $k-1$ variables $v_1, \dots, v_{k-1}$, where
\begin{displaymath}
  E_{v,i} := 
  \begin{cases}
    \set{\ol{v_1}, \dots, \ol{v_{i-1}}, v_i} & \text{if } i < k\\
    \set{\ol{v_1}, \dots, \ol{v_{k-1}}} & \text{if } i = k
  \end{cases}.
\end{displaymath}
For example for $k=4$ we have the $3$ boolean variables variables $v_1,v_2,v_3$ and the clauses $E_{v,1} = \set{v_1}$, $E_{v,2} = \set{\ol{v_1},v_2}$, $E_{v,3} = \set{\ol{v_1},\ol{v_2},v_3}$ and $E_{v,4} = \set{\ol{v_1},\ol{v_2},\ol{v_3}}$.

Now the nested translation $F \mapsto \ntrans(F)$ replaces iteratively one variable $v$ with domain size $k_v = \abs{D_v} > 2$ by boolean variables $v_1, \dots, v_{k_v-1}$, without changing the number of clauses or their conflict structure, but enlarging clauses. The replacement of a single variable $v$, a transition $F \mapsto F_v$ repeated for all variables with at least three values, happens as follows:
\begin{itemize}
\item For $i \in \tb 1k$ let $F_i$ be the sub-multi-clause-set of $F$ collecting all occurrences of literal $(v,\ve_i)$ (thus $c(F_i) = \#_{(v,\ve_i)}(F)$).
\item Consider new boolean variables $v_1, \dots, v_{k-1}$ and the Horn clause-set $H_v = \set{E_{v,1}, \dots, E_{v,k}}$ as above; the meaning of $v_i$ is related to ``$v \not= \ve_i$''.
\item Now for $i \in \tb 1k$ replace each clause $C \in F_i$ by clause $(C \sm \set{(v,\ve_i)}) \cup E_{v,i}$ (replacing the literal $(v,\ve_i)$ by the literals from $E_{v,i}$).
\item For the obtained multi-clause-set $F_v$ by definition we have 
  \begin{eqnarray*}
    c(F_v) & = & c(F)\\
    \scf(F_v) & = & \scf(F),
  \end{eqnarray*}
  and if $F$ does not contain pure variables, then we have $\delta(F_v) = \delta(F)$, while in general we have $\delta(F_v) \ge \delta(F)$.
\end{itemize}
In order to keep $\ntrans(F)$ small (that is, to keep the number of literal occurrences small), one can order the domains $D_v$ according to descending number of occurrences (so that the bigger clause-enlargements occur less often), however in our general context such considerations are not of importance. We consider now a simple example, using two variables $a, b$ with $D_a = \set{0,1,2,3}$ and $D_b = \set{0,3,2,1}$ (using these orderings for the translation), while
\begin{eqnarray*}
  F & := & \setb{ \set{(a,0)},\: \set{(b,0)},\: \set{(a,0),(b,0)},\\
    && \quad \quad \quad \set{(a,1),(b,1)},\: \set{(a,2),(b,2)},\: \set{(a,3),(b,3)} }.
\end{eqnarray*}
We have $c(F) = 6$, $n(F) = 2$, $\rd(F) = 2 \cdot (4-1) = 6$ and $\delta(F) = 6 - 6 = 0$. The new variables are $v_0, v_1, v_2$ for $a$, and $w_0, w_1, w_2$ for $b$, and the translated clause-set is
\begin{multline*}
  \ntrans(F) = \setb{
    \set{v_0}, \set{w_0}, \set{v_0, w_0},\\
    \set{\ol{v_0}, v_1, \ol{w_0}, \ol{w_1}, \ol{w_2}},
    \set{\ol{v_0}, \ol{v_1}, v_2, \ol{w_0}, \ol{w_1}, w_2},
    \set{\ol{v_0}, \ol{v_1}, \ol{v_2}, \ol{w_0}, w_1}
  }.
\end{multline*}
We have $c(\ntrans(F)) = 6$, $n(\ntrans(F)) = 6$ and $\delta(\ntrans(F)) = 0$. Note that $\delta^*(F) = 0$ (since $F$ is matching satisfiable), but $\delta^*(\ntrans(F)) = 1$, while the matching-lean kernel of $\ntrans(F)$ is $\set{\set{v_0}, \set{w_0}, \set{v_0, w_0}}$.

The idea for this translation is taken from the reduction of multiclique partitions to biclique partitions as observed in \cite{GregoryMeulen1996RPartiteDecompositions} --- another example, where a simple graph-theoretical observation extends to a fuller statement of (generalised) logic. Regarding the properties of the nested translation, we focus on the most basic properties here. We already noticed that we always have $c(\ntrans(F)) = c(F)$, while for the number of formal variables (which might not occur) we have $n(\ntrans(F)) = \rd(F)$. By definition we get:
\begin{lem}\label{lem:conntransalg}
  For multi-clause-sets $F_1, F_2$ we have
  \begin{displaymath}
    \ntrans(F_1 + F_2) = \ntrans(F_1) + \ntrans(F_2),
  \end{displaymath}
  given the same choice of translation variables and the same nesting order. Thus for a multi-clause-set $F$ and $F' \le F$ we have $\ntrans(F') \le \ntrans(F)$.
\end{lem}

\begin{lem}\label{lem:EiggTrans}
  For a generalised multi-clause-set $F$ we have:
  \begin{enumerate}
  \item\label{lem:EiggTrans1} $\scf(\ntrans(F)) = \scf(F)$.
  \item\label{lem:EiggTrans2} $\delta(\ntrans(F)) \ge \delta(F)$; if $F$ has no pure variables, then $\delta(\ntrans(F)) = \delta(F)$.
  \item\label{lem:EiggTrans3} $\ntrans(F)$ is satisfiable if and only if $F$ is satisfiable.
  \item\label{lem:EiggTrans4}  $\ntrans(F)$ is minimally unsatisfiable if and only if $F$ is minimally unsatisfiable.
  \end{enumerate}
\end{lem}
\begin{proof} Properties \ref{lem:EiggTrans1}, \ref{lem:EiggTrans2} follow from the remarks above. For Property \ref{lem:EiggTrans3} first assume that $F$ is satisfiable with satisfying total assignment $\vp$. Consider $v \in \var(F)$ with $\vp(v) = \ve_i$ and the boolean variables $v_1, \dots, v_{k-1}$ as used above in the definition of $\ntrans(F)$. Now for $j \in \tb 1{k-1}$ assign $v_j \ra 1$ for $j \not= i$, while for $j = i$ let $v_j \ra 0$. The total assignment $\vp'$ obtained in this way satisfies $\ntrans(F)$. For the other direction consider a satisfying total assignment $\vp$ for $\ntrans(F)$, and for $v \in \var(F)$ let $\vp'(v) := \ve_i$ for the smallest $i \in \tb 1k$ such that $i < k \Ra \vp(v_i) = 0$ holds; again by definition we see that $\vp'$ satisfies $F$. Finally Property \ref{lem:EiggTrans4} follows with Property \ref{lem:EiggTrans3} and Lemma \ref{lem:conntransalg}. \end{proof}

\subsection{The hermitian defect, and regular hitting clause-sets}
\label{sec:Regularhitting}

The \emph{hermitian rank} $h(M)$ of a symmetric real matrix of order $m$ (which has only real eigenvalues) is the maximum of the number of positive and the number of negative eigenvalues of $M$, while the \textbf{hermitian defect} of $M$ is $\hdef(M) := m - h(F)$, which is the Witt index of the associated bilinear form; for more on this and its relation to (boolean) satisfiability problems see \cite{Ku2003e,GalesiKullmann2003bHermitian}. We use $h(F) := h(\scf(F))$ and $\hdef(F) := \hdef(\scf(F))$ for (generalised) clause-sets $F$. In \cite{Ku2003e} it was shown that boolean clause-sets $F$ fulfil $\delta(F) \le \hdef(F)$, which is just a translation of the well-known basic Graham-Pollak theorem about biclique partitions of multigraphs. Since for regular non-empty hitting clause-sets $F$ we have $\hdef(F) = 1$, it follows $\delta(F) \le 1$ for (boolean) regular hitting clause-sets. 

By Lemma \ref{lem:EiggTrans}, Part \ref{lem:EiggTrans1} we have that 
\begin{equation}
  \label{eq:ntransh}
  h(\ntrans(F)) = h(F)
\end{equation}
for multi-clause-sets $F$. Thus by Theorem 13 in \cite{GalesiKullmann2003bHermitian} we get
\begin{lem}\label{lem:polyh1}
  SAT decision for generalised clause-sets $F$ with $h(F) \le 1$ can be done in polynomial time (and for example by self-reduction also satisfying assignments can be found).
\end{lem}
Regarding the main subject of this section, from \eqref{eq:ntransh} we obtain $\delta(F) \le \delta(\ntrans(F)) \le \hdef(\ntrans(F)) = \hdef(F)$:
\begin{thm}\label{thm:GPvKlm}
  For a generalised multi-clause-set $F$ we have $\delta(F) \le \hdef(F)$.
\end{thm}
The proof of Theorem \ref{thm:GPvKlm} works by using the nested translation, and it is not clear how to generalise the notion of hermitian rank to obtain a direct proof.
\begin{corol}\label{cor:GPvKlm}
  For a generalised clause-set $F$ which is regular hitting we have $\delta(F) \le 1$.
\end{corol}

In the terminology of graph partitions, Corollary \ref{cor:GPvKlm} generalises ``Witsenhausen's Theorem'', the special case of the Graham-Pollak Theorem asserting that every biclique partition of a complete graph $K_m$ needs at least $m-1$ bicliques: 
\begin{itemize}
\item A \emph{multiclique} in a graph is a subgraph (not necessarily induced) which is a complete $k$-partite graph for some $k$.
\item A \emph{multiclique partition} of a multigraph consists of edge-disjoint multicliques covering every edge.
\item Generalised clause-sets are multiclique partitions of their conflict multigraphs, where the multicliques correspond to variables, with the variable-values and their occurrences corresponding to the parts of the multiclique.
\item Thus in Corollary \ref{cor:GPvKlm} we allow to partition the edge set of $r \cdot K_m$ into complete multipartite graphs, where every complete $k$-partite component ($k$ not fixed) contributes the ``cost''  $k-1$, and Corollary \ref{cor:GPvKlm} now says that the total cost must be at least $m - 1$.
\item Allowing only constant $k = 2$ is the Theorem of Witsenhausen, allowing only constant $k = m$ is trivial, and for arbitrary constant $k$ compare Example 1.1 in \cite{GregoryMeulen1996RPartiteDecompositions}.
\end{itemize}
Let us consider an example here, which also shows that multiple edges can help to achieve a higher deficiency (i.e. using less variables). Let the multigraph $G$ be
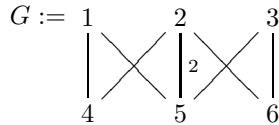
\begin{displaymath}
  G := \xymatrix {
    1 \aru[d] \aru[dr] & 2 \aru[dl] \aru[d]^2 \aru[dr] & 3 \aru[dl] \aru[d]\\
    4 & 5 & 6
  }
\end{displaymath}
where the edge between vertices $2$ and $5$ has multiplicity $2$. An optimal biclique partition of the edge sets is given by the two bicliques $(\set{1,2}, \set{4,5})$ and $(\set{2,3}, \set{5,6})$; a corresponding (``exact'', boolean) clause-set $F$ with deficiency $6 - 2 = 4$ has clauses $C_1, \dots, C_6$ and variables $v_1, v_2$, where variable $v_1$ occurs positively in $C_1, C_2$ and negatively in $C_4, C_5$, while variable $v_2$ occurs positively in $C_2, C_3$ and negatively in $C_5, C_6$. The underlying graph has adjacency matrix
\begin{displaymath}
  \begin{pmatrix}
    0 & 0 & 0 & 1 & 1 & 0\\
    0 & 0 & 0 & 1 & 1 & 1\\
    0 & 0 & 0 & 0 & 1 & 1\\
    1 & 1 & 0 & 0 & 0 & 0\\
    1 & 1 & 1 & 0 & 0 & 0\\
    0 & 1 & 1 & 0 & 0 & 0
  \end{pmatrix}
\end{displaymath}
with hermitian rank $3$ (it has eigenvalues $\pm 1$ each with multiplicity $3$), and thus by Theorem \ref{thm:GPvKlm} every multi-clause-set with this conflict matrix has a deficiency of at most $6 - 3 = 3$, i.e., $3$ variables are needed (which obviously also suffices here (so this graph is ``eigensharp'')).

Unsatisfiable hitting clause-sets are minimally unsatisfiable, and thus have deficiency at least $1$; it follows:
\begin{corol}\label{cor:Unsatreghit}
  Unsatisfiable regular hitting clause-sets have deficiency exactly $1$, i.e., $\Sclash^{\sat=0} = \Sclashk{\delta=1}^{\sat=0}$.
\end{corol}

Generalising Corollary 34 in \cite{Ku2003e}, we can now precisely characterise all unsatisfiable regular hitting (generalised) clause-sets. An unsatisfiable regular hitting clause-set with at least two clauses can not be $k$-regular for $k \not= 1$, since otherwise it would be satisfiable due to the completeness of resolution. Thus by Corollary \ref{cor:CharakSMUSATd1} we get:

\begin{corol}\label{cor:CharakUnsatreghit}
  $\Sclash^{\sat=0} = \Smusati{\delta=1}$, that is, unsatisfiable regular hitting clause-sets are exactly the saturated minimally unsatisfiable clause-sets of deficiency $1$ as characterised in Corollary \ref{cor:CharakSMUSATd1}.
\end{corol}

An interesting aspect of unsatisfiable boolean regular hitting clause-sets has been revealed in \cite{SloanSzoerenyiTuran2005Primimplikanten_1}, where the basic setting translated to our language and generalised to generalised clause-sets reads as follows:
\begin{enumerate}
\item An unsatisfiable hitting clause-set $F$ corresponds to a partitioning of the space of all total assignments $\prod_{v \in \var(F)} D_v$ into disjoint ``cubes'', where a cube is a subset of $\prod_{v \in \var(F)} D_v$ where some variables have a fixed value while the values of all other variables are arbitrary.
\item $1$-regularity of $F$ means that any two different cubes in the corresponding cube-partition have distance (exactly) $1$, where the metric on $\prod_{v \in \var(F)} D_v$ is the Hamming distance (counting the number of different positions), and the distance between two subsets is the minimal distance of their members.
\item Thus $1$-regular hitting clause-sets $F$ correspond to ``neighbourly cube partitions'' of $\prod_{v \in \var(F)} D_v$.
\end{enumerate}
Now \cite{SloanSzoerenyiTuran2005Primimplikanten_1} gave first a ``combinatorial'' direct proof of Corollary \ref{cor:CharakUnsatreghit} for the boolean case (that is, of Corollary 34 in \cite{Ku2003e}), and then showed that by adding literals with new variables to the leaves of the tree representations of the elements of $\Smusati{\delta=1}(\Bva)$ one obtains exactly those boolean clause-sets $F$ which have exactly the maximal number of prime implicates possible (for a given number of clauses), namely $2^{c(F)} - 1$ (see Corollary \ref{cor:MaxPrim}). These considerations are generalised in \cite{HendersonKullmann2007Multicliquen}; since (generalised) clause-sets have a too restricted language here, in order to obtain the maximal number of prime implicates the full power of \emph{signed} clause-sets is needed (that is, the appropriate notion of prime implicate here actually must allow signed literals).

\section{Conclusion and open problems}
\label{sec:open2}

The main purpose of this series of two articles was to set the stage for the study of generalised clause-sets as sets of ``no-goods'', where literals are given by one ``forbidden value'':
\begin{itemize}
\item We defined and summarised the basic properties of syntax, semantics, resolution calculus and autarky systems.
\item Then we considered the generalisation of the notion of deficiency for these generalised clause-sets, and we studied the basic autarky system related to this notion, namely matching autarkies.
\item We showed fixed parameter tractability of satisfiability of generalised clause-sets in the maximal deficiency, while the computation of a maximal autarky (yielding the lean kernel) can be done in polynomial time for fixed maximal deficiency.
\item For autarky systems both the application of autarkies as reductions and the properties of autarky-freeness, i.e., lean clause-sets, are of interest.
\item Lean clause-sets are a generalisation of minimally unsatisfiable clause-sets, for which we considered the basic problem, when the property of being minimally unsatisfiable is preserved under application of partial assignments, and we characterised also minimally unsatisfiable clause-sets of minimal deficiency.
\item More generally, also irredundant clause-sets have been considered.
\item Besides using the generalised tools transferred from the boolean case, also the structure preserving properties of boolean translations are important, and we investigated basic cases.
\item The notion of deficiency introduced here for generalised clause-sets is fundamental for many of these considerations, and we considered also the (close) relation to multiclique partitions of multigraphs and related areas.
\end{itemize}

\subsection{Minimally unsatisfiable clause-sets of low deficiency}
\label{sec:openMUSAT}

Having generalised the characterisation of minimally unsatisfiable clause-sets of deficiency one from the boolean case in Subsection \ref{sec:MUSATdefone}, the next question concerns the generalisation of the structure of boolean $\Musati{\delta=2}$ as studied in \cite{KleineBuening2000SubclassesMU}. This generalisation seems to be not straightforward, but we believe that minimally unsatisfiable generalised clause-sets of deficiency two are still quite close to the boolean case (while from deficiency three on generalised clause-sets behave more wildly).

A key tool for the study of boolean minimally unsatisfiable clause-sets is the observation in \cite{Ku99dKo} that for every boolean minimally unsatisfiable clause-set $F$ with $n(F) > 0$ there exists a variable $v \in \var(F)$ such that for both $\ve \in \set{0,1}$ we have $\#_{(v,\ve)}(F) \le \delta(F)$; see Lemma \ref{lem:mmvdobSchr} for a discussion of this subject. As applied in the proof of Lemma \ref{lem:GrundtatsacheDefektEins}, the direct boolean translation is a suitable tool here; additionally the nested translation from Subsection \ref{sec:newtrans} could be useful.

\subsection{Conflict combinatorics}
\label{sec:openmultihitting}

As touched upon in Lemma \ref{lem:polyzeitbitreff}, the SAT problem for bihitting clause-sets is essentially the same as the hypergraph transversal problem, and whether the latter problem can be decided in polynomial time is a long outstanding open question. Being optimistic about the potential of (generalised) clause-sets to provide a unifying framework for (hard) graph and hypergraph problems, we propose:
\begin{conj}\label{con:multihittingpoly}
  Satisfiability decision for multihitting (generalised) clause-sets can be done in polynomial time.
\end{conj}

Since the nested translation (recall Subsection \ref{sec:newtrans}) maintains the hermitian rank (as well as the hermitian defect), Conjecture 15 from \cite{GalesiKullmann2003bHermitian} is equivalent to
\begin{conj}\label{con:hermpoly}
  Satisfiability decision for generalised clause-sets can be done in polynomial time for bounded hermitian rank.
\end{conj}
See also Section 5.2 in \cite{GalesiKullmann2003bHermitian} for more specialised conjectures regarding the structure of boolean clause-sets of hermitian rank $1$ (which can be generalised to generalised clause-sets).

One of the early problems in the area of addressing graphs, the ``squashed cube conjecture'' solved by Peter Winkler (see Chapter 9 in \cite{LW1992b}), is generalised by the following conjecture (where a clause-set $F$ is called \emph{exact}, if its deficiency is maximal amongst all clause-sets with the same conflict matrix (or conflict multigraph)):
\begin{conj}\label{con:Winkler}
  Consider an exact clause-set $F$, whose conflict matrix is the distance matrix of some connected graph. Then $F$ is matching lean.
\end{conj}
What Winkler originally proved is equivalent to the statement that $\delta(F) \ge 1$ holds under these assumptions (for non-empty $F$).

The \emph{Alon-Saks-Seymour-Conjecture} (``Every graph $G$ which can be written as a union of $m$ edge-disjoint complete bipartite graphs has $\chi(G) \le m+1$.``) can be reformulated in our context as follows, where for consistency with Theorem \ref{thm:GPvKlm} we introduce the ``colouring deficiency'' $\delta_{\chi}(F) := c(F) - \chi(\cg(F))$ of a clause-set $F$ (where $\chi(G)$ is the chromatic number of graph $G$, the minimal number of colours needed for a proper vertex colouring):
\begin{conj}\label{con:AlonSaks}
  For every clause-set $F$ with at most one conflict between any two clauses we have $\delta(F) \le \delta_{\chi}(F) + 1$.
\end{conj}
Since regular hitting $F$ have $\delta_{\chi}(F) = 0$, this would yield a generalisation of Corollary \ref{cor:GPvKlm} (for $1$-regular hitting clause-sets) in a different direction. The restriction for $F$ in Conjecture \ref{con:AlonSaks}, which just states that the conflict multigraph of $F$ in fact is a graph, seems necessary, since colouring of multigraphs only considers the underlying graph, while allowing parallel edges can increase the deficiency.

We remark that w.r.t.\ the computation of the biclique partition number $\bcp(G)$ of a multigraph $G$, which is the minimal $n(F)$ for boolean clause-sets $F$ with $\cmg(F) \cong G$, there are basic open questions:
\begin{itemize}
\item It is known that the decision ``$\bcp(G) \le k$'' is NP-complete, where $k$ is part of the input.
\item This has been shown in \cite{KRW1988EigensharpGraphs}: for graphs $G$ without $4$-cycles $\bcp(G)$ is the vertex-cover-number (the special case of the transversal number for hypergraphs, applied to graphs --- the smallest number of vertices hitting every edge), since the only bicliques in such graphs are stars, while for such graphs the vertex-cover-number problem (or, equivalently, the independence-number problem) is NP-complete.
\item Now for fixed $k$ we should have polynomial time decision of ``$\bcp(G) \le k$ ?'', while to achieve fixed parameter tractability (in $k$) seems more challenging.
\end{itemize}
Instead of computing the minimal $n(F)$ for boolean $F$ with $\cmg(F) \cong G$ (then we have $n(F) = \bcp(G)$) we can also compute (equivalently) the maximal $\delta(F)$ for such $F$, a point of view which seems to have some advantages, and where by Lemma \ref{lem:EiggTrans} we can also allow all generalised clause-sets here (since the weights for the variables counterbalance that variables with larger domains are more powerful). However considerations as in \cite{GregoryMeulen1996RPartiteDecompositions} ask for the minimal $n(F)$ for generalised clause-sets $F$ with \emph{uniform} domain size $d$ (a further parameter) such that $\cmg(F) \cong G$, and then we need to consider only generalised clause-sets where all variables have domain size $d$.

\subsection{Translations to boolean clause-sets}
\label{sec:opentranslations}

In Subsection \ref{sec:newtrans} we have introduced the \emph{nested translation} $F \mapsto \ntrans(F)$ from (multi-)clause-sets to boolean (multi-)clause-sets. The idea can be generalised as follows, covering then also the direct translation (see Section \ref{sec:translating}), in the ``weak form'' we have used (considered only ALO-clauses) as well as in the ``strong'' form (using also AMO-clauses). We call this scheme the \emph{generic translation}.

To every variable $v \in \var(F)$ one needs to associate two boolean clause-sets $T(v), T'(v)$ together with a bijection $\gamma_v: D_v \ra T(v)$, such that the following conditions are fulfilled:
\begin{itemize}
\item all clause-sets $T(v) \cup T'(v)$ are unsatisfiable;
\item all clauses of $T(v)$ are necessary in $T(v) \cup T'(v)$ (that is, removing any clause of $T(v)$ renders $T(v) \cup T'(v)$ satisfiable);
\item for different variables $v, w$ we have $\var(T(v) \cup T'(v)) \cap \var(T(w) \cup T'(w)) = \es$.
\end{itemize}
Now the translated clause-set $\Theta_{T,T',\gamma}(F)$ is obtained as follows from $F$:
\begin{enumerate}
\item every clause $C \in F$ is replaced by the clause $\bc_{(v,\ve) \in C} \gamma_v(\ve)$;
\item the clauses from the clause-set $\bc_{v \in \var(F)} T'(v)$ are added.
\end{enumerate}

\begin{lem}\label{lem:gentrans}
  $\Theta_{T,T',\gamma}(F)$ is satisfiability-equivalent to $F$.
\end{lem}
\begin{proof}
  First consider a total satisfying assignment $\vp$ for $F$. For $v \in \var(F)$ choose a total satisfying assignment $\psi_v$ for $(T(v) \cup T'(v)) \sm \set{\gamma_v(\vp(v))}$; now $\bc_{v \in \var(F)} \psi_v$ is a satisfying assignment for $\Theta(F)$. Now consider a satisfying assignment $\psi$ for $\Theta(F)$. For each $v \in \var(F)$ choose one clause $C_v \in F(v)$ falsified by $\psi$, and let $\vp(v) := \gamma_v^{-1}(C_v)$; now $\vp$ is a satisfying assignment for $F$.
\end{proof}

The general scheme for handling a variable $v$ with $k$ values can be understood as choosing some minimally unsatisfiable boolean clause-set $T_0$ with at least $k$ clauses, choosing $k$ clauses from $T_0$ for $T(v)$, and putting the remaining clauses into $T'(v)$, where $T'(v)$ can be augmented with further clauses such that the clauses of $T_0$ remain necessary. Four main classes of examples are as follows:
\begin{itemize}
\item To obtain the direct translation, let $T(v) := \setb{ \set{\transl((v,\ve))} }_{\ve \in D(v)}$ be a set of unit-clauses, where $\gamma_v(\ve) := \set{\transl((v,\ve))}$, and let $T'(v) := \setb{ \set{\ol{\transl((v,\ve))} : \ve \in D(v)} }$ be a singleton clause-set (consisting of the ALO-clause for $v$). For the strong form just add the AMO-clauses to $T'(v)$. So the underlying minimally unsatisfiable boolean clause-set (with $k$ variables and $k+1$ clauses) is $\setb{\set{v_1}, \dots, \set{v_k}, \set{\ol{v_1}, \dots, \ol{v_k}}}$.
\item For the nested translation, let $T(v) := H_v$ and $T'(v) := \top$, while $\gamma_v$ is any bijection between $D_v$ and $H_v$ (using the notations from Subsection \ref{sec:newtrans}). To obtain a strong form, add all positive binary clauses to $T'(v)$. The underlying minimally unsatisfiable (saturated) boolean (Horn) clause-set (with $k-1$ variables and $k$ clauses) is $\setb{\set{v_1}, \dots, \set{\ol{v_1}, \dots, \ol{v_{k-2}},v_{k-1}}, \set{\ol{v_1}, \dots, \ol{v_{k-1}}}}$.
\item Using the same underlying minimally unsatisfiable clause-set as with the direct translation, however this time using all clauses for $T(v)$, leaving $T'(v)$ empty, and so using only $k-1$ variable, we obtain the \emph{reduced translation}, which can also be obtained from the nested translation by removing all negative literals from the first $k-1$ clauses. A strong form is obtained by adding all positive binary clauses to $T'(v)$.
\item Finally, to obtain an example of a ``logarithmic translation'', assume for simplicity that we have $k = 2^p$ for the domain size $k = \abs{D_v}$, and let $T(v)$ be the full clause-set (all clauses contain all variables) with $k$ clauses over $p$ variables, while $T'(v) := \top$.\footnote{The simplest form of handling arbitrary $k$ is to consider the smallest $p$ with $2^p \ge k$, to choose $k$ full clauses (over $p$ variables) for $T(v)$, and to put the remaining full clauses into $T'(v)$.}
\end{itemize}
Please note that for the generic translation every variable can be treated differently. A systematic study of the generic translation scheme, which admits now the possibility to adopt the translation to the problem at hand, should yield powerful tools for solving SAT for generalised clause-sets by boolean SAT solvers. For a first empirical study see \cite{Kullmann2010GreenTao}.

\bibliographystyle{plain}

\end{document}